\documentclass[10pt,journal,compsoc]{IEEEtran}
\usepackage{witharrows}
\usepackage{hyperref}       
\usepackage{url}            
\usepackage{booktabs}       
\usepackage{framed}
\usepackage{wrapfig}
\usepackage{url}
\usepackage{graphicx}
\usepackage{subfigure}
\usepackage{booktabs} 
\usepackage[shortlabels]{enumitem}
\setlist[enumerate]{nosep}
\newtheorem{theorem}{Theorem}

\newtheorem{definition}{Definition} 
\newtheorem{proposition}{Proposition} 
\newtheorem{lemma}{Lemma}
\newtheorem{proof}{Proof}
\newtheorem{corollary}{Corollary}
\usepackage{amsfonts,amssymb}
\usepackage{multirow}

\usepackage{bm}
\usepackage{algorithm}
\usepackage{algorithmic}
\newcommand*{\dif}{\mathop{}\!\mathrm{d}}

\usepackage{tikz}
\usetikzlibrary{matrix}
\usepackage{enumitem}
\usepackage{xr}

\newtheorem{remark}{Remark}
\usepackage{amsmath}
\usepackage[cmintegrals]{newtxmath}

\ifCLASSOPTIONcompsoc
  \usepackage[nocompress]{cite}
\else
  \usepackage{cite}
\fi
\ifCLASSINFOpdf
\else
\fi
\begin{document}
%
\title{On the Properties of Kullback-Leibler Divergence Between Multivariate Gaussian Distributions}

\author{Yufeng~Zhang, Wanwei~Liu, Zhenbang~Chen, Ji~Wang, Kenli~Li
	\IEEEcompsocitemizethanks{
		\IEEEcompsocthanksitem Yufeng Zhang and Kenli Li are with the College of Computer Science and Electronic Engineering, Hunan University, Changsha, China.\protect\\
		E-mail: yufengzhang@hnu.edu.cn, lkl@hnu.edu.cn
		\IEEEcompsocthanksitem Wanwei Liu, Zhenbang Chen and Ji Wang are with the College of Computer, National University of Defense Technology, Changsha, China.
		Ji Wang is also with State Key
Laboratory of High Performance
Computing, National University of
Defense Technology.\protect\\
		E-mail: wwliu@nudt.edu.cn, zbchen@nudt.edu.cn, wj@nudt.edu.cn
		
		\IEEEcompsocthanksitem Kenli Li and Ji Wang are the corresponding authors.
	}
	\thanks{Manuscript received xx xx, 2022; revised xx xx, 2022.}
}

%
%

\markboth{Journal of \LaTeX\ Class Files,~Vol.~14, No.~8, August~2015}%
{Shell \MakeLowercase{\textit{et al.}}: Bare Demo of IEEEtran.cls for Computer Society Journals}
%



\IEEEtitleabstractindextext{%
\begin{abstract}
Kullback-Leibler (KL) divergence is one of the most important divergence measures between probability distributions. 
In this paper, we prove several properties of KL divergence between multivariate Gaussian distributions. First, for any two $n$-dimensional Gaussian distributions $\mathcal{N}_1$ and $\mathcal{N}_2$, we give the supremum of $KL(\mathcal{N}_1||\mathcal{N}_2)$ when $KL(\mathcal{N}_2||\mathcal{N}_1)\leq \varepsilon\ (\varepsilon>0)$. For small $\varepsilon$, we show that the supremum is $\varepsilon + 2\varepsilon^{1.5} + O(\varepsilon^2)$.  This quantifies the approximate symmetry of small KL divergence between Gaussians. We also find the infimum of $KL(\mathcal{N}_1||\mathcal{N}_2)$ when $KL(\mathcal{N}_2||\mathcal{N}_1)\geq M\ (M>0)$. We give the conditions when the supremum and infimum can be attained. Second, for any three $n$-dimensional Gaussians $\mathcal{N}_1$, $\mathcal{N}_2$, and $\mathcal{N}_3$, we find an upper bound of $KL(\mathcal{N}_1||\mathcal{N}_3)$ if $KL(\mathcal{N}_1||\mathcal{N}_2)\leq \varepsilon_1$ and $KL(\mathcal{N}_2||\mathcal{N}_3)\leq \varepsilon_2$ for $\varepsilon_1,\varepsilon_2\ge 0$. For small $\varepsilon_1$ and $\varepsilon_2$, we show the upper bound is $3\varepsilon_1+3\varepsilon_2+2\sqrt{\varepsilon_1\varepsilon_2}+o(\varepsilon_1)+o(\varepsilon_2)$. This reveals that KL divergence between Gaussians follows a relaxed triangle inequality. Importantly, all the bounds in the theorems presented in this paper are independent of the dimension $n$. Finally, We discuss the applications of our theorems in explaining counterintuitive phenomenon of flow-based model, deriving deep anomaly detection algorithm, and extending one-step robustness guarantee to multiple steps in safe reinforcement learning. 
\end{abstract}

}

\maketitle

\IEEEdisplaynontitleabstractindextext

%
\IEEEpeerreviewmaketitle

\IEEEraisesectionheading{\section{Introduction}\label{sec:introduction}}

%
%
%
%


\IEEEPARstart{A}{} statistical divergence measures the ``distance'' between probability distributions. Let $X$ be a space of probability distributions with the same support. A statistical divergence  $D:X\times X\rightarrow \mathbb{R}^{+}$ ($\mathbb{R}^{+}$ is the set of non-negative real numbers) should satisfy (a) non-negativity: $D(p,q)\geq 0$ and (b) identity of indiscernibles: $D(p,p)=0$, where $p,q$ are probability densities. Another stronger concept,  statistical distance, also measures the distance between probability distributions. A statistical distance should satisfy two extra properties including  (c) symmetry: $D(p,q)=D(q,p)$ and (d) triangle inequality: $D(p, q) \leq D(p, g) + D(g, q)$, where $p,q$ and $g$ are probability densities.

Kullback-Leibler (KL) divergence, also referred to as relative entropy \cite{kullback1997information}, plays a key role in many fields including machine learning \cite{PRML,goodfellow2016deep}, information theory \cite{cover2012elements}, and statistics \cite{pardo2018statistical}, \textit{etc}. 
The KL divergence between two continuous probability densities $p(x)$ and $q(x)$ is defined as
\begin{align}
	KL(p(x)||q(x))=\int p(x)\log \dfrac{p(x)}{q(x)}\dif x
\end{align}

KL divergence is not a proper distance \cite{kullback1997information}. Firstly, KL divergence is not symmetric. It might happen that the forward KL divergence\footnote{ Here we can choose to call $KL(p||q)$ or $KL(q||p)$ as forward KL divergence. The  terminologies ``forward'' and ``reverse'' is just for convenience.} $KL(p||q)$ is very small but the reverse KL divergence  $KL^*(p||q)=KL(q||p)$ is very large.  Secondly, KL divergence dose not satisfy the triangle inequality. This brings obstacles in applying KL divergence in many circumstances.

KL divergence is one member of more generalized divergence families including $f$-divergence (also called $\phi$-divergence) \cite{ali1966general}, Bregman divergence \cite{bregman1967relaxation},  and R\'{e}nyi divergence \cite{renyi1961measures}. 
For example, the widely used $f$-divergence includes many commonly used measures including KL divergence, Jensen-Shannon divergence, and squared Hellinger distance \cite{pardo2018statistical}. Many $f$-divergence are not proper distance metrics. 
KL divergence also has a deep connection with other information measures. For example, the second derivative of KL divergence is Fisher information metric. By taking the second-order Taylor expansion, KL divergence between two nearly distributions can be approximated with fisher information matrix \cite{kullback1997information}. Furthermore, forward and reverse KL divergence have the same second derivatives at the point where two distributions are equal. Therefore, KL divergence is locally approximately symmetric when two distributions are close to each other.

Meanwhile, Gaussian distribution is one of the most important distributions and central to statistics. It is also pervasive in many fields including machine learning and information theory. 
The probability density function of an $n$-dimensional Gaussian distribution is given by 
\begin{equation}
	\begin{aligned}
		\mathcal{N}(\bm{\mu},\bm{\Sigma})=\dfrac{1}{(2\pi)^{n/2}\vert \bm{\Sigma}\vert^{1/2}}\exp \left(-\dfrac{1}{2}(\bm{x}-\bm{\mu})^\top\bm{\Sigma}^{-1}(\bm{x}-\bm{\mu})\right)
	\end{aligned}
\end{equation}
Here $\bm{\mu}\in \mathbb{R}^n$ is the mean and $\bm{\Sigma}\in \bm{S}^n_{++}$ is the covariance matrix, where $\bm{S}^n_{++}$ is the space of symmetric positive definite $n\times n$ matrices.
Gaussian distribution constitutes the basis for more complicated distributions. For example, the mixture of Gaussians, namely Gaussian Mixture Model (GMM) has a wide range of applications due to its power of approximation \cite{PRML}.

The KL divergence between two $n$-dimensional Gaussians $\mathcal{N}_1(\bm{\mu}_1,\bm{\Sigma}_1)$ and $\mathcal{N}_2(\bm{\mu}_2,\bm{\Sigma}_2)$  has the following closed form \cite{pardo2018statistical}
\begin{equation}
	\small
	\begin{aligned}\label{equ:KL_two_Gaussian}
		&KL(\mathcal{N}_1(\bm{\mu}_1,\bm{\Sigma}_1)||\mathcal{N}_2(\bm{\mu}_2,\bm{\Sigma}_2)) \\
		=&\dfrac{1}{2}\left(\log \dfrac{|\bm{\bm{\Sigma}_2}|}{|\bm{\Sigma}_1|}+\mathop{\mathrm{Tr}}(\bm{\bm{\Sigma}_2}^{-1}\bm{\Sigma}_1)+(\bm{\mu}_2-\bm{\mu}_1)^\top\bm{\Sigma}_2^{-1}(\bm{\mu}_2-\bm{\mu}_1)-n\right)
	\end{aligned}
\end{equation}
where the logarithm is taken to base $e$ and $\mathop{\mathrm{Tr}}$ is the trace of matrix. 
Due to the good form of Gaussians, one may expect that KL divergence between Gaussians may have some good properties. 
However, as like many other distributions, KL divergence between Gaussians is not symmetric and does not satisfy the triangle inequality either.

The concept of KL divergence has been proposed for seventy years \cite{kullback1997information}. 
It is surprising that the properties of KL divergence between Gaussians have not been investigated thoroughly.

In this paper, we investigate the following research problems.

\begin{enumerate}
	\item 
	For any two Gaussians $\mathcal{N}_1(\bm{\mu}_1,\bm{\Sigma}_1)$ and $\mathcal{N}_2(\bm{\mu}_2,\bm{\Sigma}_2)$, when forward KL divergence $KL(\mathcal{N}_1||\mathcal{N}_2)$ is bounded by a small number $\varepsilon$, what is the supremum of reverse KL divergence $KL(\mathcal{N}_2||\mathcal{N}_1)$? Although KL divergence is locally approximately symmetric, we want to step further in the investigation on such approximate symmetry in a Gaussian case. 
	
	\item For any two Gaussians $\mathcal{N}_1(\bm{\mu}_1,\bm{\Sigma}_1)$ and $\mathcal{N}_2(\bm{\mu}_2,\bm{\Sigma}_2)$, when forward KL divergence $KL(\mathcal{N}_1||\mathcal{N}_2)$ is not smaller than a number $M$, what is the infimum of reverse KL divergence $KL(\mathcal{N}_2||\mathcal{N}_1)$? This problem is dual to the first problem. 
	
	\item Does the KL divergence between Gaussians follow some property similar to the triangle inequality? Precisely, for any three Gaussians $\mathcal{N}_i(\bm{\mu}_i,\bm{\Sigma}_i)\ (i\in {1,2,3})$, when $KL(\mathcal{N}_1||\mathcal{N}_2)$ and $KL(\mathcal{N}_2||\mathcal{N}_3)$ are bounded by small numbers $\varepsilon_1, \varepsilon_2$, respectively, how large $KL(\mathcal{N}_1||\mathcal{N}_3)$ can be?
\end{enumerate}

Note that the third research problem is different from the several existing general Pythagoras theorems satisfied by KL divergence \cite{cover2012elements, cluster-bregman-div, functional_bregman_divergence, Renyi-and-KL2014}. In the existing general Pythagoras theorems, the bounds are dependent on the given distributions. In this paper, we want to obtain bounds that are independent of the parameters of Gaussians and hold for any Gaussian distributions. We will discuss this point in Section \ref{sec:quasi_triangle}.


As far as we know, we are the first to propose and investigate the above research problems.
These research problems are motivated by our research on deep anomaly detection with flow-based model \cite{dinh2014nice,dinh2016realnvp,kingma2018glow}. Flow-based model is capable of provide explicit likelihood for input.
Researchers have found that flow-based model may assign higher likelihoods for out-of-distribution (OOD) data than in-distribution (ID) data. For example, Glow \cite{kingma2018glow} trained on CIFAR-10 assigns higher likelihoods for SVHN. However, we can not sample OOD data from the model with prior although they are assigned higher likelihoods by the model. This counterintuitive phenomenon has not been explained satisfactorily. In our research, we find that the KL divergence between several Gaussian(-like) distributions are vital in explaining the behavior of flow-based model. This inspires us to conduct research on the properties of KL divergence between Gaussians. 
In this context, the parameters of distributions are learned from data or dependent on the given inputs. It is impossible to identify the parameters or the KL divergence before the model is trained. We only know that some bound is guaranteed. Therefore, the existing general Pythagoras theorems are not applicable due to their dependence on the parameters of distributions.
The theorems proved in this paper provide a solid foundation for explaining above phenomenon as well as our anomaly detection method using flow-based model. Our theorems can also be used as general conclusions in related fields including machine learning, information theory and statistics. For example, our theorems have been used as key support in safe reinforcement learning framework \cite{liu2022constrainedSafeRL} after we post our last version of this manuscript on Arxiv \cite{onThePropertiesOfKL2021Yufeng}. We will elaborate these applications in Section \ref{sec:application}.

\textbf{Contributions.} The contributions of this paper are as follows. 

Given any three $n$-dimensional Gaussians $\mathcal{N}_1, \mathcal{N}_2$ and $\mathcal{N}_3$, 
\begin{enumerate}
	\item We prove that when $KL(\mathcal{N}_2||\mathcal{N}_1)\leq \varepsilon$ for $\varepsilon>0$ the supremum of $KL(\mathcal{N}_1||\mathcal{N}_2)$ is $\frac{1}{2}\left(
	(-W_{0}(-e^{-(1+2\varepsilon)}))^{-1}+\log(-W_{0}(-e^{-(1+2\varepsilon)})) -1 \right)$,   where $W_0$ is the principal branch of Lambert $W$ function. 
	We give the condition when the supremum can be attained. For small $\varepsilon$, we show the supremum is $\varepsilon + 2\varepsilon^{1.5} + O(\varepsilon^2)$. 
	This quantifies the approximate symmetry of small KL divergence between Gaussians. 
	
	\item We find the infimum of $KL(\mathcal{N}_1||\mathcal{N}_2)$ if $KL(\mathcal{N}_2||\mathcal{N}_1)\geq M$ for $M>0$. We give two proofs for this result. The first proof has the similar structure with that for the above supremum. The second proof is based on the proof for the supremum. We also give the condition when the infimum can be attained.
	
	\item We find an upper bound of $KL(\mathcal{N}_1||\mathcal{N}_3)$ if $KL(\mathcal{N}_1||\mathcal{N}_2)\leq \varepsilon_1$ and $KL(\mathcal{N}_2||\mathcal{N}_3)\leq \varepsilon_2$ for $\varepsilon_1,\varepsilon_2\ge 0$. For small $\varepsilon_1$ and $\varepsilon_2$, we show the upper bound is $3\varepsilon_1+3\varepsilon_2+2\sqrt{\varepsilon_1\varepsilon_2}+o(\varepsilon_1)+o(\varepsilon_2)$ . This indicates that KL divergence between Gaussians follows a relaxed triangle inequality.
	
	\item All the bounds in our theorems are independent of the dimension $n$. This  is a critical property especially in contexts where dimensionality has a fundamental impact. 
	
	\item We show several applications of the theorems proved in this paper including explaining counterintuitive phenomenon in flow-based model, deriving anomaly detection algorithm, and extending robustness guarantee in safe reinforcement learning.
\end{enumerate}

The remaining part of this paper is organized as follows. In Section \ref{sec:def_lemma_notation} we prepare lemmas that will be used in all theorems. In Section \ref{sec:quasi_symmetry_kl} we investigate the supremum (infimum) of reverse KL divergence between Gaussians when forward KL divergence is bounded.
In Section \ref{sec:quasi_triangle} we investigate the relaxed triangle inequality of KL divergence between Gaussians. 
In Section \ref{sec:application} we discuss the applications of the theorems proved in this paper.
In Section \ref{sec:related-work} we discuss related work.
Finally, we conclude in Section \ref{sec:conclusion}.

\section{Lemmas and Notations} \label{sec:def_lemma_notation}

%
%
%

Before presenting our results, we introduce the famous transcendental function, the Lambert $W$ function, which occurs almost everywhere in this paper.
\begin{definition}{\textbf{Lambert $W$ Function\cite{lamberWfunction, corless1996lambertw}}.}\label{def:lamber_w}
	The inverse function of function $y=xe^x$ is called Lambert $W$ function $y=W(x)$.  
\end{definition}

When $x\in \mathbb{R}$, $W$ is a multivalued function with two branches $W_0, W_{-1}$, where $W_0$ is the principal branch (also called branch 0) and  $W_{-1}$ is the branch $-1$. The derivative of $W$ is 
\begin{gather}\label{equ:deriv_W}
	W'(x)=\dfrac{1}{x+e^{W(x)}}=\dfrac{W(x)}{x(1+W(x))}\ (x\neq 0, -e^{-1})
\end{gather}

Function $f(x)=x-\log x$ lies in the core of our problems. 
In the following, we prove some useful lemmas related to $f(x)$. 
\begin{lemma}\label{thm:lemma_1_d_fx}
	Given function $f(x)=x-\log x$ $(x\in \mathbb{R}^{++})$ ($\mathbb{R}^{++}$ is the set of positive real numbers), the following propositions hold.
	\begin{enumerate}[label=(\alph*), ref=\ref{thm:lemma_1_d_fx}\alph*]
		\item \label{equ:f_convex_minimum}
		$f(x)$ is strictly convex and takes the minimum value $1$ at $x=1$. 
		\item $f(x)>f(1/x)\ \text{for}\ x>1$ and $f(x)<f(1/x)\ \text{for}\ 0<x<1$.
		\label{prp:f_x_invs_x_1_d}
		\item \label{thm:reverse_f_def}
		The inverse function of $f$ is $f^{-1}(x)=-W(-e^{-x})\ (x\geq 1)$.
		
		\item \label{thm:solution_f_x}
		The solutions of equation $x-\log x=1+t\ (t\geq 0)$ are $w_1(t)  =  -W_{0}(-e^{-(1+t)})\in (0,1]$ and $w_2(t) =  -W_{-1}(-e^{-(1+t)})\in [1,+\infty)$.
		It is easy to know $w_1(0)=w_2(0)=1$. We treat $w_1(t),w_2(t)$ as functions of $t$.
		\item \label{thm:deriv_w1_w2}
		The derivatives of $w_1(t)$ and $w_2(t)$ are
		\begin{align}
			w_1'(t)= \dfrac{-w_1(t)}{1-w_1(t)}=\dfrac{W_{0}(-e^{-(1+t)})}{W_{0}(-e^{-(1+t)})+1}\label{equ:deriv_w1}<0\\
			w_2'(t) = \dfrac{-w_2(t)}{1-w_2(t)}=\dfrac{W_{1}(-e^{-(1+t)})}{W_{1}(-e^{-(1+t)})+1}\label{equ:deriv_w2}>0
		\end{align}
		
		\item \label{thm:inequality_solution_fx}
		For $t>0$, $f(w_1(t))<f(\frac{1}{w_1(t)}),
		f(\frac{1}{w_2(t)})<f(w_2(t))$.
		
		\item \label{prp:sup_f_x_invers_1_d}
		If $f(x)\leq 1+t\ (t\geq 0)$, then $w_1(t)\leq x \leq w_2(t)$ and 
		\begin{align} 
			S(t)=\sup\limits_{t\geq 0\atop f(x)\leq 1+t
			} f\big(\dfrac{1}{x}\big)
			=f(\dfrac{1}{w_1(t)})\label{equ:sup_f_x_invs}
		\end{align}
		
		\item \label{prp:inf_f_x_invers_1_d}
		If $f(x)\geq  1+t\ (t\geq 0)$, then $0<x\leq w_1(t) \vee x\geq w_2(t)$
		and
		\begin{align}
			I(t)=\inf\limits_{t\geq 0 \atop f(x)\geq 1+t} f(\dfrac{1}{x})=f(\dfrac{1}{w_2(t)})	
		\end{align}

		
		\item \label{thm:deriv_w2_inv_w2_compare}
		For $t\geq 0$, $f'(w_2(t))\leq -f'(\frac{1}{w_2(t)})$.
		
		\item \label{thm:f_w1w1_w2w2_plugin}
		For $t_1, t_2\geq 0$,
		\begin{align}
			f(w_1(t_1)w_2(t_2))
			=t_1+t_2+2+w_1(t_1)w_1(t_2)-w_1(t_1)-w_1(t_2)\\
			f(w_2(t_1)w_2(t_2))
			=t_1+t_2+2+w_2(t_1)w_2(t_2)-w_2(t_1)-w_2(t_2)\label{equ:plugin_fw2w2}
		\end{align}

	\end{enumerate}
\end{lemma}
\begin{proof}
	The details of the proof are shown in Appendix \ref{sec:app_proof_thm:lemma_1_d_fx}.

		$\hfill\square$ 
\end{proof}


The notations used in this paper are summarized in Table \ref{tbl:notations}. 
\begin{table} [ht]
	\small
	\vspace{-0pt}
	\caption{Notations.} 
	\label{tbl:notations}  
	\begin{center}  
		\begin{tabular}{cc}  
			\toprule[1pt]
			$f(x)$  & $x-\log x\ (x\in \mathbb{R}^{++})$ \\
			$W(x)$ & the Lambert $W$ function\\
			$W_0(x)$ & the principal branch (branch 0) of $W(x)$\\
			$W_{-1}(x)$ & the branch $-1$ of $W(x)$\\
			$w_1(t)$ & the smaller solution of  $f(x)=1+t\ (t\geq 0)$\\
			$w_2(t)$ & the larger solution of  $f(x)=1+t\ (t\geq 0)$\\
			$\bar{\bm{f}}(x_1,\dots,x_n)$ & $\sum_{i=1}^{n}f(x_i)$\\
			$\lambda$ & the eigenvalue of matrix\\
			$\lambda^*$ & the largest eigenvalue of matrix\\
			$\lambda_*$ & the least eigenvalue of matrix\\
			$f_l(x)$ & $f(1-x)-1\ (0\leq x<1)$\\
			$f_r(x)$ & $f(x+1)-1\ (x\geq 0)$\\
			$g_l(\varepsilon)$ & $f_l^{-1}(\varepsilon)$, the inverse function of $f_l$\\
			$g_r(\varepsilon)$ & $f_r^{-1}(\varepsilon)$, the inverse function of $f_r$\\
			$\mathcal{N}(0,I)$ & standard Gaussian distribution\\
			\bottomrule[1pt] 
		\end{tabular}  
	\end{center}  
\end{table}

\section{Bounds of Forward and Reverse KL Divergence Between Gaussians}\label{sec:quasi_symmetry_kl}
In this section, we give the supremum of reverse KL divergence when forward KL divergence is less than or equal to a positive number $\varepsilon$.  We also show that the supremum is small if $\varepsilon$ is small. These conclusions quantify the approximate symmetry of small KL divergence between Gaussians. We also give the infimum of reverse KL divergence when forward divergence is greater than or equal to a positive number $M$. Furthermore, we give the conditions when the supremum and infimum can be attained.


\subsection{Supremum of Reverse KL Divergence Between Gaussians}\label{sec:small_kl_quasi_symmetry}

We want to know how large the reverse KL divergence can be when forward KL divergence is bounded by a  number $\varepsilon$. The following Theorem \ref{thm:duality_small_KL_general} gives the supremum of reverse KL divergence.

\begin{theorem}\label{thm:duality_small_KL_general}
	For any two $n$-dimensional Gaussian distributions $\mathcal{N}(\bm{\mu}_1,\bm{\Sigma}_1)$ and $\mathcal{N}(\bm{\mu}_2,\bm{\Sigma}_2)$, 
	if $KL(\mathcal{N}(\bm{\mu}_1,\bm{\Sigma}_1)||\mathcal{N}(\bm{\mu}_2,\bm{\Sigma}_2))\leq \varepsilon (\varepsilon\ge 0)$, then
	\begin{equation}\nonumber
		\begin{aligned}
			&KL(\mathcal{N}(\bm{\mu}_2,\bm{\Sigma}_2)||\mathcal{N}(\bm{\mu}_1,\bm{\Sigma}_1))\\
			\leq & \dfrac{1}{2}\left(
			\dfrac{1}{-W_{0}(-e^{-(1+2\varepsilon)})}-\log \dfrac{1}{-W_{0}(-e^{-(1+2\varepsilon)})} -1 \right)
		\end{aligned}
	\end{equation}
	The supremum is attained when the following two conditions hold.
	\begin{enumerate}[(1)]
		\item There exists only one eigenvalue $\lambda_j$ of $B_2^{-1}\bm{\Sigma}_1(B_2^{-1})^\top$ or $B_1^{-1}\bm{\Sigma}_2(B_1^{-1})^\top$ equal to $-W_{0}(-e^{-(1+2\varepsilon)})$ and all other eigenvalues $\lambda_i$ $(i\neq j)$ are equal to $1$, where $B_{1}=P_{1}D_{1}^{1/2}$, $P_{1}$ is an orthogonal matrix whose columns are the eigenvectors of $\bm{\Sigma}_{1}$, $D_{1}=diag(\lambda_1,\dots,\lambda_n)$ whose diagonal elements are the corresponding eigenvalues, $B_2$ is defined in the same way as $B_1$ except on $\bm{\Sigma}_2$.  
		\item $\bm{\mu}_1=\bm{\mu}_2$.
	\end{enumerate}  
\end{theorem}

\noindent\textbf{Overview of proof of Theorem \ref{thm:duality_small_KL_general}}.

Theorem \ref{thm:duality_small_KL_general} can be seen as the following optimization problem $\bm{P_1}$.
\begin{align}
	\text{maximize}\ & KL(\mathcal{N}(\bm{\mu}_2,\bm{\Sigma}_2)||\mathcal{N}(\bm{\mu}_1,\bm{\Sigma}_1)) \\
	\text{\textit{s.t.}}\ & KL(\mathcal{N}(\bm{\mu}_1,\bm{\Sigma}_1)||\mathcal{N}(\bm{\mu}_2,\bm{\Sigma}_2))\leq \varepsilon
\end{align}
Our aim is to solve  problem $\bm{P_1}$ analytically.
The proof consists of the following several steps. 
\begin{enumerate}
	\item \textit{Invertible linear transformation}. We use a linear transformation to turn one of $\mathcal{N}_1$ and $\mathcal{N}_2$ into standard Gaussian. Since diffeomorphism preserves KL divergence \cite{nielsen2020elementary}, both the objective function and the constraints in $\bm{P_1}$ can be simplified. 
	\item \textit{Reducing to new optimization problem}. We reduce $\bm{P_1}$ to the following core problem  $\bm{P_2}$.
	\begin{align}
		\text{maximize}\ &\bar{\bm{f}}(\dfrac{1}{x_1},\dots,\dfrac{1}{x_n})\\
		\text{\textit{s.t.}}\ &\bar{\bm{f}}(x_1,\dots,x_n)\leq n+\varepsilon'
	\end{align}
	where $\bar{\bm{f}}(x_1,\dots,x_n)=\sum_{i=1}^{n}f(x_i)=\sum_{i=1}^{n}x_i-\log x_i\ (x_i\in (0,\infty))$.
	\item \textit{Investigating $f(x)$}. $f(x)$ lies in the core of the problem. We have proven several properties of $f(x)$. The inverse function of $f(x)$ is $f^{-1}=-W(-e^{-x})\ (x\geq 1)$. This allows us to conduct further analysis in all other parts of this paper. Another foundamental property is the relation between $f(x)$ and $f(\frac{1}{x})$, which provides a base for subsequent steps.
	
	\item \textit{Concentrating $\varepsilon'$}. In problem $\bm{P_2}$, the supremum of $\bar{\bm{f}}(\frac{1}{x_1},\dots,\frac{1}{x_n})$ is affected by the domain of each dimension, which is in turn determined by how $\varepsilon'$ is allocated to these dimensions. We call $(\varepsilon_1,\cdots,\varepsilon_n)$ where $\sum_{1}^{n}\varepsilon_i=\varepsilon$ as an \textit{allocation}.
	We prove that $\bar{\bm{f}}(\frac{1}{x_1},\dots,\frac{1}{x_n})$ takes its maximum when $\varepsilon'$ is allocated to only one dimension (\textit{i.e.}, an ``extreme'' allocation). In other words, there exists one $\varepsilon_j=\varepsilon$ and $\varepsilon_i=0(i\neq j)$. The key idea is to prove the convexity of function $\Delta(\varepsilon) = f(\frac{1}{w_1(\varepsilon)})-f(w_1(\varepsilon))$.
\end{enumerate}


We put the key steps of proof of  Theorem \ref{thm:duality_small_KL_general} into Lemma \ref{thm:lemma_sup_fx_invs_nd} and Lemma \ref{thm:duality_small_KL}. After that, we give the main proof.

\begin{lemma}\label{thm:lemma_sup_fx_invs_nd}
	Given $n$-ary function $\bar{\bm{f}}(\bm{x})=\bar{\bm{f}}(x_1,\dots,x_n)=\sum_{i=1}^{n}x_i-\log x_i\ (x_i\in \mathbb{R}^{++})$,
	if
	$\bar{\bm{f}}(x_1,\dots,x_n)\leq n+\varepsilon (\varepsilon>0)$, then 
	\begin{align}\label{equ:n_ary_f_inverse_bound}			
		&\sup \bar{\bm{f}}(\dfrac{1}{x_1},\dots,\dfrac{1}{x_n})\nonumber\\
		=&\dfrac{1}{-W_{0}(-e^{-(1+\varepsilon)})}-\log \dfrac{1}{-W_{0}(-e^{-(1+\varepsilon)})} +n -1
	\end{align}
	The supremum is attained when there exists only one $j$ such that $f(x_j)=1+\varepsilon$ and $f(x_i)=1$ for $i\neq j$.
\end{lemma}

\begin{proof}
	We want to solve the following optimization problem analytically. 
	\begin{align}
		\text{maximize}\ & \bar{\bm{f}}(\dfrac{1}{x_1},\dots,\dfrac{1}{x_n})\\
		\text{\textit{s.t.}}\ & \bar{\bm{f}}(x_1,\dots,x_n)\leq n+\varepsilon
	\end{align}
	Since $f(x)\geq 1$, the constraint $\bar{\bm{f}}(x_1,\dots,x_n)=\sum_{i=1}^{n}f(x_i)=\sum_{i=1}^{n}x_i-\log x_i\leq n+\varepsilon$
	can be replaced by the following constraints
	\begin{align}
		\left(\bigwedge_{i=1}^{n} f(x_i)=x_i-\log x_i\leq 1+\varepsilon_i\right) \wedge
		\left(\bigwedge_{i=1}^{n}\varepsilon_i\geq 0\right) \wedge \sum_{i=1}^{n}\varepsilon_i\leq \varepsilon
	\end{align}
	Given fixed $\varepsilon_1,\dots,\varepsilon_n$ such that $\bigwedge_{i=1}^{n}\varepsilon_i\geq 0 \wedge\sum_{i=1}^{n}\varepsilon_i\leq \varepsilon$, we define 
	\begin{align}
		\bar{\bm{S}}(\varepsilon_1,\dots,\varepsilon_n)
		=&\sup\limits_{ \bigwedge_{i=1}^{n}f(x_i)\leq 1+\varepsilon_i} \bar{\bm{f}}(\dfrac{1}{x_1},\dots,\dfrac{1}{x_n})\nonumber\\
		=&\sum_{i=1}^{n} \sup\limits_{ f(x_i)\leq 1+\varepsilon_i} f(\dfrac{1}{x_i})
		= \sum_{i=1}^{n}S(\varepsilon_i)	\label{equ:sup_all_d_factor_multi_sup}
	\end{align}
	So we have 
	\begin{align}\label{equ:sup_factor}
		\sup \bar{\bm{f}}(\dfrac{1}{x_1},\dots,\dfrac{1}{x_n}) = \sup\limits_{\bigwedge_{i=1}^{n}\varepsilon_i\geq 0 \atop \sum_{i=1}^{n}\varepsilon_i\leq \varepsilon } \bar{\bm{S}}(\varepsilon_1,\dots,\varepsilon_n)
	\end{align}	
	It is easy to know that $\bar{\bm{S}}(\varepsilon_1,\dots,\varepsilon_n)$ is continuous and strictly increasing with $\varepsilon_1,\dots,\varepsilon_n$. So the condition $\sum_{i=1}^{n}\varepsilon_i\leq \varepsilon$ in Equation \eqref{equ:sup_factor} can be changed to $\sum_{i=1}^{n}\varepsilon_i=\varepsilon$.
	
	The remaining proof consists of two steps. In step 1, we find $\bar{\bm{S}}(\varepsilon_1,\dots,\varepsilon_n)$ for fixed $\varepsilon_1,\dots,\varepsilon_n$. In step 2, we find $\sup \bar{\bm{S}}(\varepsilon_1,\dots,\varepsilon_n)$ for any  $\varepsilon_1,\dots,\varepsilon_n$ satisfying $\bigwedge_{i=1}^{n}\varepsilon_i\geq 0 \wedge \sum_{i=1}^{n}\varepsilon_i= \varepsilon$.
	
	\textbf{Step 1}: 
	According to Lemma \ref{prp:sup_f_x_invers_1_d}, for fixed $\varepsilon_i$ we get 
	\begin{align}\label{equ:sup_1_d_f_inverse}
		S(\varepsilon_i)=\sup\limits_{f(x)\leq 1+\varepsilon_i
		} f(\dfrac{1}{x})=f(\dfrac{1}{w_1(\varepsilon_i)})
	\end{align}
	Plugging Equation \eqref{equ:sup_1_d_f_inverse} into Equation \eqref{equ:sup_all_d_factor_multi_sup}, we get 
	\begin{gather}\label{equ:def_supre_fixed_epsilon}
		\bar{\bm{S}}(\varepsilon_1,\dots,\varepsilon_n) = \sum_{i=1}^{n}f(\dfrac{1}{w_1(\varepsilon_i)})
	\end{gather}
	
	\textbf{Step 2}: We define function 
	\begin{align}\label{equ:def_delta_epsilon}
		\Delta(\varepsilon)  = f(\dfrac{1}{w_1(\varepsilon)})-f(w_1(\varepsilon))
		= \dfrac{1}{w_1(\varepsilon)}-w_1(\varepsilon)+2\log w_1(\varepsilon) 
	\end{align}
	Now we prove 
	\begin{align}\label{equ:delta_te_leq_t_delta_e}
		\Delta(t\varepsilon)\leq t\Delta(\varepsilon)\ (0\leq t<1)
	\end{align} 
	When $\varepsilon=0$, it is trivial to verify that $\Delta(0)=0$. In the following we show that  $\Delta(\varepsilon)$ is strictly increasing and strictly  convex.
	It is easy to know $\frac{\dif \Delta(\varepsilon)}{\dif w_1}  = -\frac{1}{w_1^2}+\frac{2}{w_1}-1$.
	Combining Lemma \ref{thm:deriv_w1_w2}, the derivative of $\Delta(\varepsilon)$ is 
	\begin{align}\nonumber
		\dfrac{\dif \Delta(\varepsilon)}{\dif \varepsilon} = &
		\dfrac{\dif \Delta(\varepsilon)}{\dif w_1}\times \dfrac{\dif w_1(\varepsilon)}{\dif \varepsilon}\nonumber\\
		= &\left(-\dfrac{1}{w_1(\varepsilon)^2}+\dfrac{2}{w_1(\varepsilon)}-1\right)\times \dfrac{-w_1(\varepsilon)}{1-w_1(\varepsilon)}
		= \dfrac{1}{w_1(\varepsilon)}-1\nonumber
	\end{align}
	%
	%
	The second order derivative of $\Delta(\varepsilon)$ is	
	\begin{align}
		\dfrac{\dif^2 \Delta(\varepsilon)}{\dif \varepsilon^2}\nonumber
		&= -\dfrac{1}{w_1(\varepsilon)^2}\dfrac{-w_1(\varepsilon)}{1-w_1(\varepsilon)} = \dfrac{1}{w_1(\varepsilon)(1-w_1(\varepsilon))}
	\end{align}
	Since $w_1(\varepsilon)\in (0,1)$ for $\varepsilon>0$, it is easy to know  $\frac{\dif \Delta(\varepsilon)}{\dif \varepsilon} > 0, \frac{\dif^2 \Delta(\varepsilon)}{\dif \varepsilon^2} > 0$ for $\varepsilon>0$.		
	This indicates that $\Delta(\varepsilon)$ is strictly increasing and strictly  convex on $(0,+\infty)$.
	Thus, for any $\varepsilon', \varepsilon''>0$, we have $\Delta((1-t)\varepsilon'+t\varepsilon'')<(1-t)\Delta(\varepsilon')+t\Delta(\varepsilon'')$ for any $0<t<1$. Remember that we have known $\Delta(0)=0$. Since $\Delta(\varepsilon)$ is continuous, 
	it is easy to know 
	\begin{align}
		\Delta(t\varepsilon'') =&\lim\limits_{\varepsilon'\to 0}\Delta((1-t)\varepsilon'+t\varepsilon'')\nonumber\\
		\leq  & \lim\limits_{\varepsilon'\to 0}(1-t)\Delta(\varepsilon')+t\Delta(\varepsilon'')
		=  t\Delta(\varepsilon'')
	\end{align}
	Thus, we can obtain Equation \eqref{equ:delta_te_leq_t_delta_e}.
	
	Therefore, for any $\varepsilon_1,\dots,\varepsilon_n$ satisfying $\bigwedge_{i=1}^{n}\varepsilon_i\geq 0 \wedge \sum_{i=1}^{n}\varepsilon_i= \varepsilon$, 
	we have
	\begin{align}
		\bar{\bm{\Delta}}(\varepsilon_1,\dots,\varepsilon_n)  
		= & \sum_{i=1}^{n}f(\dfrac{1}{w_1(\varepsilon_i)})-f(w_1(\varepsilon_i))
		= 	\sum_{i=1}^{n}\Delta(\varepsilon_i)	\nonumber \\
		=  &\sum_{i=1}^{n}\Delta(\dfrac{\varepsilon_i}{\varepsilon}\varepsilon)	
		\leq 	 \sum_{i=1}^{n}\dfrac{\varepsilon_i}{\varepsilon}\Delta(\varepsilon)
		=	 \Delta(\varepsilon) \label{equ:bound_Delta}
	\end{align}
	Inequality \eqref{equ:bound_Delta} is tight when there exists only one $j$ such that $\varepsilon_j=\varepsilon$ and $\varepsilon_i=0$ for all $i\neq j$.
	This means that for any $\varepsilon_1,\dots,\varepsilon_n$ satisfying $\bigwedge_{i=1}^{n}\varepsilon_i\geq 0 \wedge \sum_{i=1}^{n}\varepsilon_i= \varepsilon$, the following inequality holds.
	\begin{DispWithArrows}
		&\bar{\bm{S}}(\varepsilon_1,\dots,\varepsilon_n) \Arrow{ \eqref{equ:def_supre_fixed_epsilon}}\nonumber\\
		=& \sum_{i=1}^{n}f(\dfrac{1}{w_1(\varepsilon_i)})\Arrow{ \eqref{equ:bound_Delta}} \nonumber \\
		\leq &  \Delta(\varepsilon)+\sum_{i=1}^{n}f(w_1(\varepsilon_i)) \Arrow{$\overset{\text{Lemma \ref{thm:solution_f_x}}}{\text{\eqref{equ:def_delta_epsilon}}}$} \nonumber \\
		\leq &   \dfrac{1}{w_1(\varepsilon)}-\log\dfrac{1}{w_1(\varepsilon)}-(w_1(\varepsilon)-\log w_1(\varepsilon)) \nonumber\\
		& + \sum_{i=1}^{n}(1+\varepsilon_i) \nonumber  \\
		=& \dfrac{1}{w_1(\varepsilon)}-\log\dfrac{1}{w_1(\varepsilon)} - (1+\varepsilon) + n + \varepsilon \nonumber\\
		= &\dfrac{1}{w_1(\varepsilon)}-\log\dfrac{1}{w_1(\varepsilon)} + n - 1 \nonumber\\
		=& \dfrac{1}{-W_{0}(-e^{-(1+\varepsilon)})}-\log \dfrac{1}{-W_{0}(-e^{-(1+\varepsilon)})} +n -1		
	\end{DispWithArrows}
	Finally, we have 
	\begin{align}
		&\sup \bar{\bm{f}}(\dfrac{1}{x_1},\dots,\dfrac{1}{x_n})
		= \sup\limits_{\bigwedge_{i=1}^{n}\varepsilon_i\geq 0\atop \sum_{i=1}^{n}\varepsilon_i\leq \varepsilon} \bar{\bm{S}}(\varepsilon_1,\dots,\varepsilon_n) \nonumber\\
		= & \dfrac{1}{-W_{0}(-e^{-(1+\varepsilon)})}-\log \dfrac{1}{-W_{0}(-e^{-(1+\varepsilon)})} +n -1	
	\end{align}
	$\bar{\bm{f}}(1/x_1,\dots,1/x_n)$ reaches its supremum when there exists only one $j$ such that $f(x_j)=1+\varepsilon$ and $f(x_i)=1$ for $i\neq j$. 
	
	$\hfill\square$ 
\end{proof}

In the following Lemma \ref{thm:duality_small_KL}, we deal with the case when one Gaussian is standard. Then we extend Lemma \ref{thm:duality_small_KL} to general case.

\begin{lemma}\label{thm:duality_small_KL}
	Let 
	$\mathcal{N}(0,I)$ be standard Gaussian, $\varepsilon$ be a positive number.
	For any $n$-dimensional Gaussian distribution $\mathcal{N}(\bm{\mu},\bm{\Sigma})$,  
	\begin{enumerate}[label=(\alph*), ref=\ref{thm:duality_small_KL}\alph*]
		\setlength{\itemsep}{0pt}
		\setlength{\parskip}{0pt}
		\item
		\label{thm:forward_to_backward_KL_small_indendent_n}
		If $KL(\mathcal{N}(\bm{\mu},\bm{\Sigma})||\mathcal{N}(0,I))\leq \varepsilon$, then 
		\begin{equation}\nonumber
			\begin{aligned}
				&KL(\mathcal{N}(0,I)||\mathcal{N}(\bm{\mu},\bm{\Sigma}))\\
				\leq & \dfrac{1}{2}\left(
				\dfrac{1}{-W_{0}(-e^{-(1+2\varepsilon)})}-\log \dfrac{1}{-W_{0}(-e^{-(1+2\varepsilon)})} -1 \right)
			\end{aligned}
		\end{equation}
		\item \label{thm:backward_to_forward_KL_small_indendent_n}
		If $KL(\mathcal{N}(0,I)||\mathcal{N}(\bm{\mu},\bm{\Sigma}))\leq \varepsilon$, then
		\begin{equation}\nonumber
			\begin{aligned}
				&KL(\mathcal{N}(\bm{\mu},\bm{\Sigma})||\mathcal{N}(0,I))\\
				\leq & \dfrac{1}{2}\left(
				\dfrac{1}{-W_{0}(-e^{-(1+2\varepsilon)})}-\log \dfrac{1}{-W_{0}(-e^{-(1+2\varepsilon)})} -1 \right)
			\end{aligned}
		\end{equation}
		
	\end{enumerate}
\end{lemma}
\begin{proof} 
	(a) 
	According to the definition of KL divergence, we have 
	\begin{align}
		KL(\mathcal{N}(\bm{\mu},\bm{\Sigma})||\mathcal{N}(0,I)) 
		= & \dfrac{1}{2} \left(-\log|\bm{\Sigma}|+\mathop{\mathrm{Tr}}(\bm{\Sigma})+\bm{\mu}^{\top}\bm{\mu}-n\right)\nonumber \\
		KL(\mathcal{N}(0,I)||\mathcal{N}(\bm{\mu},\bm{\Sigma}))
		= & \dfrac{1}{2} \left(\log|\bm{\Sigma}|+\mathop{\mathrm{Tr}}(\bm{\Sigma}^{-1})+\bm{\mu}^{\top}\bm{\Sigma}^{-1}\bm{\mu}-n\right)  \nonumber
	\end{align}
	where $n$ is the dimension of the distribution.
	The positive definite matrix $\bm{\Sigma}$ has factorization $\bm{\Sigma}=PDP^{\top}$ where $P$ is an orthogonal matrix whose columns are the eigenvectors of $\bm{\Sigma}$, $D=diag(\lambda_1,\dots,\lambda_n)$ ($\lambda_i>0$) whose diagonal elements are the corresponding eigenvalues. 
	We also have
	\begin{gather}
		|\bm{\Sigma}|=|P||D||P^{\top}|=|D|=\prod_{i=1}^n \lambda_i \label{equ:det_to_multi} \\ 
		\log |\bm{\Sigma}| = \sum_{i=1}^n\log \lambda_i, 
		-\log |\bm{\Sigma}| = \sum_{i=1}^n\log \dfrac{1}{\lambda_i}\\
		\mathop{\mathrm{Tr}}(\bm{\Sigma})=\mathop{\mathrm{Tr}}(PDP^{\top})=\mathop{\mathrm{Tr}}(P^{\top}PD)=\mathop{\mathrm{Tr}}(D)=\sum_{i=1}^n\lambda_i\\
		\mathop{\mathrm{Tr}}(\bm{\Sigma}^{-1})=\sum_{i=1}^n\lambda'_i=\sum_{i=1}^n\dfrac{1}{\lambda_i} \label{equ:tr_sigma}
	\end{gather}
	where $\lambda'_i=1/\lambda_i$ are eigenvalues of $\bm{\Sigma}^{-1}$. 
	
	If $KL(\mathcal{N}(\bm{\mu},\bm{\Sigma})||\mathcal{N}(0,I))\leq \varepsilon$, we have $-\log|\bm{\Sigma}|+\mathop{\mathrm{Tr}}(\bm{\Sigma})+\bm{\mu}^{\top}\bm{\mu}-n \leq  2\varepsilon$.
	This condition is equal to the following conditions
	\begin{align}
		-\log|\bm{\Sigma}|+\mathop{\mathrm{Tr}}(\bm{\Sigma})=\sum_{i=1}^n \lambda_i-\log \lambda_i&\leq n+\varepsilon_1 \label{equ:small_fKL_item1_e1}\\
		\bm{\mu}^{\top}\bm{\mu}&\leq 2\varepsilon-\varepsilon_1 \label{equ:small_fKL_item2_e}\\
		0\leq \varepsilon_1 & \leq 2\varepsilon
	\end{align}
	In the following, we find the maximum of $\log|\bm{\Sigma}|+\mathop{\mathrm{Tr}}(\bm{\Sigma}^{-1})$ and $\bm{\mu}^{\top}\bm{\Sigma}^{-1}\bm{\mu}$, respectively.
	From Equation \eqref{equ:small_fKL_item1_e1}, we have
	\begin{align}
		\sum_{i=1}^n \lambda_i-\log \lambda_i \leq  n+\varepsilon_1 \label{equ:bound_fKL_left_e1}
	\end{align}
	Applying Lemma \ref{thm:lemma_sup_fx_invs_nd} on Inequality \eqref{equ:bound_fKL_left_e1}, we can obtain
	\begin{align}\label{equ:small_fKL_to_bKL_left_item_bound}
		& \sum_{i=1}^n \dfrac{1}{\lambda_i}-\log \dfrac{1}{\lambda_i}
		= \log|\bm{\Sigma}|+\mathop{\mathrm{Tr}}(\bm{\Sigma}^{-1}) \nonumber\\
		\leq &   \dfrac{1}{-W_{0}(-e^{-(1+\varepsilon_1)})}-\log \dfrac{1}{-W_{0}(-e^{-(1+\varepsilon_1)})} +n -1 
	\end{align}
	Moreover, since $f(x)=x-\log x$ takes the minimum value $f(1)=1$ at $x=1$, it is easy to know 
	$\lambda_i-\log \lambda_i\leq 1+\varepsilon_1$ from Inequality \eqref{equ:bound_fKL_left_e1}.
	According to Lemma \ref{prp:sup_f_x_invers_1_d}, we know
	\begin{align}
		w_1(\varepsilon_1) \leq \lambda_i\leq w_2(\varepsilon_1),\ 
		\dfrac{1}{w_2(\varepsilon_1)}\leq \lambda'_i=\dfrac{1}{\lambda_i}\leq \dfrac{1}{w_1(\varepsilon_1)} \label{equ:lambda_inverse_bound}
	\end{align}
	We also have $\bm{\mu}^{\top}\bm{\Sigma}^{-1}\bm{\mu}\leq \lambda'^*\bm{\mu}^{\top}\bm{\mu}$ where $\lambda'^*$ is the maximum eigenvalue of $\bm{\Sigma}^{-1}$.
	Combining Equation \eqref{equ:small_fKL_item2_e} and \eqref{equ:lambda_inverse_bound}, we can know
	\begin{equation}\label{equ:bound_mumu_fKL_to_bKL}
		\bm{\mu}^{\top}\bm{\Sigma}^{-1}\bm{\mu}\leq \lambda'^*(2\varepsilon-\varepsilon_1) \leq\dfrac{2\varepsilon-\varepsilon_1}{w_1(\varepsilon_1)}
	\end{equation} 
	Now note that Inequalities \eqref{equ:small_fKL_to_bKL_left_item_bound} and \eqref{equ:bound_mumu_fKL_to_bKL} are tight simultaneously when there exists one $\lambda_j=w_1(\varepsilon_1)$ and all other $\lambda_i=1$ for $i\neq j$.
	Thus, we can add the two sides of Inequalities \eqref{equ:small_fKL_to_bKL_left_item_bound} and \eqref{equ:bound_mumu_fKL_to_bKL} and get 
	\begin{align}
		&KL(\mathcal{N}(0,I)||\mathcal{N}(\bm{\mu},\bm{\Sigma}))\nonumber\\
		= & \dfrac{1}{2} \left(\log|\bm{\Sigma}|+\mathop{\mathrm{Tr}}(\bm{\Sigma}^{-1})+\bm{\mu}^{\top}\bm{\Sigma}^{-1}\bm{\mu}-n\right) \nonumber\\
		\leq & \dfrac{1}{2}\left(\dfrac{1}{-W_{0}(-e^{-(1+\varepsilon_1)})}-\log \dfrac{1}{-W_{0}(-e^{-(1+\varepsilon_1)})}+n-1 \right. \nonumber\\
		&\left. + \dfrac{2\varepsilon-\varepsilon_1}{w_1(\varepsilon_1)} -n\right) \nonumber\\
		= & \dfrac{1}{2}\left(
		\dfrac{1+2\varepsilon-\varepsilon_1}{w_1(\varepsilon_1)}-\log \dfrac{1}{w_1(\varepsilon_1)} -1 \right) \label{equ:bound_bKL}\\
		=& U(\varepsilon_1)\ (0\leq \varepsilon_1\leq 2\varepsilon)  \nonumber
	\end{align}
	Notice that the derivative of $U(\varepsilon_1)$ is
	\begin{equation}\label{equ:deri_supre}
		\begin{aligned}
			U'(\varepsilon_1)
			=&\dfrac{1}{2}\left(\dfrac{w_1(\varepsilon_1)+2\varepsilon-\varepsilon_1}{w_1(\varepsilon_1)(1-w_1(\varepsilon_1))} - \dfrac{1}{1-w_1(\varepsilon_1)} \right) \\
			= &\dfrac{1}{2}\dfrac{2\varepsilon-\varepsilon_1}{w_1(\varepsilon_1)(1-w_1(\varepsilon_1))}		
		\end{aligned}
	\end{equation}
	Since $w_1(\varepsilon_1)\in (0,1)$ for $\varepsilon_1>0$ and $0\leq \varepsilon_1\leq 2\varepsilon$, we can know $U'(\varepsilon_1) \geq 0$ for $\varepsilon_1>0$. Thus, $U(\varepsilon_1)$ takes the maximum value at $\varepsilon_1=2\varepsilon$.
	Finally, we have 
	\begin{equation}\label{equ:supremum_f_to_b_kl}
		\begin{aligned}
			& KL(\mathcal{N}(0,I)||\mathcal{N}(\bm{\mu},\bm{\Sigma})) \\
			\leq & U(2\varepsilon)
			=  \dfrac{1}{2}\left(
			\dfrac{1}{-W_{0}(-e^{-(1+2\varepsilon)})}-\log \dfrac{1}{-W_{0}(-e^{-(1+2\varepsilon)})} -1 \right) 
		\end{aligned}
	\end{equation}
	Inequality \eqref{equ:supremum_f_to_b_kl} is tight only when there exists one $\lambda_j=-W_{0}(-e^{-(1+2\varepsilon)})$ and all other $\lambda_i=1$ for $i\neq j$, and $|\bm{\mu}|=0$.
	
	We can see that when $\varepsilon$ is small, the right hand side of Equation \eqref{equ:supremum_f_to_b_kl} is also small.
	
	(b) The proof of Theorem \ref{thm:backward_to_forward_KL_small_indendent_n} is similar. See Appendix \ref{sec:proof_thm:backward_to_forward_KL_small_indendent_n} for the details.

	$\hfill\square$ 
\end{proof}

In the following, we extend Lemma \ref{thm:duality_small_KL} to general Gaussians. 
Before our generalized theorem, we recall the following proposition which states that diffeomorphism preserves KL divergence ($f$-divergence) \cite{nielsen2020elementary}.
\begin{proposition} \label{thm:preserve_KL}
	(See \cite{nielsen2020elementary}) Let $\bm{z}=f(\bm{x})$ be a diffeomorphism, $X_1\sim p_X$ and $X_2\sim q_X$ be two random variables and $Z_1=f(X_1)\sim p_Z$, $Z_2=f(X_2)\sim q_Z$. 
	Then $KL(p_X||q_X)=KL(p_Z||q_Z)$.
	
\end{proposition}
\vspace{10pt}
\noindent \textbf{Main Proof of Theorem \ref{thm:duality_small_KL_general}}\\
\begin{proof}
With the help of Proposition \ref{thm:preserve_KL}, it is not hard to extend Lemma \ref{thm:duality_small_KL} to general Gaussians. The key idea is to use an invertible linear transformation to transform one Gaussian to standard Gaussian, and then apply Lemma \ref{thm:duality_small_KL}. Please see Appendix \ref{sec:app_proof_thm:duality_small_KL_general} for details.	
		
		$\hfill\square$	
\end{proof}

To investigate the bound in Theorem \ref{thm:duality_small_KL_general} further, we can expand Lambert $W$ function using the series presented in \cite{2016PrincetonCompanion, corless1996lambertw} for small $\varepsilon$.
This is expressed by the following Theorem.

\begin{theorem}\label{thm:symmetry_KL_simplify}
	For any two $n$-dimensional Gaussian distributions $\mathcal{N}(\bm{\mu}_1,\bm{\Sigma}_1)$, $\mathcal{N}(\bm{\mu}_2,\bm{\Sigma}_2)$,  and a small positive number $\varepsilon$,  
	\label{cor:forward_to_backward_KL_small_general_simplify}
	if $KL(\mathcal{N}(\bm{\mu}_1,\bm{\Sigma}_1)||\mathcal{N}(\bm{\mu}_2,\bm{\Sigma}_2))\leq \varepsilon$, then
	\begin{equation}
		\begin{aligned}
			KL(\mathcal{N}(\bm{\mu}_2,\bm{\Sigma}_2)||\mathcal{N}(\bm{\mu}_1,\bm{\Sigma}_1)) \leq \varepsilon + 2\varepsilon^{1.5} + O(\varepsilon^2)
		\end{aligned}
	\end{equation}
\end{theorem}

\begin{proof}
	Please see Appendix \ref{sec:appendix_proof_symmetry_kl_simplify} for the details of the proof.

	$\hfill\square$ 
\end{proof}

Theorem \ref{thm:duality_small_KL_general} holds for any two Gaussians $\mathcal{N}(\bm{\mu}_1,\bm{\Sigma}_1)$ and $\mathcal{N}(\bm{\mu}_2,\bm{\Sigma}_2)$. According to the proof of Theorem \ref{thm:duality_small_KL_general} (Lemma \ref{thm:duality_small_KL}), one of $\mathcal{N}(\bm{\mu}_1,\bm{\Sigma}_1)$ and $\mathcal{N}(\bm{\mu}_2,\bm{\Sigma}_2)$ can be fixed. It is not hard to extend Lemma \ref{thm:duality_small_KL} to case where the fixed one Gaussian is not standard. We can apply linear transformation (see Equation \eqref{equ:linear_tranform_for_symmetry}) as what we have done in the main proof of Theorem \ref{thm:duality_small_KL_general} (see Appendix \ref{sec:app_proof_thm:duality_small_KL_general}). Therefore, we have the following corollary.
\begin{corollary}
	Theorem \ref{thm:duality_small_KL_general} and Theorem \ref{thm:symmetry_KL_simplify} hold when one of $\mathcal{N}(\bm{\mu}_1,\bm{\Sigma}_1)$ and $\mathcal{N}(\bm{\mu}_2,\bm{\Sigma}_2)$ is fixed.
\end{corollary}

\begin{remark}
	The supremum in Theorem \ref{thm:duality_small_KL_general} has the following properties.
	\begin{enumerate}
		\item  The supremum is small (zero) when $\varepsilon$ is small (zero). Figure \ref{fig:sup_KL} shows some values of the supremum of KL divergence. 
		\item  The supremum increases rapidly when $\varepsilon>2$ due to the rapid increase of term $\frac{1}{-W_0(-e^{-(1+2\varepsilon)})}$.
		\item 
		It is hard to reach the supremum in typical applications (\textit{e.g.}, in machine learning practice) due to the strict conditions.
		
		\item The bound is independent of the dimension $n$. This is a critical property in high-dimensional problems.
	\end{enumerate}
\end{remark}
\begin{figure}[h]
	\centering	\includegraphics[width=7cm]{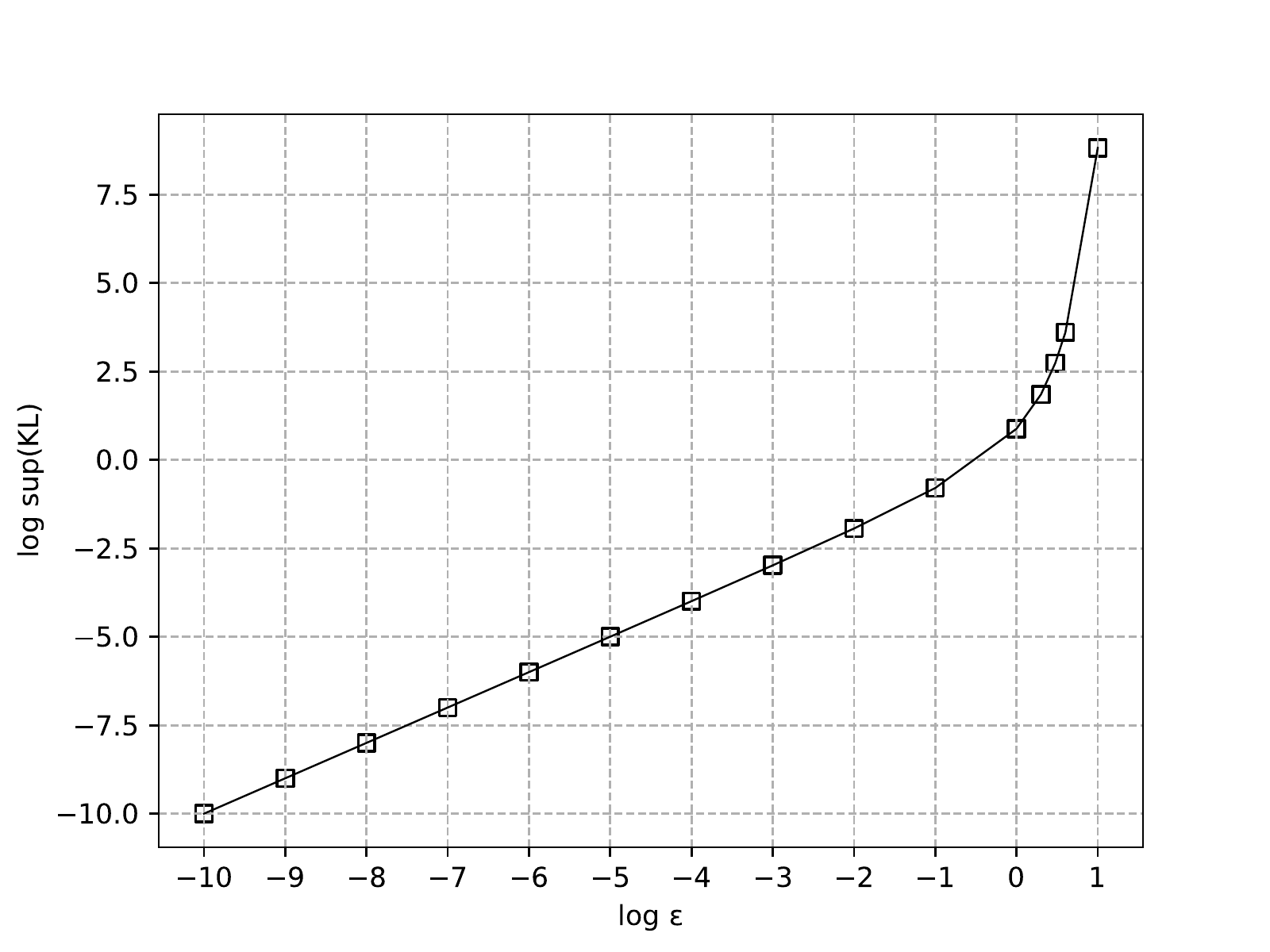}
	\caption{Values of supremum of KL divergence shown on a logarithmic scale.}
	\label{fig:sup_KL}
\end{figure}

\subsection{Infimum of Reverse KL Divergence Between Gaussians}\label{sec:large_kl_quasi_symmetry}

We also want to know how small the reverse KL divergence could be when forward KL divergence is not less than a given number. 
In this subsection, we give the infimum of $KL(\mathcal{N}_2||\mathcal{N}_1)$ when $KL(\mathcal{N}_1||\mathcal{N}_2)\geq M$ ($M>0$). The main result is shown in Theorem  \ref{thm:duality_big_KL_general}.

\begin{theorem}\label{thm:duality_big_KL_general}
	For any two $n$-dimensional Gaussian distributionss $\mathcal{N}(\bm{\mu}_1,\bm{\Sigma}_1)$ and $\mathcal{N}(\bm{\mu}_2,\bm{\Sigma}_2)$, if $KL(\mathcal{N}(\bm{\mu}_1,\bm{\Sigma}_1)||\mathcal{N}(\bm{\mu}_2,\bm{\Sigma}_2))\geq M(M>0)$, then  
	\begin{equation}
		\small
		\begin{aligned}
			&KL(\mathcal{N}(\bm{\mu}_2,\bm{\Sigma}_2)||\mathcal{N}(\bm{\mu}_1,\bm{\Sigma}_1)) \\
			\geq & \dfrac{1}{2}\left(\dfrac{1}{-W_{-1}(-e^{-(1+2M)})}-\log \dfrac{1}{-W_{-1}(-e^{-(1+2M)})} -1\right)
		\end{aligned}
	\end{equation}
	The infimum is attained when the following two conditions hold.
	\begin{enumerate}[(1)]
		\item There exists only one eigenvalue $\lambda_j$ of $B_2^{-1}\bm{\Sigma}_1(B_2^{-1})^\top$ or $B_1^{-1}\bm{\Sigma}_2(B_1^{-1})^\top$ equal to $-W_{-1}(-e^{-(1+2M)})$ and all other eigenvalues $\lambda_i$ $(i\neq j)$ are equal to $1$, where $B_{1}=P_{1}D_{1}^{1/2}$, $P_{1}$ is an orthogonal matrix whose columns are the eigenvectors of $\bm{\Sigma}_{1}$, $D_{1}=diag(\lambda_1,\dots,\lambda_n)$ whose diagonal elements are the corresponding eigenvalues, $B_2$ is defined in the same way as $B_1$ except on $\bm{\Sigma}_2$.  
		\item $\bm{\mu}_1=\bm{\mu}_2$.
	\end{enumerate}  
\end{theorem}

\noindent\textbf{Proof of Theorem \ref{thm:duality_big_KL_general}}\\
Intuitively, the problems in  Theorem \ref{thm:duality_small_KL_general} and \ref{thm:duality_big_KL_general} should has a tight relation. In this paper, we give two proofs of Theorem \ref{thm:duality_big_KL_general}.
The first proof has the similar structure as that of Theorem \ref{thm:duality_small_KL_general}, except that Theorem \ref{thm:duality_big_KL_general} needs $W_{-1}$. We put the first proof in Appendix \ref{app:big_KL_general_proof1}.
The second proof can be drawn from Theorem \ref{thm:duality_small_KL_general} directly by analyzing the supremum. We put the second proof below. These two proofs can verify each other.

\begin{proof}
	In the following, we give the second proof which is drawn from Theorem \ref{thm:duality_small_KL_general}.
	
	Suppose $KL(\mathcal{N}(\bm{\mu}_1,\bm{\Sigma}_1)||\mathcal{N}(\bm{\mu}_2,\bm{\Sigma}_2))\leq t\ (t>0)$, according to Theorem \ref{thm:duality_small_KL_general}, we know 
	\begin{align}
		&KL(\mathcal{N}(\bm{\mu}_2,\bm{\Sigma}_2)||\mathcal{N}(\bm{\mu}_1,\bm{\Sigma}_1))\nonumber \\
		\leq  & \dfrac{1}{2}\left(
		\dfrac{1}{-W_{0}(-e^{-(1+2t)})}-\log \dfrac{1}{-W_{0}(-e^{-(1+2t)})} -1 \right)  \\
		= & \dfrac{1}{2}\left(
		\dfrac{1}{w_1(2t)}-\log \dfrac{1}{w_1(2t)} -1 \right)\\
		= & \bar{S}(t)
	\end{align}
	Since $\dfrac{1}{w_1(2t)}$ is strictly  increasing with $t$, $\bar{S}(t)$ is continuous and strictly  increasing with $t$.
	Besides, the range of function $\bar{S}(t)$ for $(t>0)$ is $(0,+\infty)$. 
	
	Given positive number $M$, according to Theorem \ref{thm:duality_small_KL_general}, there exists  $\mathcal{N}(\bm{\mu}_1,\bm{\Sigma}_1)$, $\mathcal{N}(\bm{\mu}_2,\bm{\Sigma}_2)$ and $m$ such that 
	\begin{align}
		\bar{S}(m)=&M\\
		KL(\mathcal{N}(\bm{\mu}_1,\bm{\Sigma}_1)||\mathcal{N}(\bm{\mu}_2,\bm{\Sigma}_2))=&M\\
		KL(\mathcal{N}(\bm{\mu}_2,\bm{\Sigma}_2)||\mathcal{N}(\bm{\mu}_1,\bm{\Sigma}_1))=&m
	\end{align}
	Thus, given the precondition $KL(\mathcal{N}(\bm{\mu}_1,\bm{\Sigma}_1)||\mathcal{N}(\bm{\mu}_2,\bm{\Sigma}_2))\geq M$, we can know that 
	\begin{align}
		\inf KL(\mathcal{N}(\bm{\mu}_2,\bm{\Sigma}_2)||\mathcal{N}(\bm{\mu}_1,\bm{\Sigma}_1)) \leq m
	\end{align}
	In the following, we show 
	\begin{align}\label{equ:inf_KL_in_concise_proof}
		\inf KL(\mathcal{N}(\bm{\mu}_2,\bm{\Sigma}_2)||\mathcal{N}(\bm{\mu}_1,\bm{\Sigma}_1)) = m
	\end{align}	
	must holds.	Otherwise, there exists an $m'<m$ and $\mathcal{N}(\bm{\mu}_1,\bm{\Sigma}_1)$, $\mathcal{N}(\bm{\mu}_2,\bm{\Sigma}_2)$ such that 
	\begin{align}
		KL(\mathcal{N}(\bm{\mu}_1,\bm{\Sigma}_1)||\mathcal{N}(\bm{\mu}_2,\bm{\Sigma}_2))\geq M\\
		KL(\mathcal{N}(\bm{\mu}_2,\bm{\Sigma}_2)||\mathcal{N}(\bm{\mu}_1,\bm{\Sigma}_1))=m' \label{equ:KL_n2_n1_m'}
	\end{align}
	Applying Theorem \ref{thm:duality_small_KL_general} on Equation \eqref{equ:KL_n2_n1_m'}, it is easy to know 
	\begin{align}
		\sup KL(\mathcal{N}(\bm{\mu}_1,\bm{\Sigma}_1)||\mathcal{N}(\bm{\mu}_2,\bm{\Sigma}_2))= \bar{S}(m')
	\end{align}
	This contradicts with the precondition $KL(\mathcal{N}(\bm{\mu}_1,\bm{\Sigma}_1)||\mathcal{N}(\bm{\mu}_2,\bm{\Sigma}_2))\geq M$ because $\bar{S}(m')<\bar{S}(m)= M$. Thus, Equation \eqref{equ:inf_KL_in_concise_proof} holds.
	
	Now we can solve $m$ from $\bar{S}(m)=M$ as follows.
	\begin{align}
		&\dfrac{1}{2}\left(\dfrac{1}{-W_{0}(-e^{-(1+2m)})}-\log \dfrac{1}{-W_{0}(-e^{-(1+2m)})} -1 \right) = M \nonumber\\
		\Leftrightarrow & \dfrac{1}{-W_{0}(-e^{-(1+2m)})}-\log \dfrac{1}{-W_{0}(-e^{-(1+2m)})}=1+2M  \nonumber\\ 
		\Leftrightarrow & \dfrac{1}{-W_{0}(-e^{-(1+2m)})}= -W_{-1}(-e^{-(1+2M)})  \nonumber\\
		\Leftrightarrow & \dfrac{1}{-W_{-1}(-e^{-(1+2M)})} = -W_{0}(-e^{-(1+2m)})  \nonumber\\
		\Leftrightarrow & \dfrac{1}{-W_{-1}(-e^{-(1+2M)})} -\log \dfrac{1}{-W_{-1}(-e^{-(1+2M)})} = 1+2m  \nonumber\\
		\Leftrightarrow &  m=\dfrac{1}{2}\left(\dfrac{1}{-W_{-1}(-e^{-(1+2M)})}-\log \dfrac{1}{-W_{-1}(-e^{-(1+2M)})} -1 \right) \label{equ:solve_m_from_M}
	\end{align}
	where the third and fifth equations follow from Lemma \ref{thm:solution_f_x}.
	Plugging Equation \eqref{equ:solve_m_from_M} into \eqref{equ:inf_KL_in_concise_proof}, we can prove Theorem \ref{thm:duality_big_KL_general}.
	
	$\hfill\square$
\end{proof}

\begin{remark}
	The bound in Theorem \ref{thm:duality_big_KL_general} has the similar form with that in Theorem \ref{thm:duality_small_KL_general}. In fact, Theorem \ref{thm:duality_small_KL_general} and Theorem \ref{thm:duality_big_KL_general} forms a duality. Firstly, these two theorems can be proved independently in the similar way. Secondly, these two theorems can be derived from each other. 
\end{remark}


\section{Relaxed Triangle Inequality}\label{sec:quasi_triangle}

Until now, we have quantified the approximate symmetry of KL divergence between Gaussians. A natural question is how large can $KL(\mathcal{N}_1||\mathcal{N}_3)$ be when $KL(\mathcal{N}_1||\mathcal{N}_2)$ and $KL(\mathcal{N}_2)||\mathcal{N}_3)$ are small for three Gaussians $\mathcal{N}_1$, $\mathcal{N}_2$, and $\mathcal{N}_3$.
In this section, we give a bound of $KL(\mathcal{N}_1||\mathcal{N}_3)$ that is also independent of the dimension $n$. 
Proving the relaxed triangle inequality is more difficult.
The main result is presented in Theorem \ref{thm:triangle_n1_n2_n3}. 
We put the key steps of proof of Theorem \ref{thm:triangle_n1_n2_n3} in Lemma \ref{thm:simple_relation_ab} $\sim$ \ref{thm:kl_triangle_fw2w2_extreme_distribute_e_max} and Lemma \ref{thm:triangle_n1_n2_standard}.

\begin{theorem}\label{thm:triangle_n1_n2_n3}
	For any three $n$-dimensional Gaussians $\mathcal{N}(\bm{\mu}_i,\bm{\Sigma}_i)(i\in\{1,2,3\})$  
	such that $KL(\mathcal{N}(\bm{\mu}_1,\bm{\Sigma}_1)||\mathcal{N}(\bm{\mu}_2,\bm{\Sigma}_2))\leq \varepsilon_1$ and 
	$ KL(\mathcal{N}(\bm{\mu}_2,\bm{\Sigma}_2)||\mathcal{N}(\bm{\mu}_3,\bm{\Sigma}_3))\leq \varepsilon_2$ for $\varepsilon_1, \varepsilon_2\ge 0$, then 
	\begin{align}
		&KL((\mathcal{N}(\bm{\mu}_1,\bm{\Sigma}_1)||\bm{\Sigma}(\bm{\mu}_3,\bm{\Sigma}_3))\nonumber\\
		<&  \varepsilon_1 + \varepsilon_2 + \dfrac{1}{2}\left(W_{-1}(-e^{-(1+2\varepsilon_1)})W_{-1}(-e^{-(1+2\varepsilon_2)}) \vphantom{\left(\sqrt{2\varepsilon_1}+\sqrt{\dfrac{2\varepsilon_2}{-W_{0}(-e^{-(1+2\varepsilon_2)})}}\right)^2}\right. \\ 
		 & +W_{-1}(-e^{-(1+2\varepsilon_1)})+W_{-1}(-e^{-(1+2\varepsilon_2)})+1 \\ 
		 &\left. -W_{-1}(-e^{-(1+2\varepsilon_2)})\left(\sqrt{2\varepsilon_1}+\sqrt{\dfrac{2\varepsilon_2}{-W_{0}(-e^{-(1+2\varepsilon_2)})}}\right)^2 \right)
	\end{align}
\end{theorem}


\noindent\textbf{Overview of proof of Theorem \ref{thm:triangle_n1_n2_n3}}\\
We want to solve the following optimzation problem $\bm{P}_3$ analytically. 
\begin{align}
	\text{maximize}\ &KL(\mathcal{N}(\bm{\mu}_1,\bm{\Sigma}_1)||\mathcal{N}(\bm{\Sigma}(\bm{\mu}_3,\bm{\Sigma}_3))\nonumber\\
	\text{\textit{s.t.}}\ & KL(\mathcal{N}(\bm{\mu}_1,\bm{\Sigma}_1)||\mathcal{N}(\bm{\mu}_2,\bm{\Sigma}_2))\leq \varepsilon_1\nonumber\\
	&  KL(\mathcal{N}(\bm{\mu}_2,\bm{\Sigma}_2)||\mathcal{N}(\bm{\mu}_3,\bm{\Sigma}_3))\leq \varepsilon_2 \nonumber
\end{align}
Unfortunately, it is hard to find the supremum due to the complexity caused by Lambert $W$ function. 
So we relax the constraints to simplify the problem.
Our proof consists of the following several steps.
\begin{enumerate}
	\item \textit{Invertible linear transformation}. The first step is similar to that of Theorem \ref{thm:duality_small_KL_general}. We use a linear transformation on $\mathcal{N}_1$, $\mathcal{N}_2$, and $\mathcal{N}_3$ to simply the problem. After transformation, $\mathcal{N}_2$ is converted to standard Gaussian. 
	\item \textit{Relaxing constraints}. In this step, we relax the constraints to get a simpler problem, which is in turn reduced to the following core problem $\bm{P}_4$.
	\begin{align}
		\text{maximize} \ & \sum_{i=1}^{n}\lambda_{1,[i]}\lambda{'}_{2,[i]}- \log \lambda_{1,[i]}\lambda{'}_{2,[i]}\label{equ:obj_reduced_core_triangle_problem}\\
		\text{\textit{s.t.}} \ & \lambda_{1,[i]}-\log \lambda_{1,[i]}= 1+\varepsilon_{1,[i]}\ (1\leq i\leq n) \nonumber\\
		& \bigwedge\limits_{i=1}^n\varepsilon_{1,[i]}\geq 0\wedge \sum\limits_{i=1}^{n}\varepsilon_{1,[i]}=2\varepsilon_1\nonumber\\
		& \lambda{'}_{2,[i]}-\log \lambda{'}_{2,[i]}= 1+\varepsilon_{2,[i]}\ (1\leq i\leq n)\nonumber \\
		& \bigwedge\limits_{i=1}^n\varepsilon_{2,[i]}\geq 0\wedge\sum\limits_{i=1}^{n}\varepsilon_{2,[i]}=2\varepsilon_2\nonumber
	\end{align}
	where $\lambda_{1,[i]}, \lambda'_{2,[i]}$ are the eigenvalues of $\bm{\Sigma}_1, \bm{\Sigma_2^{-1}}$ arranged in decreasing order, respectively, and $\varepsilon_{1,[i]},\varepsilon_{2,[i]}$ are  arranged in decreasing order too.
	
	\item \label{text:overview_proof_key_lemma}\textit{Concentrating $\varepsilon_1$ and $\varepsilon_2$}. The objective function \eqref{equ:obj_reduced_core_triangle_problem} is determined by  how 
	$\varepsilon_{1}$ and $\varepsilon_{2}$ are allocated to $(\varepsilon_{1,[1]},\cdots,\varepsilon_{1,[n]})$ and $(\varepsilon_{2,[1]},\cdots,\varepsilon_{2,[n]})$.
	We prove that an ``extreme allocation'' can make the objective function maximized. In other words, Equation  
	\eqref{equ:obj_reduced_core_triangle_problem} takes its maximum when $\varepsilon_{1,[1]}=\varepsilon_1$ and $\varepsilon_{2,[1]}=\varepsilon_2$. 
	We use a key Lemma \ref{thm:kl_triangle_fw2w2_extreme_distribute_e_max} to deal with the two dimensional case (\textit{i.e.}, $n=2$). Then, we use Lemma \ref{thm:trace_AB_supremum} to extend the conclusion to high dimensional problems.
	
	Lemma \ref{thm:kl_triangle_fw2w2_extreme_distribute_e_max} is the most tricky part in this paper. In the proof, concentrating $\varepsilon_1$ and $\varepsilon_2$ is much harder than that in last section for Theorem \ref{thm:duality_small_KL_general}. $f(x)=x-\log x$ is a transcendental function whose inverse function is expressed by Lambert $W$ function. This makes even a 2-dimensional case of problem $\bm{P}_4$ hard to solve. 
	We use an iterated way to prove Lemma \ref{thm:kl_triangle_fw2w2_extreme_distribute_e_max}.
	We prove that, for any fixed ``non-extreme allocation'' $(\varepsilon_{1,[1]},\varepsilon_{1,[2]})$ (\textit{i.e.}, $\varepsilon_{1,[2]}>0$), there is a ``more extreme'' allocation $(\varepsilon_{2,[1]},\varepsilon_{2,[2]})$ that can make the objective function maximized. Then we fix $(\varepsilon_{2,[1]},\varepsilon_{2,[2]})$ and find a more extreme allocation $(\varepsilon'_{1,[1]},\varepsilon'_{1,[2]})$ to lift the objective function further. Using these iterations, we can find an infinite sequence of allocations which can make the objective function reach its supremum when the allocation is an extreme one.

	
	%
	%
	%
	%
\end{enumerate}


\textbf{Notations.}
Before the proof, we define the following helper functions based on $f(x)=x-\log x$.
\begin{align}\label{equ:flfr}
	f_l(x) = f(1-x)-1\ (0\leq x<1),\ 
	f_r(x)= f(x+1)-1\ (x\geq 0)
\end{align}
The derivatives of $f_l(x), f_r(x)$ are
\begin{align}
	f'_l(x)= -f'(1-x)= \dfrac{1}{1-x}-1 \label{equ:deriv_flr},\ 
	f'_r(x)=f'(1+x) =1-\dfrac{1}{x+1}  
\end{align}
So both $f_l(x)$ and $f_r(x)$ are strictly  increasing.
We note the inverse functions of  $f_l, f_r$ as $g_l, g_r$, respectively. Combining Lemma \ref{thm:reverse_f_def}, it is not hard to verify that $g_l, g_r$ are
\begin{align}
	g_l(\varepsilon) =& f_l^{-1}(\varepsilon) = 1-w_1(\varepsilon) = 1+W_{0}(-e^{-(1+\varepsilon)})\ (\varepsilon\geq 0) \\
	g_r(\varepsilon) =& f_r^{-1}(\varepsilon) = w_2(\varepsilon)-1= -W_{-1}(-e^{-(1+\varepsilon)})-1\ (\varepsilon\geq 0) \label{equ:def_f_r_inv}
\end{align}
According to Lemma \ref{thm:deriv_w1_w2}, the derivatives of $g_l, g_r$ are 
\begin{align}
	g_l'(\varepsilon)=f_l^{-1\prime}(\varepsilon) = \dfrac{w_1(\varepsilon)}{1-w_1(\varepsilon)} =\dfrac{1-f_l^{-1}(\varepsilon)}{f_l^{-1}(\varepsilon)}= \dfrac{1}{1-w_1(\varepsilon)}-1\label{equ:deriv_f_l_inv}\\
	g_r'(\varepsilon)=f_r^{-1\prime}(\varepsilon) = \dfrac{w_2(\varepsilon)}{w_2(\varepsilon)-1} =\dfrac{f_r^{-1}(\varepsilon)+1}{f_r^{-1}(\varepsilon)}= 1+\dfrac{1}{w_2(\varepsilon)-1} \label{equ:deriv_f_r_inv}
\end{align}
Specially, since $\lim\limits_{\varepsilon\to 0}w_2(\varepsilon)=w_2(0)=1$, it is easy to know 
\begin{align}\label{equ:limits_f_r_inv_e_to_0}
	\lim\limits_{\varepsilon\to 0}g_r'(\varepsilon)= +\infty
\end{align}
In the following, we note $g_r^{\prime}(0)=+\infty$ for convenience.

Lemma \ref{thm:simple_relation_ab} gives two useful conclusions for subsequent analysis. They hold apparently.
\begin{lemma}\label{thm:simple_relation_ab}
	Let  $a, b, a^+, b^-$ be positive real numbers.
	\begin{enumerate}[label=(\alph*), ref=\ref{thm:simple_relation_ab}\alph*]
		\item \label{thm:lemma_a1_div_b1}
		if $a>b, a<a^+, b>b^-$, then $\frac{a+1}{b+1}<\frac{a^++1}{b^-+1}$.
		\item \label{thm:lemma_a_div_b_b_plus_div_a_plus}
		if $a\leq b$, then $\frac{a(b+1)}{b(a+1)}\leq 1$.
	\end{enumerate}
	
\end{lemma}
%
%

\begin{lemma}\label{thm:w2_farther_from_1_than_w1}
	Given $f(x)=x-\log x$ and $\varepsilon\geq 0$, 
	then 
	$w_2(\varepsilon)-1\geq 1-w_1(\varepsilon)$ holds and the inequality is tight when $\varepsilon=0$;
\end{lemma}
\begin{proof}
	The details of the proof are shown in Appendix \ref{sec:app_proof_thm:w2_farther_from_1_than_w1}.

$\hfill\square$ 
\end{proof}

\begin{lemma}\label{thm:sup_fxy}
	Given $f(x)=x-\log x$ and $\varepsilon_x, \varepsilon_y\geq 0$, 
	if $f(x)\leq 1+\varepsilon_x$ and $f(y)\leq 1+ \varepsilon_y$, then 
	\begin{align}
		f(xy)\leq f(w_2(\varepsilon_x)w_2(\varepsilon_y))
	\end{align}
\end{lemma}
\begin{proof}
	The details of the proof are shown in Appendix \ref{sec:app_proof_thm:sup_fxy}.

$\hfill\square$ 
\end{proof}
In Lemma \ref{thm:lemma_sup_fx_invs_nd} in the last section (and Lemma \ref{thm:lemma_inf_fx_inv_nd} in Section \ref{app:big_KL_general_proof1} in Supplementary material), we eliminate the dimension $n$ from the bound by showing the convexity of  constructed function. Unfortunately, the relaxed triangle inequality involves three Gaussians which make the analysis more complex. 
The following Lemma \ref{thm:kl_triangle_fw2w2_extreme_distribute_e_max} is the core of proof of the relaxed triangle inequality theorem. It is the most techinical part in this paper.  
We will use Lemma \ref{thm:kl_triangle_fw2w2_extreme_distribute_e_max} to make the bound in Theorem \ref{thm:triangle_n1_n2_n3} independent of the dimension $n$.

\begin{lemma}\label{thm:kl_triangle_fw2w2_extreme_distribute_e_max}
	Given $f(x)=x-\log x$, 
	let $\varepsilon_{x,1},\varepsilon_{x,2},\varepsilon_{y,1},\varepsilon_{y,2}$ be four non-negative numbers such that $\varepsilon_{x,1}\geq \varepsilon_{x,2}, \varepsilon_{y,1}\geq \varepsilon_{y,2}$. 
	Then 
	\begin{align}\label{equ:f_w2x_w2y_max_target}
		&f(w_2(\varepsilon_{x,1})w_2(\varepsilon_{y,1}))+f(w_2(\varepsilon_{x,2})w_2(\varepsilon_{y,2}))\nonumber\\
		\leq &f(w_2(\varepsilon_{x,1}+\varepsilon_{x,2})w_2(\varepsilon_{y,1}+\varepsilon_{y,2}))+1 
	\end{align}
	
\end{lemma}

\noindent\textbf{Overview of proof of Lemma \ref{thm:kl_triangle_fw2w2_extreme_distribute_e_max}}\\
In the overview of proof of Theorem \ref{thm:triangle_n1_n2_n3} in the beginning of Section \ref{sec:quasi_triangle}, we have mentioned Lemma  \ref{thm:kl_triangle_fw2w2_extreme_distribute_e_max}. 
In the left hand side of Inequality \eqref{equ:f_w2x_w2y_max_target}, $\varepsilon_{x,2}$ and $\varepsilon_{y,2}$ stay in the second term. Intuitively, we use Inequality \eqref{equ:f_w2x_w2y_max_target}  to move $\varepsilon_{x,2}, \varepsilon_{y,2}$ into the first item. It is hard to prove Inequality \eqref{equ:f_w2x_w2y_max_target} directly due to the lack of conclusions relating to Lambert $W$ function. In the proof,  We use an iterative way to absorb $\varepsilon_{x,2}, \varepsilon_{y,2}$  into the first term gradually. 

We treat $$(\varepsilon_{x,1}+\theta_x\varepsilon_{x,2},\ \varepsilon_{x,2}-\theta_x\varepsilon_{x,2})\ \text{and}\
(\varepsilon_{y,1}+\theta_y\varepsilon_{y,2},\ \varepsilon_{y,2}-\theta_y\varepsilon_{y,2})$$ 
as two allocations, where $\theta_x$ and $\theta_y$ control how $\varepsilon_{x,2}$ and $\varepsilon_{y,2}$ are allocated among the two terms. The whole proof can be seen as an variation of coordinate ascent. In each iteration, we fix one of $\theta_x$ and $\theta_y$ (\textit{i.e.}, one allocation) and make another one vary. The goal is to maximize the objective function (Equation \eqref{equ:def_S}). In this way, we will construct an infinite sequence of allocations. The procedure is much harder than a simple coordinate ascent algorithm. 
The proof mainly consists of the following four aspects.
\begin{enumerate}[itemindent=1em, label={\textbf{A\arabic*}}]
	\setlength{\itemsep}{0pt}
	\setlength{\parsep}{0pt}
	\setlength{\parskip}{0pt}
	
	\item \label{itm:A1} In each step, once we fix one allocation and make another one vary, we prove there exists one and only one supremum.
	\item \label{itm:A2} We find an equation to express above supremum implicitly.
	\item \label{itm:A3} We prove the procedure is really lifting the objective function.
	\item \label{itm:A4} We construct an infinit sequence of allocations. Then we prove the limit of the allocation sequence will make the objective function reach its supremum.
\end{enumerate}

In this procedure, the hardest part is how to find a more extreme allocation based on the last one. There is no analytical solution to express these allocations. Luckily, we find a key equation to express the property of these allocations implicitly (see Equations \eqref{equ:from_thex_posi_to_they_posi}, \eqref{equ:deduction_ratio_xi2}, \eqref{equ:deduction_ratio_yi2}). Based on our analysis on such equation, we succeed to construct a sequence of allocations and finally prove Lemma \ref{thm:kl_triangle_fw2w2_extreme_distribute_e_max}.

\begin{proof}
	Inequality \eqref{equ:f_w2x_w2y_max_target} is equal to 
	\begin{align}
		&f(w_2(\varepsilon_{x,1})w_2(\varepsilon_{y,1}))+f(w_2(\varepsilon_{x,2})w_2(\varepsilon_{y,2}))\nonumber\\
		\leq& f(w_2(\varepsilon_{x,1}+\varepsilon_{x,2})w_2(\varepsilon_{y,1}+\varepsilon_{y,2}))+f(w_2(0)w_2(0)) 
	\end{align}
	We define function 
	\begin{align}\label{equ:def_S}
		S(\theta_x, \theta_y)
		=&f(w_2(\varepsilon_{x,1}+\theta_x\varepsilon_{x,2})w_2(\varepsilon_{y,1}+\theta_y\varepsilon_{y,2}))\nonumber\\
		&+f(w_2(\varepsilon_{x,2}-\theta_x\varepsilon_{x,2})w_2(\varepsilon_{y,2}-\theta_y\varepsilon_{y,2}))
	\end{align}
	for $-\frac{\varepsilon_{x,1}}{\varepsilon_{x,2}}\leq \theta_x\leq 1,-\frac{\varepsilon_{y,1}}{\varepsilon_{y,2}}\leq \theta_y\leq 1 $.
	The domains of $\theta_x, \theta_y$ are restricted to make $w_2(\cdot)$ in the definition of $S(\theta_x,\theta_y)$ meaningful.
	Inequation \eqref{equ:f_w2x_w2y_max_target} states that $S(0,0)\leq S(1,1)$.
	
	We can prove $S(0,0)\leq S(1,1)$ in the following three cases.
	\begin{enumerate}[itemindent=2.9em, label={\textbf{Case \arabic*}}]
		\item  \label{itm:case1} $\varepsilon_{x,2}=\varepsilon_{y,2}=0$. 
		\item  \label{itm:case2} $\varepsilon_{x,2}>0,\varepsilon_{y,2}>0$. In this case, we have $\varepsilon_{x,1}\geq \varepsilon_{x,2}>0, \varepsilon_{y,1}\geq \varepsilon_{y,2}>0$. 
		\item    \label{itm:case3} only one of $\varepsilon_{x,2}$ and $\varepsilon_{y,2}$ equals to 0.
	\end{enumerate} 
	
	It is easy to verify that $S(0,0)=S(1,1)$ for \ref{itm:case1}. In the following, we discuss \ref{itm:case2} first and deal with \ref{itm:case3} at the end of the proof.

	\noindent\ref{itm:case2}:
	
	In $S(\theta_x,\theta_y)$, $\theta_x,\theta_y$ are symmetric. Without loss of generality, we choose any $0<\theta_{x,0}<1$ at the beginning.
	The following proof consists of two steps. In \textbf{Step 1}, we  prove that 
	for any fixed $0<\theta_{x,0}<1$, there exists one and only one $-\frac{\varepsilon_{y,1}}{\varepsilon_{y,2}}< \theta_{y,1}<1$ such that $S(\theta_{x,0},\theta_y)$ takes its maximum. This accomplishes aspects \ref{itm:A1} and \ref{itm:A2} in the first iteration.  In \textbf{Step 2}, we prove $S(1,1)\geq S(0,0)$. The key idea is finding a strictly   increasing sequence $\{S[i]\}$ such that $S[0]$ can be arbitrarily close to $S(0,0)$ and $\lim\limits_{i\to\infty}S[i]=S(1,1)$. \textbf{Step 2} will accomplish aspects \ref{itm:A1} $\sim$ \ref{itm:A4} in all iterations.
	

	\textbf{Step 1.}
	At the beginning, we select any $0<\theta_{x,0}<1$. 
	For brevity, we note 
	\begin{align}\label{equ:def_tiled_exy_12}
		\tilde{\varepsilon}_{x,1}[0]=\varepsilon_{x,1}+\theta_{x,0}\varepsilon_{x,2},\ 
		\tilde{\varepsilon}_{x,2}[0]=\varepsilon_{x,2}-\theta_{x,0}\varepsilon_{x,2}\\
		\tilde{\varepsilon}_{y,1}=\varepsilon_{y,1}+\theta_y\varepsilon_{y,2},\ 
		\tilde{\varepsilon}_{y,2}=\varepsilon_{y,2}-\theta_y\varepsilon_{y,2}
	\end{align}
	%
	%
	where we use $\tilde{\varepsilon}_{x,(\cdot)}[0]$ to denote the variable is computed with $\theta_{x,0}$.

	Note that $g_r(\varepsilon)$ (defined in Equation \eqref{equ:def_f_r_inv}) is strictly  increasing with $\varepsilon$. Combining  the precondition $\varepsilon_{x,1}\geq \varepsilon_{x,2}$, we can know 
	\begin{align}
		\dfrac{g_r(\tilde{\varepsilon}_{x,1}[0])}{g_r(\tilde{\varepsilon}_{x,2}[0])}=\dfrac{g_r(\varepsilon_{x,1}+\theta_{x,0}\varepsilon_{x,2})}{g_r(\varepsilon_{x,2}-\theta_{x,0}\varepsilon_{x,2})}>\dfrac{g_r(\varepsilon_{x,1})}{g_r(\varepsilon_{x,2})}\geq 1
	\end{align}
	We note this condition as $\bm{C_1}[0]$ as follows. 
	\begin{align}
		\bm{C_1}[0]:& \dfrac{g_r(\tilde{\varepsilon}_{x,1}[0])}{g_r(\tilde{\varepsilon}_{x,2}[0])}> 1 \label{equ:start_from_thetax_0_ratio_fx1x2_larger_than_1}
	\end{align}
	Now given the fixed $\theta_{x,0}$, the derivative of $S(\theta_{x,0},\theta_y)$ is
	\begin{align}
		&\dfrac{\dif S(\theta_{x,0},\theta_y)}{\dif \theta_y}\nonumber\\
		= & \varepsilon_{y,2}\left(w_2(\tilde{\varepsilon}_{x,1}[0])\dfrac{w_2(\tilde{\varepsilon}_{y,1})}{w_2(\tilde{\varepsilon}_{y,1})-1}-\dfrac{1}{w_2(\tilde{\varepsilon}_{y,1})}\dfrac{w_2(\tilde{\varepsilon}_{y,1})}{w_2(\tilde{\varepsilon}_{y,1})-1}\right)\nonumber \\
		& -\varepsilon_{y,2}\left( w_2(\tilde{\varepsilon}_{x,2}[0])\dfrac{w_2(\tilde{\varepsilon}_{y,2})}{w_2(\tilde{\varepsilon}_{y,2})-1}-\dfrac{1}{w_2(\tilde{\varepsilon}_{y,2})}\dfrac{w_2(\tilde{\varepsilon}_{y,2})}{w_2(\tilde{\varepsilon}_{y,2})-1}\right)\nonumber\\
		= & \varepsilon_{y,2}\left(\dfrac{w_2(\tilde{\varepsilon}_{x,1}[0])w_2(\tilde{\varepsilon}_{y,1})-1}{w_2(\tilde{\varepsilon}_{y,1})-1}- \dfrac{w_2(\tilde{\varepsilon}_{x,2}[0])w_2(\tilde{\varepsilon}_{y,2})-1}{w_2(\tilde{\varepsilon}_{y,2})-1}\right)\nonumber\\
		= & \varepsilon_{y,2}\left(\dfrac{w_2(\tilde{\varepsilon}_{x,1}[0])w_2(\tilde{\varepsilon}_{y,1})-w_2(\tilde{\varepsilon}_{x,1}[0])+w_2(\tilde{\varepsilon}_{x,1}[0])-1}{w_2(\tilde{\varepsilon}_{y,1})-1}\right.\nonumber\\
		& \left. - \dfrac{w_2(\tilde{\varepsilon}_{x,2}[0])w_2(\tilde{\varepsilon}_{y,2})-w_2(\tilde{\varepsilon}_{x,2}[0])+w_2(\tilde{\varepsilon}_{x,2}[0])-1}{w_2(\tilde{\varepsilon}_{y,2})-1}\right)\nonumber\\ 
		= &  \varepsilon_{y,2}\left(\left(w_2(\tilde{\varepsilon}_{x,1}[0]) + \dfrac{w_2(\tilde{\varepsilon}_{x,1}[0])-1}{w_2(\tilde{\varepsilon}_{y,1})-1}\right) \right. \nonumber\\
		& \left. -\left(w_2(\tilde{\varepsilon}_{x,2}[0])+\dfrac{w_2(\tilde{\varepsilon}_{x,2}[0])-1}{w_2(\tilde{\varepsilon}_{y,2})-1}   \right) \right)\label{equ:deriv_S_thex0_they}
	\end{align}
	The second order derivative is 
	\begin{align}
		& \dfrac{\dif^2 S(\theta_{x,0},\theta_y)}{\dif \theta_y^2}\nonumber \\
		= & \dfrac{-(w_2(\tilde{\varepsilon}_{x,1}[0])-1)}{(w_2(\tilde{\varepsilon}_{y,1})-1)^2}\dfrac{w_2(\tilde{\varepsilon}_{y,1})}{w_2(\tilde{\varepsilon}_{y,1})-1}(\varepsilon_{y,2})^2 \nonumber\\
		& -\dfrac{-(w_2(\tilde{\varepsilon}_{x,2}[0])-1)}{(w_2(\tilde{\varepsilon}_{y,2})-1)^2}\dfrac{w_2(\tilde{\varepsilon}_{y,2})}{w_2(\tilde{\varepsilon}_{y,2})-1}(-(\varepsilon_{y,2})^2) \nonumber \\
		= & -\dfrac{(w_2(\tilde{\varepsilon}_{x,1}[0])-1)w_2(\tilde{\varepsilon}_{y,1})(\varepsilon_{y,2})^2}{(w_2(\tilde{\varepsilon}_{y,1})-1)^3} \nonumber\\
		&-\dfrac{(w_2(\tilde{\varepsilon}_{x,2}[0])-1)w_2(\tilde{\varepsilon}_{y,2})(\varepsilon_{y,2})^2}{(w_2(\tilde{\varepsilon}_{y,2})-1)^3}
	\end{align}
	Since $w_2(\varepsilon)>1$ for $\varepsilon>0$, it is easy to know  $\frac{\dif^2 S(\theta_{x,0},\theta_y)}{\dif \theta_y^2}<0$ for $\theta_y<1$. Thus we get the following proposition.
	
	\begin{proposition}\label{thm:proposition_S(thetax0_y)_concave}
		$S(\theta_{x,0},\theta_y)$ is strictly  concave and has at most one maximum for $\theta_y<1$.
	\end{proposition}

	Remember that we are discussing \ref{itm:case2}, so $\varepsilon_{y,2}>0$. Now letting $\frac{\dif S(\theta_{x,0},\theta_y)}{\dif \theta_y} =0$ $($\textit{i.e.}, Equation \eqref{equ:deriv_S_thex0_they}$=0$$)$, we can obtain
	\begin{align}
		&\dfrac{\dif S(\theta_{x,0},\theta_y)}{\dif \theta_y} =0 \Leftrightarrow \nonumber\\
		& w_2(\tilde{\varepsilon}_{x,1}[0]) + \dfrac{w_2(\tilde{\varepsilon}_{x,1}[0])-1}{w_2(\tilde{\varepsilon}_{y,1})-1}
		= w_2(\tilde{\varepsilon}_{x,2}[0])+\dfrac{w_2(\tilde{\varepsilon}_{x,2}[0])-1}{w_2(\tilde{\varepsilon}_{y,2})-1} \nonumber\\  
		\label{equ:fix_thetax_find_thetay_equal}
	\end{align}
	Now, it seems that the proof is stuck because we can not solve Equation \eqref{equ:fix_thetax_find_thetay_equal} analytically. However, we succeed to go further by analyzing Equation \eqref{equ:fix_thetax_find_thetay_equal}. Our analysis starts from the following transformation in Equations \eqref{equ:start_transform_dS=0} $\sim$ \eqref{equ:from_thex_posi_to_they_posi}, which is hard to obtain but easy to verify.
	
	Using the notations of helper functions $g_r(\varepsilon)=f_r^{-1}(\varepsilon), g_r'(\varepsilon)=f_r^{-1\prime}(\varepsilon)$ in Equations  \eqref{equ:def_f_r_inv} and \eqref{equ:deriv_f_r_inv},
	we can rewrite Equation \eqref{equ:fix_thetax_find_thetay_equal} as follows.
	\begin{align}
		& \text{Equation \eqref{equ:fix_thetax_find_thetay_equal}} \nonumber\\
		\Leftrightarrow & w_2(\tilde{\varepsilon}_{x,1}[0])-1 + \dfrac{w_2(\tilde{\varepsilon}_{x,1}[0])-1}{w_2(\tilde{\varepsilon}_{y,1})-1}\nonumber\\
		&= w_2(\tilde{\varepsilon}_{x,2}[0])-1+\dfrac{w_2(\tilde{\varepsilon}_{x,2}[0])-1}{w_2(\tilde{\varepsilon}_{y,2})-1} \label{equ:start_transform_dS=0}\\
		\Leftrightarrow & (w_2(\tilde{\varepsilon}_{x,1}[0])-1)\left(1+\dfrac{1}{w_2(\tilde{\varepsilon}_{y,1})-1}\right)\nonumber\\
		&= (w_2(\tilde{\varepsilon}_{x,2}[0])-1)\left(1+\dfrac{1}{w_2(\tilde{\varepsilon}_{y,2})-1}\right)\nonumber\\
		\Leftrightarrow &g_r(\tilde{\varepsilon}_{x,1}[0]) g_r'(\tilde{\varepsilon}_{y,1})=g_r(\tilde{\varepsilon}_{x,2}[0])g_r'(\tilde{\varepsilon}_{y,2})\nonumber\\
		\Leftrightarrow & \dfrac{g_r(\tilde{\varepsilon}_{x,1}[0])}{g_r(\tilde{\varepsilon}_{x,2}[0])}=
		\dfrac{g_r'(\tilde{\varepsilon}_{y,2})}{g_r'(\tilde{\varepsilon}_{y,1})}  \label{equ:ratio_f_r_inv_x_equals_f_r_inv_der_y}\\
		\Leftrightarrow & \dfrac{g_r(\tilde{\varepsilon}_{x,1}[0])}{g_r(\tilde{\varepsilon}_{x,2}[0])}=
		\dfrac{\left(\dfrac{1}{g_r'(\tilde{\varepsilon}_{y,1})}\right)}{\left(\dfrac{1}{g_r'(\tilde{\varepsilon}_{y,2})}\right)} \nonumber\\
		\Leftrightarrow &  \dfrac{g_r(\tilde{\varepsilon}_{x,1}[0])}{g_r(\tilde{\varepsilon}_{x,2}[0])}=
		\dfrac{\left(\dfrac{g_r(\tilde{\varepsilon}_{y,1})}{g_r(\tilde{\varepsilon}_{y,1})+1}\right)}{\left(\dfrac{g_r(\tilde{\varepsilon}_{y,2})}{g_r(\tilde{\varepsilon}_{y,2})+1}\right)} \label{equ:fourth_line_derive_0_transformation}\\
		\Leftrightarrow & \dfrac{g_r(\tilde{\varepsilon}_{x,1}[0])}{g_r(\tilde{\varepsilon}_{x,2}[0])}=
		\dfrac{g_r(\tilde{\varepsilon}_{y,1})}{g_r(\tilde{\varepsilon}_{y,2})}\dfrac{g_r(\tilde{\varepsilon}_{y,2})+1}{g_r(\tilde{\varepsilon}_{y,1})+1} \label{equ:ratio_fx1fx2_left_fy1yx_right}\\
		\Leftrightarrow & \dfrac{g_r(\tilde{\varepsilon}_{y,1})}{g_r(\tilde{\varepsilon}_{y,2})}=\dfrac{g_r(\tilde{\varepsilon}_{x,1}[0])}{g_r(\tilde{\varepsilon}_{x,2}[0])}\dfrac{g_r(\tilde{\varepsilon}_{y,1})+1}{g_r(\tilde{\varepsilon}_{y,2})+1}\label{equ:from_thex_posi_to_they_posi}
	\end{align}
	where Equation \eqref{equ:fourth_line_derive_0_transformation} follows from Equation \eqref{equ:deriv_f_r_inv}. 
	
	Up to now, we transform the condition $\frac{\dif S(\theta_{x,0},\theta_y)}{\dif \theta_y} =0$ in Equation \eqref{equ:fix_thetax_find_thetay_equal} to Equation \eqref{equ:from_thex_posi_to_they_posi}. In the following, Equation \eqref{equ:from_thex_posi_to_they_posi} will be used to investigate the property of the maximum for $S(\theta_{x,0},\theta_y)$. The goal is to accomplish aspect \ref{itm:A2} of the proof.
	
	In the following \textbf{Substep 1.1}, we show that Equation \eqref{equ:fix_thetax_find_thetay_equal} must have one and only one solution. In other words, there must be one and only one point making $\frac{\dif S(\theta_{x,0},\theta_y)}{\dif \theta_y} =0$. Unfortunately, it is hard to get an analytical solution from Equation \eqref{equ:fix_thetax_find_thetay_equal} due to the complexity brought by Lambert $W$ function. Therefore, in \textbf{Substep 1.2}, we analyze Equations \eqref{equ:start_transform_dS=0} $\sim$ \eqref{equ:from_thex_posi_to_they_posi} to investigate the properties of the solution. Overall, the analysis in \textbf{Step 1} will be used as a basic step in \textbf{Step 2}.

	\textbf{	Substep 1.1.} 
	According to the definition of $g_r'(\varepsilon)$ in Equation \eqref{equ:deriv_f_r_inv}, 
	$g_r'(\varepsilon)$ is strictly decreasing with $\varepsilon$. So $g_r'(\tilde{\varepsilon}_{y,2})=g_r'(\varepsilon_{y,2}-\theta_y\varepsilon_{y,2})$ is strictly  increasing and $g_r'(\tilde{\varepsilon}_{y,1})=g_r'(\varepsilon_{y,1}+\theta_y\varepsilon_{y,2})$ is strictly decreasing with $\theta_y$. Thus, the right hand side of Equation \eqref{equ:ratio_f_r_inv_x_equals_f_r_inv_der_y} is continous and strictly  increasing with $\theta_y$. Besides, according to Equation \eqref{equ:limits_f_r_inv_e_to_0} and the definition of $\tilde{\varepsilon}_{y,1},\tilde{\varepsilon}_{y,2}$ in Equation \eqref{equ:def_tiled_exy_12} and \eqref{equ:def_tiled_exy_12}, it is easy to know 
	\begin{align}
		\lim_{\theta_y\to -\frac{\varepsilon_{y,1}}{\varepsilon_{y,2}}}\dfrac{g_r'(\tilde{\varepsilon}_{y,2})}{g_r'(\tilde{\varepsilon}_{y,1})}
		= & \lim_{\theta_y\to -\frac{\varepsilon_{y,1}}{\varepsilon_{y,2}}}\dfrac{g_r'(\varepsilon_{y,2}-\theta_y\varepsilon_{y,2})}{g_r'(\varepsilon_{y,1}+\theta_y\varepsilon_{y,2})}\nonumber\\
		= &\dfrac{g_r'(\varepsilon_{y,2}+\varepsilon_{y,1})}{g_r'(0)}
		=  \dfrac{g_r'(\varepsilon_{y,2}+\varepsilon_{y,1})}{+\infty}
		=  0\\
		\lim_{\theta_y\to 1}\dfrac{g_r'(\tilde{\varepsilon}_{y,2})}{g_r'(\tilde{\varepsilon}_{y,1})}
		= & \lim_{\theta_y\to 1}\dfrac{g_r'(\varepsilon_{y,2}-\theta_y\varepsilon_{y,2})}{g_r'(\varepsilon_{y,1}+\theta_y\varepsilon_{y,2})} \nonumber\\
		= &\dfrac{g_r'(0)}{g_r'(\varepsilon_{y,1}+\varepsilon_{y,2})}
		=\dfrac{+\infty}{g_r'(\varepsilon_{y,1}+\varepsilon_{y,2})}
		= +\infty
	\end{align}
	So the range of the right hand side of Equation \eqref{equ:ratio_f_r_inv_x_equals_f_r_inv_der_y} is $(0,+\infty)$ when $-\frac{\varepsilon_{y,1}}{\varepsilon_{y,2}}< \theta_y< 1$.
	
	Remember that we start from $0<\theta_{x,0}<1$, combining the precondition $\varepsilon_{x,1}\geq \varepsilon_{x,2}$ and the definitions of $\tilde{\varepsilon}_{x,1}[0], \tilde{\varepsilon}_{x,2}[0]$ in Equation \eqref{equ:def_tiled_exy_12} and \eqref{equ:def_tiled_exy_12}, it is easy to know that the left hand side of Equation \eqref{equ:ratio_f_r_inv_x_equals_f_r_inv_der_y} is a positive constant number. 
	Therefore, Equation \eqref{equ:ratio_f_r_inv_x_equals_f_r_inv_der_y} must has one and only one solution.
	We note such solution  as $\theta_{y,1}$.	 
	Combining with Proposition \ref{thm:proposition_S(thetax0_y)_concave}, we can know that for any fixed $0<\theta_{x,0}<1$, there exists one and only one $-\frac{\varepsilon_{y,1}}{\varepsilon_{y,2}}< \theta_{y,1}<1$ that maximize $S(\theta_{x,0},\theta_y)$.
	
	Here note that we still have no guarantee for $\theta_{y,1}>0$ currently.

	\textbf{	Substep 1.2.}
	We can investigate the property of the solution $\theta_{y,1}$ by analyzing Equation \eqref{equ:ratio_f_r_inv_x_equals_f_r_inv_der_y}, \eqref{equ:ratio_fx1fx2_left_fy1yx_right} and \eqref{equ:from_thex_posi_to_they_posi}.
	
	Firstly, we can show 
	\begin{align}
		\dfrac{g_r(\tilde{\varepsilon}_{y,1}[1])}{g_r(\tilde{\varepsilon}_{y,2}[1])}=\dfrac{g_r(\varepsilon_{y,1}+\theta_{y,1}\varepsilon_{y,2})}{g_r(\varepsilon_{y,2}-\theta_{y,1}\varepsilon_{y,2})}>1
	\end{align} by contradiction. 
	Assume to the contrary that $\frac{g_r(\tilde{\varepsilon}_{y,1}[1])}{g_r(\tilde{\varepsilon}_{y,2}[1])}\leq 1$, then $\frac{g_r(\tilde{\varepsilon}_{y,1}[1])+1}{g_r(\tilde{\varepsilon}_{y,2}[1])+1}\geq \frac{g_r(\tilde{\varepsilon}_{y,1}[1])}{g_r(\tilde{\varepsilon}_{y,2}[1])}$. Combining condition $\bm{C_1}[0]$ in Equation \eqref{equ:start_from_thetax_0_ratio_fx1x2_larger_than_1}, we can deduce 
	\begin{align}
		\dfrac{g_r(\tilde{\varepsilon}_{x,1}[0])}{g_r(\tilde{\varepsilon}_{x,2}[0])}\dfrac{g_r(\tilde{\varepsilon}_{y,1}[1])+1}{g_r(\tilde{\varepsilon}_{y,2}[1])+1}
		>\dfrac{g_r(\tilde{\varepsilon}_{y,1}[1])+1}{g_r(\tilde{\varepsilon}_{y,2}[1])+1}
		\geq \dfrac{g_r(\tilde{\varepsilon}_{y,1}[1])}{g_r(\tilde{\varepsilon}_{y,2}[1])}
	\end{align}
	This contradicts with Equation \eqref{equ:from_thex_posi_to_they_posi}. 
	Therefore,  we can obtain the following condition $\bm{C_1}[1]$.
	\begin{align}\label{equ:condition_ratio_fey1_ey2_larger_than_1_for_next_iteration}
		\bm{C_1}[1]:& \dfrac{g_r(\tilde{\varepsilon}_{y,1}[1])}{g_r(\tilde{\varepsilon}_{y,2}[1])}>1
	\end{align}
	
	Secondly, according to condition $\bm{C_1}[1]$, it is easy to know
	\begin{align}\label{equ:ratio_fey1[1]_larger_fey1[1]plus}
		\dfrac{g_r(\tilde{\varepsilon}_{y,1}[1])}{g_r(\tilde{\varepsilon}_{y,2}[1])}>\dfrac{g_r(\tilde{\varepsilon}_{y,1}[1])+1}{g_r(\tilde{\varepsilon}_{y,2}[1])+1}>1
	\end{align}
	Now combining Equation \eqref{equ:from_thex_posi_to_they_posi} and \eqref{equ:ratio_fey1[1]_larger_fey1[1]plus}, we can know the following condition\footnote{In this context, $\bm{C_2}[1]$ is stronger than $\bm{C_1}[1]$. We separate $\bm{C_1}[1]$ and $\bm{C_2}[1]$ away for clarity.} $\bm{C_2}[1]$ holds. 
	\begin{align}\label{equ:ratio_fey1_larger_fex1}
		\bm{C_2}[1]: \dfrac{g_r(\tilde{\varepsilon}_{y,1}[1])}{g_r(\tilde{\varepsilon}_{y,2}[1])}>\dfrac{g_r(\tilde{\varepsilon}_{x,1}[0])}{g_r(\tilde{\varepsilon}_{x,2}[0])}
	\end{align} 
	
	In summary, we start from  any fixed $0<\theta_{x,0}<1$ making condition $\bm{C_1}[0]$ in Equation \eqref{equ:start_from_thetax_0_ratio_fx1x2_larger_than_1} hold. Using \textbf{Step 1},  we find the only one $-\frac{\varepsilon_{y,1}}{\varepsilon_{y,2}}< \theta_{y,1}< 1$ such that $S(\theta_{x,0},\theta_y)$ takes its maximum at $\theta_{y,1}$ and conditions $\bm{C_1}[1]$ and $\bm{C_2}[1]$ hold. 

	\textbf{Step 2.} The deduction in \textbf{Step 1} can be iterated repeatedly due to the symmetry of $\theta_x$ and $\theta_y$ in $S(\theta_x, \theta_y)$. 
	For consistency, at the begining of the iterations, we can choose any $\theta_y\neq \theta_{y,1}$ as $\theta_{y,0}$. 
	
	In the following, we use notations
	\begin{align}\label{equ:def_sim_exy_12_i}
		\tilde{\varepsilon}_{x,1}[i]=\varepsilon_{x,1}+\theta_{x,i}\varepsilon_{x,2},\ 
		\tilde{\varepsilon}_{x,2}[i]=\varepsilon_{x,2}-\theta_{x,i}\varepsilon_{x,2}\nonumber\\
		\tilde{\varepsilon}_{y,1}[i]=\varepsilon_{y,1}+\theta_{y,i}\varepsilon_{y,2},\ 
		\tilde{\varepsilon}_{y,2}[i]=\varepsilon_{y,2}-\theta_{y,i}\varepsilon_{y,2}
	\end{align}
	where $\tilde{\varepsilon}_{(\cdot),k}[i]\ (k\in\{1,2\})$ is computed with $\theta_{(\cdot),i}$.
	Now we fix $\theta_{y,1}$ and make $\theta_x$ vary, then we repeat \textbf{Step 1} on $\theta_x$. Note that, in the second iteration the condition $\bm{C_1}[1]$  plays the same role as condition $\bm{C_1}[0]$ plays in the first iteration.
	Therefore, we can find a $\theta_{x,2}$ such that the following conditions hold
	\begin{align}
		&S(\theta_{x,2},\theta_{y,1})>S(\theta_{x,0},\theta_{y,1})\\
		&\dfrac{g_r(\tilde{\varepsilon}_{x,1}[2])}{g_r(\tilde{\varepsilon}_{x,2}[2])}
		=\dfrac{g_r(\tilde{\varepsilon}_{y,1}[1])}{g_r(\tilde{\varepsilon}_{y,2}[1])}\dfrac{g_r(\tilde{\varepsilon}_{x,1}[2])+1}{g_r(\tilde{\varepsilon}_{x,2}[2])+1} \nonumber\\
		&= \dfrac{g_r(\tilde{\varepsilon}_{x,1}[0])}{g_r(\tilde{\varepsilon}_{x,2}[0])}\dfrac{g_r(\tilde{\varepsilon}_{y,1}[1])+1}{g_r(\tilde{\varepsilon}_{y,2}[1])+1}\dfrac{g_r(\tilde{\varepsilon}_{x,1}[2])+1}{g_r(\tilde{\varepsilon}_{x,2}[2])+1} \label{equ:frac_fex1[2]_fex2[2]}\\
		&\bm{C_1}[2]: \dfrac{g_r(\tilde{\varepsilon}_{x,1}[2])}{g_r(\tilde{\varepsilon}_{x,2}[2])}>1\\
		&\bm{C_2}[2]:\dfrac{g_r(\tilde{\varepsilon}_{x,1}[2])}{g_r(\tilde{\varepsilon}_{x,2}[2])}>\dfrac{g_r(\tilde{\varepsilon}_{y,1}[1])}{g_r(\tilde{\varepsilon}_{y,2}[1])}
	\end{align}
	where the first equation is by Equation \eqref{equ:from_thex_posi_to_they_posi}.	
	Note that, combing conditions $\bm{C_1}[0]$, $\bm{C_1}[1]$, $\bm{C_2}[1]$, $\bm{C_1}[2]$ and Equation \eqref{equ:frac_fex1[2]_fex2[2]}, we know it is impossible that $\tilde{\varepsilon}_{x,1}[2]=\tilde{\varepsilon}_{x,1}[0]$ and $\theta_{x,2}=\theta_{x,0}$. So it is impossible $S(\theta_{x,2},\theta_{y,1})=S(\theta_{x,0},\theta_{y,1})$.
	
	We can repeat \textbf{Step 1} on $\theta_x$ and $\theta_y$ alternatively and construct a sequence $\{\theta_{x,0},\theta_{y,1},\theta_{x,2},\theta_{y,3},\dots$\} such that the following conditions hold.
	
	\noindent For $i\in \mathbb{N}$ we have 
	\begin{align}
		& S(\theta_{x,2i+2},\theta_{y,2i+1})>S(\theta_{x,2i},\theta_{y,2i+1})\\
		&\dfrac{g_r(\tilde{\varepsilon}_{x,1}[2i+2])}{g_r(\tilde{\varepsilon}_{x,2}[2i+2])}
		= \dfrac{g_r(\tilde{\varepsilon}_{y,1}[2i+1])}{g_r(\tilde{\varepsilon}_{y,2}[2i+1])}\dfrac{g_r(\tilde{\varepsilon}_{x,1}[2i+2])+1}{g_r(\tilde{\varepsilon}_{x,2}[2i+2])+1}\nonumber\\
		&= \dfrac{g_r(\tilde{\varepsilon}_{x,1}[2i])}{g_r(\tilde{\varepsilon}_{x,2}[2i])}\dfrac{g_r(\tilde{\varepsilon}_{y,1}[2i+1])+1}{g_r(\tilde{\varepsilon}_{y,2}[2i+1])+1}\dfrac{g_r(\tilde{\varepsilon}_{x,1}[2i+2])+1}{g_r(\tilde{\varepsilon}_{x,2}[2i+2])+1}\label{equ:deduction_ratio_xi2}\\
		&\bm{C_1}[2i+2]: \dfrac{g_r(\tilde{\varepsilon}_{x,1}[2i+2])}{g_r(\tilde{\varepsilon}_{x,2}[2i+2])}>1\\
		&\bm{C_2}[2i+2]:\dfrac{g_r(\tilde{\varepsilon}_{x,1}[2i+2])}{g_r(\tilde{\varepsilon}_{x,2}[2i+2])}>\dfrac{g_r(\tilde{\varepsilon}_{y,1}[2i+1])}{g_r(\tilde{\varepsilon}_{y,2}[2i+1])}\label{equ:condition_c_2_2i+2}
	\end{align}
	For $i\in \mathbb{N} \wedge i>0$ we have
	\begin{align}
		& S(\theta_{x,2i},\theta_{y,2i+1})>S(\theta_{x,2i},\theta_{y,2i-1})\\
		&\dfrac{g_r(\tilde{\varepsilon}_{y,1}[2i+1])}{g_r(\tilde{\varepsilon}_{y,2}[2i+1])}
		= \dfrac{g_r(\tilde{\varepsilon}_{x,1}[2i])}{g_r(\tilde{\varepsilon}_{x,2}[2i])}\dfrac{g_r(\tilde{\varepsilon}_{y,1}[2i+1])+1}{g_r(\tilde{\varepsilon}_{y,2}[2i+1])+1}\nonumber\\
		=& \dfrac{g_r(\tilde{\varepsilon}_{y,1}[2i-1])}{g_r(\tilde{\varepsilon}_{y,2}[2i-1])}\dfrac{g_r(\tilde{\varepsilon}_{x,1}[2i])+1}{g_r(\tilde{\varepsilon}_{x,2}[2i])+1}\dfrac{g_r(\tilde{\varepsilon}_{y,1}[2i+1])+1}{g_r(\tilde{\varepsilon}_{y,2}[2i+1])+1} \label{equ:deduction_ratio_yi2}\\
		&\bm{C_1}[2i+1]:\dfrac{g_r(\tilde{\varepsilon}_{y,1}[2i+1])}{g_r(\tilde{\varepsilon}_{y,2}[2i+1])}>1\\
		&\bm{C_2}[2i+1]:\dfrac{g_r(\tilde{\varepsilon}_{y,1}[2i+1])}{g_r(\tilde{\varepsilon}_{y,2}[2i+1])}>\dfrac{g_r(\tilde{\varepsilon}_{x,1}[2i])}{g_r(\tilde{\varepsilon}_{x,2}[2i])}\label{equ:condition_c_2_2i+1}
	\end{align}
	Now combining conditions $\bm{C_2}[2i+2]$, $\bm{C_2}[2i+1]$ and Equations \eqref{equ:def_sim_exy_12_i}, we can also obtain 
	\begin{align}
		\theta_{x,2i+2}>\theta_{x,2i},\ 
		\theta_{y,2i+3}>\theta_{y,2i+1}\ (i\in\mathbb{N})
	\end{align}		
	Up to now, we have constructed the following strictly  increasing sequences
	\begin{align}
		\Theta_x[i]&=\theta_{x,2i}\label{equ:Thetaxi}\\
		\Theta_y[i]&=\theta_{y,2i+1}\label{equ:Thetayi}\\
		R_x[i]&=\dfrac{g_r(\tilde{\varepsilon}_{x,1}[2i])}{g_r(\tilde{\varepsilon}_{x,2}[2i])}\\
		R_y[i]&=\dfrac{g_r(\tilde{\varepsilon}_{y,1}[2i+1])}{g_r(\tilde{\varepsilon}_{y,2}[2i+1])}
	\end{align}
	for $i\in \mathbb{N}$. According to conditions $\bm{C_1}[0],\bm{C_1}[1],\cdots$, it is easy to know $R_x[i]>1, R_y[i]>1$ for $i\in \mathbb{N}$. 
	Besides, we note 
	\begin{align}
		R_x^{+}[i]=\dfrac{g_r(\tilde{\varepsilon}_{x,1}[2i])+1}{g_r(\tilde{\varepsilon}_{x,2}[2i])+1}, \ 
		R_y^{+}[i]=\dfrac{g_r(\tilde{\varepsilon}_{y,1}[2i+1])+1}{g_r(\tilde{\varepsilon}_{y,2}[2i+1])+1}
	\end{align}
	for $i\in \mathbb{N}$.
	According to Lemma \ref{thm:lemma_a1_div_b1}, it is easy to know both $R_x^{+}[i], R_y^{+}[i]$ are strictly   increasing.
	Importantly, we have constructed the following strictly   increasing sequence for $i\in \mathbb{N}$.
	\begin{align}\label{equ:def_Si_seris}
		S[i]=\left\{
		\begin{aligned}
			S(\theta_{x,0},\theta_{y,0}) &,\ i=0 \\
			S(\theta_{x,i-1},\theta_{y,i}) &,\ i\%2=1\\
			S(\theta_{x,i},\theta_{y,i-1}) &,\ i\%2=0\wedge i>0
		\end{aligned}
		\right.
	\end{align}
	
	This accomplishes aspect \ref{itm:A3}.
	
	Here we note that $\bm{C_1}[i] (i\geq 0)$ plays an important role in each iteration. When $\bm{C_1}[i]$ holds, we let the derivative of $S$ equal to 0. Then we get the maximum of $S$ and make $\bm{C_2}[i] (i\geq 1)$ hold in each iteration. Importantly, $\bm{C_2}[i]$ guarantees  $R_x[i]$ and $R_y[i]$ are strictly increasing.
	
	In the following, we prove   
	\begin{align}
		\lim_{i \to +\infty}R_x[i]=+\infty,\ 
		\lim_{i \to +\infty}R_y[i]=+\infty
	\end{align}
	in order to accomplish aspect \ref{itm:A4} finally.
	Now let's observe how $R_x[i]$ increases. 
	Using the notations of $R_x[i]$ and $R_x^{+}[i]$, we rewrite Equation \eqref{equ:deduction_ratio_xi2} and get the following relation 
	\begin{align}
		R_x[i+1]&=R_x[i]R_y^{+}[i]R_x^{+}[i+1]
	\end{align}
	This indicates that 
	\begin{align}
		R_x[i+1]-R_x[i]&=R_x[i](R_y^{+}[i]R_x^{+}[i+1]-1) \label{equ:difference_Ri}
	\end{align}
	Here $R_x[i], R_y^{+}[i],R_x^{+}[i]$ are all strictly  increasing and larger than $1$. The relation in Equation \eqref{equ:difference_Ri} 
	indicates that the difference between neighbouring elements of $\{R_x[i]\}$ is strictly  increasing. This violates the  Cauchy's criterion for convergence. Thus, we can conclude  $\lim\limits_{i\to+\infty}R_x[i]=+\infty$.
	Similarly, we can also conclude $\lim\limits_{i\to+\infty}R_y[i]=+\infty$. 
	
	Now from $\Theta_x[i]<1$, we can know
	\begin{align}
		R_x[i]= & \dfrac{g_r(\tilde{\varepsilon}_{x,1}[2i])}{g_r(\tilde{\varepsilon}_{x,2}[2i])}
		=\dfrac{w_2(\varepsilon_{x,1}+\theta_{x,2i}\varepsilon_{x,2})-1}{w_2(\varepsilon_{x,2}-\theta_{x,2i}\varepsilon_{x,2})-1}\nonumber\\
		< & \dfrac{w_2(\varepsilon_{x,1}+\varepsilon_{x,2})-1}{w_2(\varepsilon_{x,2}-\theta_{x,2i}\varepsilon_{x,2})-1} \label{equ:upper_bound_ration_fex}
	\end{align}
	The numerator of the rightmost item of Equation \eqref{equ:upper_bound_ration_fex} is a constant. From  
	$\lim\limits_{i \to +\infty}R_x[i]=+\infty$, we can conclude  $\lim\limits_{i\to +\infty}w_2(\varepsilon_{x,2}-\theta_{x,2i}\varepsilon_{x,2})-1= 0$ and  $\lim\limits_{i\to +\infty}\varepsilon_{x,2}-\theta_{x,2i}\varepsilon_{x,2}=0$. Thus, we obtain
	\begin{align}\label{limit_theta_xi}
		\lim_{i\to +\infty}\Theta_x[i]=\lim_{i\to +\infty}\theta_{x,2i}=1
	\end{align}
	Similarly, we can also obtain 
	\begin{align}\label{limit_theta_yi}
		\lim_{i\to +\infty}\Theta_y[i]=\lim_{i\to +\infty}\theta_{y,2i+1}=1
	\end{align}
	Now combining Equations \eqref{equ:def_S}, \eqref{equ:Thetaxi}, \eqref{equ:Thetayi}, \eqref{equ:def_Si_seris}, \eqref{limit_theta_xi} and \eqref{limit_theta_yi}, we can know 
	\begin{align}
		&\lim_{i\to +\infty}S[i] \nonumber\\
		= & S(1,1)
		=f(w_2(\varepsilon_{x,1}+\varepsilon_{x,2})w_2(\varepsilon_{y,1}+\varepsilon_{y,2}))\nonumber\\
		&+f(w_2(\varepsilon_{x,2}-\varepsilon_{x,2})w_2(\varepsilon_{y,2}-\varepsilon_{y,2}))\nonumber\\
		=&f(w_2(\varepsilon_{x,1}+\varepsilon_{x,2})w_2(\varepsilon_{y,1}+\varepsilon_{y,2}))+1
	\end{align}
	Since $S[i]$ is strictly  increasing, we can conclude
	\begin{align}
		S(\theta_{x,0},\theta_{y,0})<f(w_2(\varepsilon_{x,1}+\varepsilon_{x,2})w_2(\varepsilon_{y,1}+\varepsilon_{y,2}))+1 \label{equ:extrem_distribute_exy_larger_than_anyneighbor}
	\end{align} 
	Remember that, we can take any $\theta_{x,0}>0$ and any $\theta_{y,0}\neq \theta_{y,1}$ as the start point of above iterations.
	This means that  Equation \eqref{equ:extrem_distribute_exy_larger_than_anyneighbor} holds for any point $(\theta_{x,0},\theta_{y,0})$ satisfying $\theta_{x,0}>0$ in the  neighborhood of $(0,0)$ on the $\theta_x\theta_y$ plane.
	
	Now we can show
	\begin{align}
		S(0,0)\leq f(w_2(\varepsilon_{x,1}+\varepsilon_{x,2})w_2(\varepsilon_{y,1}+\varepsilon_{y,2}))+1=S(1,1)
	\end{align}
	by contradiction.
	Assume to the contrary that $S(0,0)>f(w_2(\varepsilon_{x,1}+\varepsilon_{x,2})w_2(\varepsilon_{y,1}+\varepsilon_{y,2}))+1$, due to continuity of $S(\theta_x, \theta_y)$, we can find a neighbour $(\theta_x', \theta_y')$ ($\theta_x'>0$) of $(0,0)$ such that $S(\theta_x', \theta_y')>f(w_2(\varepsilon_{x,1}+\varepsilon_{x,2})w_2(\varepsilon_{y,1}+\varepsilon_{y,2}))+1$. This contradicts with Inequlity \eqref{equ:extrem_distribute_exy_larger_than_anyneighbor}.
	
	\noindent\ref{itm:case3}:
	
	Finally, we can discuss \ref{itm:case3} when one of $\varepsilon_{x,2}$ and $\varepsilon_{y,2}$ equals 0.
	Without loss of generality, we suppose that $\varepsilon_{y,2}=0$. 
	We can discuss this in two subcases.
	\begin{itemize}[-]
		\item \textbf{Subcase 3.1}: $\varepsilon_{y,1}=\varepsilon_{y,2}=0$. 
		By Lemma \ref{thm:solution_f_x} and Equation \eqref{equ:def_S},
		it is easy to know 
		\begin{align}
			S(1,1) =&f(w_2(\varepsilon_{x,1}+\varepsilon_{x,2})w_2(0))+f(w_2(0)w_2(0))\nonumber\\
			= & 1+\varepsilon_{x,1}+\varepsilon_{x,2}+1\nonumber\\
			S(0,0) = & f(w_2(\varepsilon_{x,1})w_2(0))+f(w_2(\varepsilon_{x,2})w_2(0))\nonumber\\
			=& 1+\varepsilon_{x,1}+1+\varepsilon_{x,2}\nonumber
		\end{align}
		This satisfies $S(0,0)\leq S(1,1)$.
		\item \textbf{Subcase 3.2}: $\varepsilon_{y,1}>\varepsilon_{y,2}=0$. 
		
		Here we treat $S(\theta_x, \theta_y)$ as a function of three variables $\theta_x, \theta_y, \varepsilon_{y,2}$ as follows.
		\begin{align}
			&S(\theta_x, \theta_y, \varepsilon_{y,2})\nonumber\\
			=&f(w_2(\varepsilon_{x,1}+\theta_x\varepsilon_{x,2})w_2(\varepsilon_{y,1}+\theta_y\varepsilon_{y,2}))\nonumber\\
			&+f(w_2(\varepsilon_{x,2}-\theta_x\varepsilon_{x,2})w_2(\varepsilon_{y,2}-\theta_y\varepsilon_{y,2}))
		\end{align}
		It is easy to know  $S(\theta_x, \theta_y, \varepsilon_{y,2})$ is continuous.
		Note that, we have proven $S(0,0,\varepsilon_{y,2}) \leq S(1,1,\varepsilon_{y,2})$ for any $\varepsilon_{y,2}>0$ in \ref{itm:case2}.
		Therefore, we have
		\begin{align}
			S(0, 0, 0)=\lim\limits_{\varepsilon_{y,2}\to 0}S(0, 0, \varepsilon_{y,2}) 
			\leq \lim\limits_{\varepsilon_{y,2}\to 0}S(1, 1, \varepsilon_{y,2})
			=  S(1,1,0)
		\end{align}
		so $S(0,0)\leq S(1,1)$ for $\varepsilon_{y,1}>\varepsilon_{y,2}=0$. 
	\end{itemize}	
	This concludes the proof of Lemma \ref{thm:kl_triangle_fw2w2_extreme_distribute_e_max}.
	
	$\hfill\square$ 
\end{proof}

Up to now, we have resolved the 2-dimensional case of the core of the proof. To extend to high-dimensional case, we need the following Lemma \ref{thm:trace_AB_supremum}.

\begin{lemma}\label{thm:trace_AB_supremum}
	(See \cite{lasserre1995trace}) For any two Hermition positive semidefinite $n\times n$-matrices $A,B$
	\begin{align}
		\mathop{\mathrm{Tr}}(AB)\leq \sum_{i=1}^{n}\lambda_{A,[i]}\lambda_{B,[i]}
	\end{align}
	where $\lambda_{A,[i]},\lambda_{B,[i]}$ are the eigenvalues of $A, B$ arranged in decreasing order, respectively.
\end{lemma}

Now we present our theorem on the relaxed triangle inequality of KL divergences between Gaussians.
Firstly, we deal with the case when one of the Gaussians is standard Gaussian. Then, we generalize the conclusion to general case.
\begin{lemma}\label{thm:triangle_n1_n2_standard}
	For any two $n$-dimensional Gaussian distributionss $\mathcal{N}(\bm{\mu}_1,\bm{\Sigma}_1)$ and $\mathcal{N}(\bm{\mu}_2,\bm{\Sigma}_2)$
	such that $KL(\mathcal{N}(\bm{\mu}_1,\bm{\Sigma}_1)||\mathcal{N}(0,I))\leq \varepsilon_1$, $KL(\mathcal{N}(0,I)||\mathcal{N}(\bm{\mu}_2,\bm{\Sigma}_2))\leq \varepsilon_2\ (\varepsilon_1,\varepsilon_2\ge 0)$ , then 
	\begin{align}
		&KL((\mathcal{N}(\bm{\mu}_1,\bm{\Sigma}_1)||\mathcal{N}(\bm{\mu}_2,\bm{\Sigma}_2))\\
		<&  \varepsilon_1 + \varepsilon_2 + \dfrac{1}{2}\left(W_{-1}(-e^{-(1+2\varepsilon_1)})W_{-1}(-e^{-(1+2\varepsilon_2)})\vphantom{\left(\sqrt{2\varepsilon_1}+\sqrt{\dfrac{2\varepsilon_2}{-W_{0}(-e^{-(1+2\varepsilon_2)})}}\right)^2}\right. \nonumber\\
		&+W_{-1}(-e^{-(1+2\varepsilon_1)})   +W_{-1}(-e^{-(1+2\varepsilon_2)})+1 \nonumber\\
		&\left. -W_{-1}(-e^{-(1+2\varepsilon_2)})\left(\sqrt{2\varepsilon_1}+\sqrt{\dfrac{2\varepsilon_2}{-W_{0}(-e^{-(1+2\varepsilon_2)})}}\right)^2\right)\nonumber
	\end{align}
\end{lemma}

\begin{proof}
	In the proofs of Lemma \ref{thm:lemma_sup_fx_invs_nd} (and Lemma \ref{thm:lemma_inf_fx_inv_nd} in Appendix), we construct equivalent optimization problems by introducing new variables in the constraints. Unfortunately, in the proof of Lemma \ref{thm:triangle_n1_n2_standard}, we cannot use the same step. Otherwise, the bound would be too complicated to resolve. To obtain a bound independent of the dimension $n$, we need to relax the constraint in the beginning.

	Our aim is to find an upper bound of $KL((\mathcal{N}(\bm{\mu}_1,\bm{\Sigma}_1)||\mathcal{N}(\bm{\mu}_2,\bm{\Sigma}_2))$ under the constraints $KL(\mathcal{N}(\bm{\mu}_1,\bm{\Sigma}_1)||\mathcal{N}(0,I))\leq \varepsilon_1$, $KL(\mathcal{N}(0,I)||\mathcal{N}(\bm{\mu}_2,\bm{\Sigma}_2))\leq \varepsilon_2$.
	In the following, we first relax the constraints and then find an upper bound under the relaxed constraints.
	
	According to the definition of KL divergence, we have
	\begin{equation}\label{equ:def_kl_n1_n2}
		\begin{aligned}\nonumber
			&KL((\mathcal{N}(\bm{\mu}_1,\bm{\Sigma}_1)||\mathcal{N}(\bm{\mu}_2,\bm{\Sigma}_2))\\
			=&\dfrac{1}{2}\left(\log \dfrac{|\bm{\bm{\Sigma}_2}|}{|\bm{\Sigma}_1|}+\mathop{\mathrm{Tr}}(\bm{\bm{\Sigma}_2}^{-1}\bm{\Sigma}_1)+(\bm{\mu}_2-\bm{\mu}_1)^\top\bm{\Sigma}_2^{-1}(\bm{\mu}_2-\bm{\mu}_1)-n\right) 
		\end{aligned} 
	\end{equation}
	
	In the following two steps, we first find an upper bound for the first two items, then we find an upper bound for the rest items.
	
	\textbf{Step 1}. According to Lemma \ref{thm:trace_AB_supremum}, we have
	\begin{align}
		&\log \dfrac{|\bm{\bm{\Sigma}_2}|}{|\bm{\Sigma}_1|}+\mathop{\mathrm{Tr}}(\bm{\bm{\Sigma}_2}^{-1}\bm{\Sigma}_1)\nonumber\\
		=& \mathop{\mathrm{Tr}}(\bm{\bm{\Sigma}_2}^{-1}\bm{\Sigma}_1) - \log \dfrac{|\bm{\bm{\Sigma}_1}|}{|\bm{\Sigma}_2|}\nonumber\\
		 =& \mathop{\mathrm{Tr}}(\bm{\bm{\Sigma}_2}^{-1}\bm{\Sigma}_1) - \log (|\bm{\bm{\Sigma}_2}^{-1}||\bm{\bm{\Sigma}_1}|)\nonumber\\
		= & \mathop{\mathrm{Tr}}(\bm{\bm{\Sigma}_2}^{-1}\bm{\Sigma}_1)- \log \prod_{i=1}^{n}\lambda_{1,i}\lambda{'}_{2,i}\nonumber\\
		 \leq &\sum_{i=1}^{n}\lambda_{1,[i]}\lambda{'}_{2,[i]}- \log \prod_{i=1}^{n}\lambda_{1,[i]}\lambda{'}_{2,[i]}\nonumber\\
		 =& \sum_{i=1}^{n}\lambda_{1,[i]}\lambda{'}_{2,[i]}- \log \lambda_{1,[i]}\lambda{'}_{2,[i]} \label{equ:left_item_kl_n1_n2_bound_to_core_function}
	\end{align}
	where $\lambda_{1,i}, \lambda{'}_{2,i}$ are the eigenvalues of $\bm{\Sigma}_1, \bm{\Sigma}^{-1}_2$ arranged in decreasing order, respectively. In the following, we find an upper bound for Equation \eqref{equ:left_item_kl_n1_n2_bound_to_core_function}.

	By the definition of KL divergence, the constraint $KL(\mathcal{N}(\bm{\mu}_1,\bm{\Sigma}_1)||\mathcal{N}(0,I))\leq \varepsilon_1$ is equal to 
	\begin{align}\label{equ:precondition_kl_n1_standard}
		-\log|\bm{\Sigma}_1|+\mathop{\mathrm{Tr}}(\bm{\Sigma}_1)+\bm{\mu}_1^{\top}\bm{\mu}_1-n\leq 2\varepsilon_1
	\end{align}
	Combining Lemma \ref{equ:f_convex_minimum}, Equation \eqref{equ:det_to_multi} and \eqref{equ:tr_sigma}, 
	we relax the constraint in Inquality \eqref{equ:precondition_kl_n1_standard} as follows.
	\begin{align}
		-\log|\bm{\Sigma}_1|+\mathop{\mathrm{Tr}}(\bm{\Sigma}_1)=\sum_{i=1}^n\lambda_{1,i}-\log \lambda_{1,i}\leq n+2\varepsilon_1\label{equ:bound_kl_1_part1_all_d}\\
		\bm{\mu}_1^{\top}\bm{\mu}_1\leq 2\varepsilon_1 \label{equ:bound_kl_1_part2}
	\end{align}
	where $\lambda_{1,i}$ are the eigenvalues of $\bm{\Sigma}_1$. 
	For simplicity, we modify the constraint in Inquality \eqref{equ:bound_kl_1_part1_all_d} to the following constraint.
	\begin{align}
		-\log|\bm{\Sigma}_1|+\mathop{\mathrm{Tr}}(\bm{\Sigma}_1)=\sum_{i=1}^n\lambda_{1,i}-\log \lambda_{1,i}= n+2\varepsilon_1\label{equ:bound_kl_1_part1_all_d_modify}
	\end{align}
	In the following, we find the upper bound for Equation \eqref{equ:left_item_kl_n1_n2_bound_to_core_function} under constraints \eqref{equ:bound_kl_1_part1_all_d_modify}. Then we will see that the upper bound is increasing with $\varepsilon_1$. So there is no difference between constraints \eqref{equ:bound_kl_1_part1_all_d}  and  \eqref{equ:bound_kl_1_part1_all_d_modify}.
	
	Form the perspective of optimization, 
	the constraint in Inequality \eqref{equ:bound_kl_1_part1_all_d_modify} can be replaced by the following constraints
	\begin{align}
		\lambda_{1,i}-\log \lambda_{1,i}= 1+\varepsilon_{1,i}\ (1\leq i\leq n) \label{equ:bound_each_lambda_1_d}\\
		\bigwedge\limits_{i=1}^n\varepsilon_{1,i}\geq 0\wedge \sum\limits_{i=1}^{n}\varepsilon_{1,i}=2\varepsilon_1 \label{equ:bound_each_e1i}
	\end{align}	
	Similarly, the constraint $KL(\mathcal{N}(0,I)||\mathcal{N}(\bm{\mu}_2,\bm{\Sigma}_2))\leq \varepsilon_2$ is equal to 
	\begin{align}
		\log|\bm{\Sigma}_2|+\mathop{\mathrm{Tr}}(\bm{\Sigma}^{-1}_2)+\bm{\mu}_2^{\top}\bm{\Sigma}_2^{-1}\bm{\mu}_2-n\leq 2\varepsilon_2
	\end{align}
	which implies the following constraints
	\begin{align}
		\log|\bm{\Sigma}_2|+\mathop{\mathrm{Tr}}(\bm{\Sigma}^{-1}_2)=\sum_{i=1}^n\lambda{'}_{2,i}-\log \lambda{'}_{2,i}\leq n+2\varepsilon_2 \label{equ:bound_kl_2_part1_all_d}\\
		\bm{\mu}_2^{\top}\bm{\Sigma}_2^{-1}\bm{\mu}_2\leq 2\varepsilon_2 \label{equ:bound_kl_2_part2}
	\end{align}
	where $\lambda{'}_{2,i}$ are the eigenvalues of $\bm{\Sigma}_2^{-1}$.
	We also modify the constraint in Inequality \eqref{equ:bound_kl_2_part1_all_d} to the following constraint which does not affect the upper bound.
	\begin{align}
		\log|\bm{\Sigma}_2|+\mathop{\mathrm{Tr}}(\bm{\Sigma}^{-1}_2)=\sum_{i=1}^n\lambda{'}_{2,i}-\log \lambda{'}_{2,i}= n+2\varepsilon_2 \label{equ:bound_kl_2_part1_all_d_modify}
	\end{align}
	Furthermore, constraint \eqref{equ:bound_kl_2_part1_all_d_modify} can be replaced by the following constraints.
	\begin{align}
		\lambda{'}_{2,i}-\log \lambda{'}_{2,i}= 1+\varepsilon_{2,i}\ (1\leq i\leq n) \label{equ:bound_each_lambda_2_d}\\
		\bigwedge\limits_{i=1}^n\varepsilon_{2,i}\geq 0\wedge\sum\limits_{i=1}^{n}\varepsilon_{2,i}=2\varepsilon_2\label{equ:bound_each_e2i} 
	\end{align}
	In the following, we find an upper bound of Equation \eqref{equ:left_item_kl_n1_n2_bound_to_core_function} under constraints 
	\eqref{equ:bound_each_lambda_1_d},
	\eqref{equ:bound_each_e1i},
	\eqref{equ:bound_each_lambda_2_d}, and
	\eqref{equ:bound_each_e2i}. 
	
	
	Applying Lemma \ref{thm:sup_fxy} to Equation \eqref{equ:left_item_kl_n1_n2_bound_to_core_function}  with conditions \eqref{equ:bound_each_lambda_1_d} and \eqref{equ:bound_each_lambda_2_d}, we can obtain 
	\begin{align}\label{equ:bound_lambda_prod_to_fw2_w2}
		\sum_{i=1}^{n}\lambda_{1,[i]}\lambda{'}_{2,[i]}- \log \lambda_{1,[i]}\lambda{'}_{2,[i]} \leq \sum_{i=1}^{n}f(w_2(\varepsilon_{1,[i]})w_2(\varepsilon_{2,[i]}))
	\end{align}
	where $\varepsilon_{1,[i]}$ and $\varepsilon_{2,[i]}$ are also arranged in decreasing order.
	
	%
	Now we apply Lemma \ref{thm:kl_triangle_fw2w2_extreme_distribute_e_max} to the right hand side of Inequality \eqref{equ:bound_lambda_prod_to_fw2_w2} repeatedly on the first two dimensions as follows. Here we use notations 
	$E_{1,k}=\sum_{i=1}^{k}\varepsilon_{1,[i]}, E_{2,k}=\sum_{i=1}^{k}\varepsilon_{2,[i]}$ for brevity. 
	
	\begin{DispWithArrows}
		& \log \dfrac{|\bm{\bm{\Sigma}_2}|}{|\bm{\Sigma}_1|}+\mathop{\mathrm{Tr}}(\bm{\bm{\Sigma}_2}^{-1}\bm{\Sigma}_1) \Arrow{by \eqref{equ:left_item_kl_n1_n2_bound_to_core_function}}\nonumber\\
		\leq & \sum_{i=1}^{n}\lambda_{1,[i]}\lambda{'}_{2,[i]}- \log \lambda_{1,[i]}\lambda{'}_{2,[i]} \Arrow{by \eqref{equ:bound_lambda_prod_to_fw2_w2}}\nonumber\\
		\leq & \sum_{i=1}^{n}f(w_2(\varepsilon_{1,[i]})w_2(\varepsilon_{2,[i]}))\nonumber\\
		= & f(w_2(\varepsilon_{1,[1]})w_2(\varepsilon_{2,[1]})) + f(w_2(\varepsilon_{1,[2]})w_2(\varepsilon_{2,[2]})) \nonumber\\
		&+ \sum_{i=3}^{n}f(w_2(\varepsilon_{1,[i]})w_2(\varepsilon_{2,[i]}))\Arrow{Lemma \ref{thm:kl_triangle_fw2w2_extreme_distribute_e_max}}\nonumber\\
		\leq & f(w_2(\varepsilon_{1,[1]}+\varepsilon_{1,[2]})w_2(\varepsilon_{2,[1]}+\varepsilon_{2,[2]}))+1\nonumber\\
		&+ \sum_{i=3}^{n}f(w_2(\varepsilon_{1,[i]})w_2(\varepsilon_{2,[i]}))\nonumber\\
		= & f(w_2(E_{1,2})w_2(E_{2,2})) +f(w_2(\varepsilon_{1,[3]})w_2(\varepsilon_{2,[3]}))\nonumber\\
		&+ \sum_{i=4}^{n}f(w_2(\varepsilon_{1,[i]})w_2(\varepsilon_{2,[i]}))+1\Arrow{Lemma \ref{thm:kl_triangle_fw2w2_extreme_distribute_e_max}}\nonumber\\
		\leq & f(w_2(E_{1,2}+\varepsilon_{1,[3]})w_2(E_{2,2}+\varepsilon_{2,[3]}))+ 1\nonumber\\ &+\sum_{i=4}^{n}f(w_2(\varepsilon_{1,[i]})w_2(\varepsilon_{2,[i]}))+1\nonumber\\
		=& f(w_2(E_{1,3})w_2(E_{2,3})) + \sum_{i=4}^{n}f(w_2(\varepsilon_{1,[i]})w_2(\varepsilon_{2,[i]})) \nonumber\\
		&+2 \nonumber \\
		\dots & \nonumber\\
		\leq & f(w_2(E_{1,n})w_2(E_{2,n}))+n-1\nonumber\\
		= & f(w_2(\sum_{i=1}^{n}\varepsilon_{1,[i]})w_2(\sum_{i=1}^{n}\varepsilon_{2,[i]}))+n-1\nonumber\\
		= & f(w_2(2\varepsilon_1)w_2(2\varepsilon_2))+n-1 \Arrow{ Lemma \ref{thm:f_w1w1_w2w2_plugin}}\nonumber\\
		= & 2\varepsilon_1 + 2\varepsilon_2 + 2 + w_2(2\varepsilon_1)w_2(2\varepsilon_2)-w_2(2\varepsilon_1)\nonumber\\
		&-w_2(2\varepsilon_2)+n-1\nonumber\\
		= & 2\varepsilon_1 + 2\varepsilon_2 + w_2(2\varepsilon_1)w_2(2\varepsilon_2)-w_2(2\varepsilon_1)-w_2(2\varepsilon_2)\nonumber\\
		&+n+1\label{equ:bound_kl_n1_n2_part1}
	\end{DispWithArrows}
	The bound in Equation \eqref{equ:bound_kl_n1_n2_part1} is increasing with $\varepsilon_1$ and $\varepsilon_2$. Therefore, the constraints \eqref{equ:bound_kl_1_part1_all_d_modify} and \eqref{equ:bound_kl_2_part1_all_d_modify}  can be modified back to \eqref{equ:bound_kl_1_part1_all_d} and \eqref{equ:bound_kl_2_part1_all_d}, respectively .

	%
	%
	
	\textbf{Step 2}. from  Equation \eqref{equ:bound_kl_1_part2}, we know 
	\begin{align}\label{equ:bound_mu_1_norm}
		\vert\bm{\mu}_1\vert\leq \sqrt{2\varepsilon_1}
	\end{align}
	where $\vert\cdot\vert$ denotes the $L_2$ norm of vector.
	From Inequlity \eqref{equ:bound_kl_2_part2}, we also know $\lambda_{2*}'\bm{\mu}_2^\top\bm{\mu}_2\leq \bm{\mu}_2^{\top}\bm{\Sigma}_2^{-1}\bm{\mu}_2\leq 2\varepsilon_2 $, where $\lambda_{2*}'$ is the minimum eigenvalue of $\bm{\Sigma}_2^{-1}$. Now combining the condition \eqref{equ:bound_kl_2_part1_all_d} and Lemma \ref{prp:sup_f_x_invers_1_d}, we get 
	\begin{align}
		\bm{\mu}_2^\top\bm{\mu}_2\leq \dfrac{2\varepsilon_2}{\lambda_{2*}'}\leq \dfrac{2\varepsilon_2}{w_1(2\varepsilon_2)}\implies  \vert\bm{\mu}_2\vert \leq  \sqrt{\dfrac{2\varepsilon_2}{w_1(2\varepsilon_2)}} \label{equ:bound_mu_2_norm}
	\end{align}
	Combining Inequalities \eqref{equ:bound_mu_1_norm}, \eqref{equ:bound_mu_2_norm} and using the triangle inequality for norms of vectors, we have
	\begin{align}
		\vert\bm{\mu}_2-\bm{\mu}_1\vert \leq \vert\bm{\mu}_2\vert + \vert\bm{\mu}_1\vert \leq   \sqrt{2\varepsilon_1}+\sqrt{\dfrac{2\varepsilon_2}{w_1(2\varepsilon_2)}}
	\end{align}
	Again, we have $(\bm{\mu}_2-\bm{\mu}_1)^\top\bm{\Sigma}_2^{-1}(\bm{\mu}_2-\bm{\mu}_1)\leq \lambda{'}_2^{*}\vert\bm{\mu}_2-\bm{\mu}_1\vert^2$, where $\lambda{'}_2^{*}$ is the maximum eigenvalue of $\bm{\Sigma}_2^{-1}$. From Lemma \ref{prp:sup_f_x_invers_1_d} and condition \eqref{equ:bound_kl_2_part1_all_d}, we know $\lambda{'}_2^{*}\leq w_2(2\varepsilon_2)$. Thus, we can conclude that 
	\begin{align}\label{equ:bound_kl_n1_n2_part2}
		(\bm{\mu}_2-\bm{\mu}_1)^\top\bm{\Sigma}_2^{-1}(\bm{\mu}_2-\bm{\mu}_1)
		\leq & w_2(2\varepsilon_2)\vert\bm{\mu}_2-\bm{\mu}_1\vert^2 \nonumber\\
		\leq & w_2(2\varepsilon_2)\left(\sqrt{2\varepsilon_1}+\sqrt{\dfrac{2\varepsilon_2}{w_1(2\varepsilon_2)}}\right)^2
	\end{align}
	Finally, combining Inequalities \eqref{equ:bound_kl_n1_n2_part1} and \eqref{equ:bound_kl_n1_n2_part2}, we can conclude that 
	\begin{align}\label{equ:bound_kl_relexed_triangle_one_std_Gaussian}
		\small
		&KL((\mathcal{N}(\bm{\mu}_1,\bm{\Sigma}_1)||\mathcal{N}(\bm{\mu}_2,\bm{\Sigma}_2))\nonumber\\
		<& \dfrac{1}{2}\left(2\varepsilon_1 + 2\varepsilon_2 + w_2(2\varepsilon_1)w_2(2\varepsilon_2)-w_2(2\varepsilon_1)-w_2(2\varepsilon_2)\vphantom{\left(\sqrt{2\varepsilon_1}+\sqrt{\dfrac{2\varepsilon_2}{w_1(2\varepsilon_2)}}\right)^2}\right. 
	 \nonumber\\
		&	+n+1 + \left.w_2(2\varepsilon_2)\left(\sqrt{2\varepsilon_1}+\sqrt{\dfrac{2\varepsilon_2}{w_1(2\varepsilon_2)}}\right)^2-n\right)\nonumber\\
		= &  \varepsilon_1 + \varepsilon_2 + \dfrac{1}{2}\left(W_{-1}(-e^{-(1+2\varepsilon_1)})W_{-1}(-e^{-(1+2\varepsilon_2)}) \vphantom{\left(\sqrt{2\varepsilon_1}+\sqrt{\dfrac{2\varepsilon_2}{-W_{0}(-e^{-(1+2\varepsilon_2)})}}\right)^2}\right.  \nonumber\\ 
		&   +W_{-1}(-e^{-(1+2\varepsilon_1)})+W_{-1}(-e^{-(1+2\varepsilon_2)})+1\nonumber \\
		&  \left. -W_{-1}(-e^{-(1+2\varepsilon_2)})\left(\sqrt{2\varepsilon_1}+\sqrt{\dfrac{2\varepsilon_2}{-W_{0}(-e^{-(1+2\varepsilon_2)})}}\right)^2\right)
	\end{align}
	
	$\hfill\square$	
\end{proof}
\begin{remark} The bound in Equation \eqref{equ:bound_kl_relexed_triangle_one_std_Gaussian} has the following properties. 
	\begin{enumerate}
		\item The bound becomes 0 when $\varepsilon_1=\varepsilon_2=0$.
		\item When both $\varepsilon_1$ and $\varepsilon_2$ are small, all the items in the bound are small and hence the bound is small.
		\item The bound is independent of the dimension $n$ because we have eliminated the impact of dimension $n$ by Lemma \ref{thm:kl_triangle_fw2w2_extreme_distribute_e_max}. This is the most tricky part in this proof.
		\item When $\varepsilon_2$ is large, the bound is mostly dominated by the last item in the branket.
		\item In fact, we can distribute $2\varepsilon_1$ into two parts in Equation \eqref{equ:bound_kl_1_part1_all_d} and \eqref{equ:bound_kl_1_part2}. However, this will lead to a complicated expression which is very hard to solve a supremum as like what we do on Equation \eqref{equ:bound_bKL}. Numerical experiments show that the supremum varies with how $2\varepsilon_1$ ($2\varepsilon_2$) are allocated into two parts in the left hand sides of Equation \eqref{equ:bound_kl_1_part1_all_d} and \eqref{equ:bound_kl_1_part2} (\eqref{equ:bound_kl_2_part1_all_d} and \eqref{equ:bound_kl_2_part2}).
		Therefore, in constraints 
		\eqref{equ:bound_kl_1_part1_all_d}, \eqref{equ:bound_kl_1_part2}, \eqref{equ:bound_kl_2_part1_all_d}, and \eqref{equ:bound_kl_2_part2}, we relax the conditions and get an relaxed upper bound with a simpler form in Inequality \eqref{equ:bound_kl_relexed_triangle_one_std_Gaussian}.
	\end{enumerate}
\end{remark}

\noindent\textbf{Proof of Theorem \ref{thm:triangle_n1_n2_n3}}.
\begin{proof}
	Theorem \ref{thm:triangle_n1_n2_n3} extends Lemma \ref{thm:triangle_n1_n2_standard} to three general Gaussians.  
	We can use linear invertible transformation to convert one Gaussian into standard Gaussian and then apply Lemma \ref{thm:triangle_n1_n2_standard}. Please see Appendix \ref{sec:appen_proof_thm:triangle_n1_n2_n3} for details.
	
		$\hfill\square$	
\end{proof}

In Theorem \ref{thm:triangle_n1_n2_n3}, we try to find an upper bound as tight as possible. So the bound seems a little complicated.  
We can expand Lambert $W$ function by series \cite{2016PrincetonCompanion, corless1996lambertw} and simplify the bound as follows \cite{liu2022constrainedSafeRL} \footnote{After we post our last version of manuscript on Arxiv \cite{onThePropertiesOfKL2021Yufeng}, Liu \textit{et al.} cited our manuscript in their work \cite{liu2022constrainedSafeRL} in which they simplify the bound in Theorem \ref{thm:triangle_n1_n2_n3} by using series in simpler case $\varepsilon_1=\varepsilon_2$.}.
\begin{theorem}
	\label{thm:triangle_n1_n2_n3_simplify}
	For any three $n$-dimensional Gaussians $\mathcal{N}(\bm{\mu}_i,\bm{\Sigma}_i)(i\in\{1,2,3\})$  
	such that $KL(\mathcal{N}(\bm{\mu}_1,\bm{\Sigma}_1)||\mathcal{N}(\bm{\mu}_2,\bm{\Sigma}_2))\leq \varepsilon_1$ and 
	$ KL(\mathcal{N}(\bm{\mu}_2,\bm{\Sigma}_2)||\mathcal{N}(\bm{\mu}_3,\bm{\Sigma}_3))\leq \varepsilon_2$ for small $\varepsilon_1, \varepsilon_2\ge 0$, then 
	\begin{align}
		&KL(\mathcal{N}(\bm{\mu}_1,\bm{\Sigma}_1)||\mathcal{N}(\bm{\mu}_3,\bm{\Sigma}_3))\nonumber \\
		< & 3\varepsilon_1+3\varepsilon_2+2\sqrt{\varepsilon_1\varepsilon_2}+o(\varepsilon_1)+o(\varepsilon_2)
	\end{align}
\end{theorem}
\begin{proof}
	See Appendix \ref{sec:proof_thm:triangle_n1_n2_n3_simplify} for the details of the proof.

	$\hfill\square$ 
\end{proof}

Finally, in the proof of Theorem \ref{thm:triangle_n1_n2_n3}, we use invertible linear transformation to convert $\mathcal{N}_2$ to standard Gaussian with preserving KL divergence. This still holds  when $\mathcal{N}(\bm{\mu}_2,\bm{\Sigma}_2)$ is fixed. So we get the the following corollary.

\begin{corollary}
	Theorem \ref{thm:triangle_n1_n2_n3} and \ref{thm:triangle_n1_n2_n3_simplify} hold when $\mathcal{N}(\bm{\mu}_2,\bm{\Sigma}_2)$ is fixed.
\end{corollary}

\begin{remark} \label{remark_different_bregman_triangle}
\noindent\textbf{Comparison with existing general Pythagoras inequalities}\\
It is known that KL divergence satisfies some general Pythagoras inequalities which seem similar to our relaxed triangle inequality. We note that they are different in the follows.

The bound in our relaxed triangle inequality is independent of the parameters of Gaussians and only related to  $\varepsilon_1$ and $\varepsilon_2$. 
Our theorem is different from the several existing generalized Pythagoras inequalities satisfied by KL divergence, where the bounds are functions of the given distributions. 
We list them as follows.
	\begin{enumerate}
		\item The generalized Pythagoras inequality for KL divergence \cite{cover2012elements,Renyi-and-KL2014} states that for a convex set of distributions $\mathcal{P}$, any distribution $Q$ not in $\mathcal{P}$, and $D_{min}=\inf_{P\in\mathcal{P}}KL(P||Q)$, there exists a distribution $P^*$ such that 
		$$KL(P||Q)\ge KL(P||P^*)+D_{min}\ \ \ \text{for all}\ \ \ P \in \mathcal{P}$$
		
		\item Erven \textit{et al.} generalize the Pythagoras inequality for KL divergence to R\'{e}nyi divergence which includes KL divergence with order 1. See \cite{Renyi-and-KL2014} for details.
		
		\item Functional Bregman divergence also satisfies a generalized Pythagoras theorem \cite{functional_bregman_divergence}. 
		Let $(\mathbb{R}^d,\Omega,v)$ be a measure space, where $d$ is a positive integer and $v$ is a Borel measure. Let $\mathcal{A}$ be a convex subset of $L^p(v)$.  
		For any $f,g,h\in \mathcal{A}$, functional Bregman divergence $d_\phi$ satisfies
		\begin{equation}\label{equ:triangle-bregman}
			d_{\phi}[f,h]=d_{\phi}[f,g]+d_{\phi}[g,h]+\delta\phi[g;f-g]-\delta\phi[h;f-g]	
		\end{equation}
		where $\phi:L^p(v)\rightarrow\mathbb{R}$ is a strictly convex, twice-continuously Fr\'{e}chet-differentiable functional.  $\delta\phi[g;\cdot]$ is the Fr\'{e}chet derivative of $\phi$ at $g$.
		KL divergence is a special form of functional Bregman divergence when $\phi=\int p(x)\log p(x)\dif x$ whose Fr\'{e}chet derivative at $g$ is $\delta\phi[g;t]=\int (\log g(x)+1)t(x)\dif x$. Plugging $\phi$ and $\delta\phi$ into Equation \eqref{equ:triangle-bregman}, we get
		\begin{align}
			\small
			&KL(f||h) \nonumber \\
			=& KL(f||g)+KL(g||h) 
			+\int (\log g(x)+1)(f(x)-g(x))\dif x \nonumber \\
			& -\int (\log h(x)+1)(f(x)-g(x))\dif x \nonumber\\
			=& KL(f||g)+ \int f(x)\log \dfrac{g(x)}{h(x)}\dif x
		\end{align}
		
	\end{enumerate}

All the bounds in the above inequalities are dependent on the parameters of the given distributions. 
	
	
\end{remark}
In our theorem, we allow all parameters are unknown or one Gaussian is fixed. Therefore, our theorems are suitable for contexts where the re can vary. This is common in deep learning where the parameters are learned by the model. Therefore, it is impossible to identify the parameters or the KL divergence before the model is trained. In some cases, we only know that some bound is guaranteed. 
In the next section, we discuss the applications of our theorems in deep learning.

\section{Applications}\label{sec:application}

\subsection{Anomaly Detection with Flow-based Model}
The research question in this paper comes from our research on deep anomaly detection using flow-based model \cite{dinh2014nice,dinh2016realnvp,kingma2018glow, flow-tpami-2021}. 
Flow-based model constructs diffeomorphism between data space to latent space.  Compared with other generative models such as generative adversarial networks, flow-based model has the advantage of providing explicit likelihood $p_{\theta}(x)$ to input $x$, where $\theta$ refer to model parameters. Usually, flow-based model is trained by maximum likelihood estimate with Gaussian prior. Intuitively, it is natural to believe that samples from the training (in-distribution, ID in short) dataset should have higher likelihoods than out-of-distribution (OOD) data (\textit{i.e.}, anomalies). However, Nalisnick \textit{et al.} reveal that deep generative models including flow-based models may assign higher likelihoods to OOD data \cite{nalisnick2018deep}. 
For example, Glow \cite{kingma2018glow} assigns higher likelihoods for SVHN when trained on CIFAR-10.
This observation is also verified by many other researchers including ourselves \cite{shafaei2018digitnotcat, choi2018generative, vskvara2018generative,nalisnick2019detecting,whyflowfailood}. 
This brings obstacles to anomaly detection in flow-based model according to model likelihood \cite{nalisnick2019detecting, whyflowfailood}. However, we can not sample these OOD data from the model although they may have higher likelihoods than training data. Nalisnick et al. explain this phenomenon by the discrepancy of typical set and high probability density regions of model distribution \cite{nalisnick2019detecting}. This can explain why we can not sample OOD data that have higher likelihoods than ID data. But their explanation fails when OOD data has coinciding likelihoods with ID data. Before our analysis, this counterintuitive phenomenon has not been satisfactorily explained. 

In this context, we want to explain \textit{why we can not sample OOD data from flow-based model with prior regardless of when OOD data have higher, lower, or coinciding likelihoods}.
We investigate this problem from a statistical divergence perspective. 
Let $z=f(x)$ be the flow-based model which maps data $x$ in data space to $z$ in latent space. Assume that the prior distribution $p_Z^r$ is the most commonly used Gaussian distribution.  
Suppose that $X_1\sim p_X(x)$, $X_2\sim q_X(x)$ represent distributions of ID and OOD datasets, respectively. We note $Z_1=f(X_1)\sim p_Z(z)$, $Z_2=f(X_2)\sim q_Z(z)$ to represent the distributions of representations of ID and OOD datasets, respectively. We also note the model induced distribution $p_X^r$ such that $Z_r\sim p_Z^r$ and $X_r=f^{-1}(Z_r)\sim p^r_X$. 
Flow-based model is usually trained by maximum likelihood estimation. 	This is equal to minimizing  forward KL divergence $KL(p_X||p^r_X)$ \cite{papamakarios2019flow_model_survey, goodfellow2016deep}.	
In our experiments, we conduct generalized Shapiro-Wilk test for multivariate normality. Results demonstrate that $p_Z$ is Gaussian-like for all datasets. Surprisingly, $q_Z$ is also Gaussian-like for OOD datasets with higher or coinciding likelihoods except for just one case. These results allow us to approximate $p_Z$ and $q_Z$ with Gaussians.
Note that, it seems that the normality of represetations of ID and OOD dataset is a characteristic of flow-based model. We did not get similar observations in variational autoencoders.

The theorems proved in this paper can help us to analyze the KL divergences between $p_Z$, $p_Z^r$, and $q_Z$. On one hand, according to Proposition \ref{thm:preserve_KL}, we can know $KL(p_X||p^r_X)=KL(p_Z||p^r_Z)$, so $KL(p_Z||p^r_Z)$ is trained to be small. By Theorem \ref{thm:duality_small_KL_general}, we can know  $KL(p^r_Z||p_Z)$ is small too. So we can assume $p^r_Z\approx  p_Z$ when $KL(p_Z||p^r_Z)$ is sufficiently small. 	
On the other hand, we can also assume that the distributions of ID and OOD data are far from each other. This implies that $KL(p_X||q_X)=KL(p_Z||q_Z)$ can be any large. By the relaxed triangle inequality, we can infer that $KL(p_Z^r||q_Z)$ must be large.
This answers the question why we can not sample OOD data from flow-based model with prior. 
Furthermore, we decompose the large divergence  $KL(q_Z||p_Z^r)$ into dimensional-wise KL divergence and total correlation (generalized mutual information) measuring the mutual dependence between dimensions. We demonstrate that the representations of OOD data are more correlated than that of ID data. From a geometric perspective, strong correlation indicates that the representations of OOD data locate in specific directions. In high dimensional space, it is hard to sample data residing in specific directions from prior. This gives the second explanation to the above question.
Based on the theoretical analysis and further observation on the local pixel dependence in the representation of OOD dataset, we propose a KL divergence-based anomaly detection algorithm. Experimental results have shown the effectiveness of our method. More details of the application of our theorems in deep anomaly detection research can be referred to in our manuscript \cite{zhang2021outofdistribution}, which is submitted independently.

Importantly, flow-based model constructs diffeomorphism between data space to latent space with thousands of dimensions. It is important that the bounds found in this paper are independent of the dimension. 
Furthermore, since both $p_Z$ and $q_Z$ are dependent on model parameters and $q_Z$ is also dependent on the input OOD dataset, it is impossible to determine the parameters of $p_Z$ and $q_Z$ in advance. Our theorems do not dependent on the parameters of distributions and only requires some bound is restricted. This is why we need to prove the theorems in this paper rather than using existing theorems. 


\subsection{Safety Guarantee in Reinforcement Learning}
The theorems proved in this paper can also be used as general conclusions in related fields.
Since we post the last version of this manuscript on Arxiv \cite{onThePropertiesOfKL2021Yufeng}, our manuscript has been cited by other researchers. For example, the relaxed triangle inequality (Theorem \ref{thm:triangle_n1_n2_n3}) has been used in the research of constrained variational policy optimization for safe reinforcement learning \cite{liu2022constrainedSafeRL}. In their work, Liu \textit{et al.} propose an Expectation-Maximization style approach for learning safe policy in reinforcement learning. After achieving one-step robutness guarantee, a natural question is extending to multiple steps policy updating robustness guarantee. This requires triangle inequality for consecutive updated policies. It is known that KL divergence does not has such property in general case. However, multivariate Gaussian is commonly used as policy in continuous action space tasks. In such context, our relaxed triangle inequality (Theorem \ref{thm:triangle_n1_n2_n3}) is used to extend one-step robustness guarantee to multiple steps.  Particularly, Liu \textit{et al.} use big-$O$ to simplify the bound in Theorem \ref{thm:triangle_n1_n2_n3}  in case $\varepsilon_1=\varepsilon_2$ . Please see \cite{liu2022constrainedSafeRL} for more details about the application.

\section{Related work}\label{sec:related-work}
KL divergence is an important divergence and has a wide range of applications \cite{cover2012elements,PRML,goodfellow2016deep,pardo2018statistical,opti-reinforce-KL2010, MDP-KL-cost2014,  unifying-entropies-KL}. 
Researchers have investigated KL divergence between many distributions including Markov sources \cite{rached2004kullback}, GMM models \cite{low_upper_bound_kl_gmm,hershey2007approximating}, multivariate generalized Gaussians \cite{KLD_MGGD}, univariate mixtures \cite{nielsen2016guaranteed}, discrete normal distributions \cite{nielsen2021kullback}, \textit{etc}. In \cite{pardo2018statistical}, a bound of KL divergence between Gaussians is given. 
As we discussed in Remark \ref{remark_different_bregman_triangle}, existing  generalized Pythagorean inequality satisfied by KL divergence are all dependent on the parameters of distributions \cite{cover2012elements, Renyi-and-KL2014}. 
As far as we know, there is no related work that focuses on the similar properties of KL divergence between Gaussians as this paper. 

KL divergence is one member of more general divergences such as Bregman divergence \cite{bregman1967relaxation,cluster-bregman-div, functional_bregman_divergence, on-bregmen-distance2012}, $f$-divergence \cite{ali1966general, pardo2018statistical, JMLR-bound-f-divergence-2021}, R\'{e}nyi divergence \cite{renyi1961measures, pardo2018statistical,Renyi-and-KL2014}, and recently proposed $(f,\Gamma)$-divergence \cite{f-gamma-divergence}.
Bregman divergence defines a class of divergences \cite{cluster-bregman-div} in vector space. KL divergence between multinomial distributions is a special form of Bregman divergence when the convex function for Bregman divergence is chosen as $\sum_{i=1}^{n}p_i\log p_i$, where $p_i\ge 0$ for $i=1,\dots,n$ define a multinomial distribution. 
Frigyik \textit{et. al.} \cite{functional_bregman_divergence} extends vector Bregman divergence to functional Bregman divergence in $L^p$. Similarly, KL divergence is a special form of functional Bregman divergene.
(functional) Bregman divergence also satisfies generalized Pythagoras theorem \cite{cluster-bregman-div, functional_bregman_divergence}. We note that our relaxed triangle inequality has a different meaning in Remark \ref{remark_different_bregman_triangle}. 

$f$-divergence also defines a class of divergences based on convex functions \cite{ali1966general,pardo2018statistical,2006-f-divergence-and-pro, f-divergence-inequalities}. Many commonly used measures including the KL divergence, Jensen-Shannon divergence, and squared Hellinger distance are special cases of $f$-divergence. Many $f$-divergences are not proper distance metrics and do not satisfy the triangle inequality. KL divergence is the unique divergence belong to both $f$-divergence and Bregman divergence \cite{2009alpha-bregman}.

R\'{e}nyi divergence defines another class of divergences \cite{renyi1961measures, pardo2018statistical,Renyi-and-KL2014, bobkov2019renyi}. R\'{e}nyi divergence with order of 1 becomes KL divergence. As discussed in Remark \ref{remark_different_bregman_triangle}, R\'{e}nyi divergence also satisfies a generalized Pythogras theorem \cite{Renyi-and-KL2014}.

KL divergence between general distributions does not have a closed form. In application, it is not easy to estimate KL divergence when only samples of distributions are available especially in high dimensional problems. A line of research is dedicated to the estimation of divergences \cite{universal-estimation-div2006, estimation2009QingWang,estimate2010Nguyen, f-div-esti2012, gil2013renyi,  ensembleEstimation2014, NIPS2014f-div-esti,estimate-KL-minimax2018, Rubenstein2019estimation, minimax-estimate-kl2020, neilsen2021approJeff, neural-estimation-div-JMLR}. Unlike other distributions, the KL divergence between Gaussians has a closed form. The theorems presented in this paper can deepen our understanding of KL divergence between Gaussians.

The asymmetry of KL divergence has restricted the application of KL divergence in practical applications. Many other divergences have been investigated \cite{rached2001renyi, davis2007information, abou2012note, fGAN, improved_WGANs, towardsRepresentationDistance, cluster-bregman-div, 2018-new-family, 2019-robust-KL-divergence, pro-generalized-divergence-2020}.  Pardo gives a comprehensive survey on a wide range of statistical divergences in his book \cite{pardo2018statistical}. 



\section{Conclusion}\label{sec:conclusion}
In this paper, we research the properties of KL divergences between Gaussians. First, we find the supremum of reverse KL divergence $KL(\mathcal{N}_2||\mathcal{N}_1)$ if the forward KL divergence $KL(\mathcal{N}_1||\mathcal{N}_2)\leq \varepsilon$ ($\varepsilon>0$).  This conclusion quantifies the approximate symmetry of small KL divergence between Gaussians. We also find the infimum of $KL(\mathcal{N}_2||\mathcal{N}_1)$ if $KL(\mathcal{N}_1||\mathcal{N}_2)\geq M$ ($M>0$).  We give the conditions when the supremum and infimum can be attained.
Second, we find a bound for $KL(\mathcal{N}_1||\mathcal{N}_3)$ when $KL(\mathcal{N}_1||\mathcal{N}_2)$ and $KL(\mathcal{N}_2||\mathcal{N}_3)$ are bounded. This indicates that KL divergence between Gaussians follows a relaxed triangle inequality.
Importantly, all the bounds in the theorems in this paper are independent of the dimension of distributions.
The theorems presented in this paper is suitable especially for contexts where parameters may vary or can not be identified in advance (e.g., machine learning).
Finally, we discuss the applications of our theorems in deep anomaly detection and safe reinforcement learning.
We hope our research can shed light on more research in related field.
In the future, we plan to explore the properties of KL divergence between more general distributions such as Gaussian mixture models and exponential family of distributions.







\ifCLASSOPTIONcompsoc
  \section*{Acknowledgments}
\else
  \section*{Acknowledgment}
\fi
This work is supported by NSFC Program (No. 62002107, 62172429, 61872371).


\bibliographystyle{IEEEtran}
\bibliography{main}
%
%
%
\newpage

\appendices
\renewcommand\appendix{\par
	\setcounter{section}{0}
	\setcounter{subsection}{0}
	\gdef\thesection{附录 \Alph{section}}}
	
\section{Proof of Lemma \ref{thm:lemma_1_d_fx}}\label{sec:app_proof_thm:lemma_1_d_fx}
\begin{proof}
	\begin{enumerate}[label=(\alph*), ref=\ref{thm:lemma_1_d_fx}\alph*]
		\item This is because $	f'(x)=1-\frac{1}{x},\ f''(x)=\frac{1}{x^2}>0	$
		\item 
		We note $\Delta(x)=f(\frac{1}{x})-f(x)=\frac{1}{x}-x+2\log x$. Then $\Delta'(x) =-(\frac{1}{x}-1)^2\leq 0 \ \text{and}\ 
		\Delta(1)=0$
		So it is easy to know Lemma \ref{prp:f_x_invs_x_1_d} holds.
		\item 
		We can verify this by definition.
		\begin{align}\nonumber
			y-\log y=x  
			\Longleftrightarrow & e^{y-x}=y 
			\Longleftrightarrow  (-y)e^{-y}=-e^{-x} \\
			\Longleftrightarrow & y=-W(-e^{-x}) 
		\end{align}
		
		\item We can get Lemma \ref{thm:solution_f_x} from \ref{thm:reverse_f_def} immediately.
		
		\item According to Equation \eqref{equ:deriv_W}, we can have
		\begin{align}
			\dfrac{\dif w_1(t)}{\dif t}
			& = -\dfrac{\dif (W_{0}(-e^{-(1+t)}))}{\dif t}
			= \dfrac{-W_{0}(-e^{-(1+t)})}{-e^{-(1+t)}(1+W_{0}(-e^{-(1+t)}))} \nonumber\\
			&\times \dfrac{\dif (-e^{-(1+t)})}{\dif t} 
			= \dfrac{W_{0}(-e^{-(1+t)})}{W_{0}(-e^{-(1+t)})+1}
			= \dfrac{-w_1(t)}{1-w_1(t)}\nonumber
		\end{align}
		The derivative of $w_2(t)$ can be computed in a similar way.
		
		\item From Lemma \ref{prp:f_x_invs_x_1_d}, we can know Lemma \ref{thm:inequality_solution_fx}.
		
		\item 
		This is because 
		\begin{equation}
			\begin{aligned}\nonumber
				f(x)\leq 1+t
				\Longrightarrow & w_1(t) < x < w_2(t)
				\Longrightarrow & \dfrac{1}{w_2(t)} < \dfrac{1}{x} < \dfrac{1}{w_1(t)}
			\end{aligned}
		\end{equation}
		Combining Lemma \ref{prp:f_x_invs_x_1_d}, we have 
		\begin{align}\nonumber
			f(\dfrac{1}{w_2(t)}) < f(w_2(t))=1+t=f(w_1(t))<f(\dfrac{1}{w_1(t)})
		\end{align}
		Thus Equation \eqref{equ:sup_f_x_invs} holds. It is also easy to know that $S(t)=f(\frac{1}{w_1(t)})$ is continuous and strictly increasing with $t$.
		
		\item 
		We have
		\begin{equation}
			\begin{aligned}
				f(x)\geq 1+t
				\implies & x\leq w_1(t) \vee x\geq w_2(t)\nonumber\\
				\implies & \dfrac{1}{x}\leq \dfrac{1}{w_2(t)} \vee \dfrac{1}{x}\geq \dfrac{1}{w_1(t)}
			\end{aligned}
		\end{equation}
		Combining Lemma \ref{prp:f_x_invs_x_1_d}, we have $	f(\frac{1}{w_2(t)})< f(\frac{1}{w_1(t)})$, so we have Lemma \ref{prp:inf_f_x_invers_1_d}.
		%
		
		\item Since $f'(x)=1-\frac{1}{x}$ and $w_2(t)\geq 1$ for $t\geq 0$, we have
		\begin{align}
			f'(w_2(t))=&1-\dfrac{1}{w_2(t)}
			=  \dfrac{w_2(t)-1}{w_2(t)} \nonumber\\
			\leq & w_2(t)-1
			=  -\left(1-\dfrac{1}{\dfrac{1}{w_2(t)}}\right)
			=  -f'(\dfrac{1}{w_2(t)})
		\end{align}
		
		\item 
		\begin{align}
			& f(w_1(t_1)w_1(t_2))	\nonumber\\
			= &w_1(t_1)w_1(t_2)-\log w_1(t_1)w_1(t_2)\nonumber\\
			= &  w_1(t_1)w_1(t_2) + (w_1(t_1)-\log w_1(t_1)) \nonumber\\
			&+ (w_1(t_2)
			-\log w_1(t_2))-w_1(t_1)-w_1(t_2)\nonumber\\
			= & w_1(t_1)w_1(t_2) + 1 + t_1 + 1 + t_2  -w_1(t_1)-w_1(t_2) \nonumber\\
			= & t_1 + t_2 + 2+ w_1(t_1)w_1(t_2) - w_1(t_1)-w_1(t_2)
		\end{align}
		where the third equation follows from Lemma \ref{thm:solution_f_x}.
		Equation \eqref{equ:plugin_fw2w2} can be proved in a similar way.
	\end{enumerate}
	$\hfill\square$ 
\end{proof}

\section{Proof of Lemma \ref{thm:backward_to_forward_KL_small_indendent_n}}\label{sec:proof_thm:backward_to_forward_KL_small_indendent_n}

\begin{proof}
	
	The condition $KL(\mathcal{N}(0,I)||\mathcal{N}(\bm{\mu},\bm{\Sigma}))\leq \varepsilon$ is equal to the following conditions
	\begin{align}
		\log|\bm{\Sigma}|+\mathop{\mathrm{Tr}}(\bm{\Sigma}^{-1})&\leq n+\varepsilon_1 \label{equ:small_bKL_right_e}\\
		\bm{\mu}^{\top}\bm{\Sigma}^{-1}\bm{\mu}&\leq 2\varepsilon-\varepsilon_1 \label{equ:small_bKL_left_e}\\
		0\leq \varepsilon_1 & \leq 2\varepsilon
	\end{align}
	We can apply Lemma \ref{thm:lemma_sup_fx_invs_nd} on Equation \eqref{equ:small_bKL_right_e} and get 
	\begin{equation}\label{equ:small_bKL_to_fKL_left_item_bound}
		\begin{aligned}
			&-\log|\bm{\Sigma}|+\mathop{\mathrm{Tr}}(\bm{\Sigma})\nonumber\\
			\leq &  \dfrac{1}{-W_{0}(-e^{-(1+\varepsilon_1)})}-\log \dfrac{1}{-W_{0}(-e^{-(1+\varepsilon_1)})} +n -1
		\end{aligned}
	\end{equation}
	Applying Lemma \ref{prp:sup_f_x_invers_1_d} on Equation \eqref{equ:small_bKL_right_e}, we get 
	\begin{align}\label{equ:bound_lambda_prime}
		w_1(\varepsilon_1)<\lambda'<w_2(\varepsilon_1)
	\end{align}
	
	From Equation \eqref{equ:small_bKL_left_e} we know  $\bm{\mu}^{\top}\bm{\Sigma}^{-1}\bm{\mu}\leq 2\varepsilon-\varepsilon_1$. Since $\bm{\mu}^{\top}\bm{\Sigma}^{-1}\bm{\mu}\geq \lambda'_*\bm{\mu}^{\top}\bm{\mu}$ where $\lambda'_*$ is the minimum eigenvalue of $\bm{\Sigma}^{-1}$, combining Equation \eqref{equ:bound_lambda_prime}, we can know 
	\begin{gather}\label{equ:bound_mumu_bKL_to_fKL_right_item_bound}
		\bm{\mu}^{\top}\bm{\mu}\leq \dfrac{2\varepsilon-\varepsilon_1}{\lambda'_*}\leq \dfrac{2\varepsilon-\varepsilon_1}{w_1(\varepsilon_1)}
	\end{gather}
	
	Adding the two sides of Inequalities \eqref{equ:small_bKL_to_fKL_left_item_bound}, and \eqref{equ:bound_mumu_bKL_to_fKL_right_item_bound}, we get the same result as Equation \eqref{equ:bound_bKL}.
	Therefore, we can get the same supremum  as follows.
	\begin{equation}\label{equ:supremum_b_to_f_kl}
		\begin{aligned}
			&KL(\mathcal{N}(\bm{\mu},\bm{\Sigma})||\mathcal{N}(0,I)) \nonumber\\
			\leq & \dfrac{1}{2}\left(
			\dfrac{1}{-W_{0}(-e^{-(1+2\varepsilon)})}-\log \dfrac{1}{-W_{0}(-e^{-(1+2\varepsilon)})} -1 \right) \\
		\end{aligned}
	\end{equation}
	Inequality \eqref{equ:supremum_b_to_f_kl} is tight only when there exists one $\lambda'_j=-W_{0}(-e^{-(1+2\varepsilon)})$ and all other $\lambda'_i=1$ for $i\neq j$, and $|\bm{\mu}|=0$.
	
	$\hfill\square$ 
\end{proof}

\section{Proof of Theorem \ref{thm:duality_small_KL_general}}\label{sec:app_proof_thm:duality_small_KL_general}
\begin{proof}
	
	For  $X \sim \mathcal{N}(\bm{\mu},\bm{\Sigma})$, there exists an invertible matrix $B$ such that $X'=B^{-1}(X-\bm{\mu})\sim \mathcal{N}(0,I)$ \cite{PRML}. Here $B=PD^{1/2}$, $P$ is an orthogonal matrix whose columns are the eigenvectors of $\bm{\Sigma}$, $D=diag(\lambda_1,\dots,\lambda_n)$ whose diagonal elements are the corresponding eigenvalues. 
	For $X_1\sim \mathcal{N}(\bm{\mu}_1,\bm{\Sigma}_1)$ and $X_2\sim \mathcal{N}(\bm{\mu}_2,\bm{\Sigma}_2)$, we define the following linear transformations $T_1$, $T_2$
	\begin{align}\label{equ:linear_tranform_for_symmetry}
		X^1_1=T_1(X_1)=B_1^{-1}(X_1-\bm{\mu}_1)\ \text{such that} \ X^1_1\sim \mathcal{N}(0,I)\\
		X^2_2=T_2(X_2)=B_2^{-1}(X_2-\bm{\mu}_2)\ \text{such that} \ X^2_2\sim \mathcal{N}(0,I)
	\end{align}
	and the reverse transformations $T^{-1}_1$, $T^{-1}_2$ such that 
	$X_1=T^{-1}_1(X^1_1)=B_1 X^1_1+\bm{\mu}_1$ and $X_2=T^{-1}_2(X^2_2)=B_2 X^2_1 + \bm{\mu}_2$,
	where $p_{X_1^1}=p_{X_2^2}=\mathcal{N}(0,I)$. Besides, it is easy to know $	X^2_1=T_2(X_1)=B_2^{-1}(X_1-\bm{\mu}_2)$ and $X^1_2=T_1(X_2)=B_1^{-1}(X_2-\bm{\mu}_1)$
	are both Gaussian variables.
	We also have 
	\begin{align}
		X^2_1\sim \mathcal{N}(B_2^{-1}(\bm{\mu}_1-\bm{\mu}_2), B_2^{-1}\bm{\Sigma}_1(B_2^{-1})^\top)\\
		X^1_2\sim \mathcal{N}(B_1^{-1}(\bm{\mu}_2-\bm{\mu}_1), B_1^{-1}\bm{\Sigma}_2(B_1^{-1})^\top)
	\end{align}
	
	With the help of invertible linear transformations, we can convert the KL divergence between two arbitrary Gaussians into that between one Gaussian and standard Gaussian.
	According to Proposition \ref{thm:preserve_KL}, diffeomorphisms preserve KL divergence.
	If we apply $T_2$ simultaneously on $X_1$, $X_2$, we can have
	\begin{align}
		KL(\mathcal{N}(\bm{\mu}_1,\bm{\Sigma}_1)||\mathcal{N}(\bm{\mu}_2,\bm{\Sigma}_2))	= & KL(p_{X_1^2}||p_{X^2_2}) \nonumber \\
		= & KL(p_{X_1^2}||\mathcal{N}(0,I))
	\end{align}
	Then we can apply $T_2^{-1}$ on $X_1^2$, $X_2^2$ and also have
	\begin{align}
		KL(\mathcal{N}(0,I)||p_{X_1^2}) = & KL(p_{X^2_2}||p_{X_1^2}) \\
		=& KL(\mathcal{N}(\bm{\mu}_2,\bm{\Sigma}_2)||\mathcal{N}(\bm{\mu}_1,\bm{\Sigma}_1))
	\end{align}
	According to precondition, it is easy to know $KL(\mathcal{N}(\bm{\mu}_1,\bm{\Sigma}_1)||\mathcal{N}(\bm{\mu}_2,\bm{\Sigma}_2)) 
	= KL(p_{X_1^2}||\mathcal{N}(0,I))$.
	Applying Theorem \ref{thm:forward_to_backward_KL_small_indendent_n} on $KL(p_{X_1^2}||\mathcal{N}(0,I))$, we can prove  
	\begin{equation}	
		\begin{aligned}\label{equ:small_KL_general_bound_1}
			&KL(\mathcal{N}(0,I)||p_{X_1^2})
			= KL(\mathcal{N}(\bm{\mu}_2,\bm{\Sigma}_2)||\mathcal{N}(\bm{\mu}_1,\bm{\Sigma}_1))\nonumber \\
			\leq & \dfrac{1}{2}\left(
			\dfrac{1}{-W_{0}(-e^{-(1+2\varepsilon)})}-\log \dfrac{1}{-W_{0}(-e^{-(1+2\varepsilon)})} -1 \right) 
		\end{aligned}
	\end{equation}
	Similarly, if we use $T_1$ simultaneously on $X_1$ and $X_2$, we can get the same result.

	Inequality \eqref{equ:small_KL_general_bound_1} is tight when there exists only one eigenvalue $\lambda_j$ of $B_2^{-1}\bm{\Sigma}_1(B_2^{-1})^\top$ or $B_1^{-1}\bm{\Sigma}_2(B_1^{-1})^\top$ is equal to $-W_{0}(-e^{-(1+2\varepsilon)})$ and all other eigenvalues $\lambda_i$ ($i\neq j$) are equal to $1$, and $\bm{\mu}_1=\bm{\mu}_2$.
	
	$\hfill\square$ 
\end{proof}

\section{Proof of Theorem \ref{thm:symmetry_KL_simplify}} \label{sec:appendix_proof_symmetry_kl_simplify}
\begin{proof}
	
	When $\varepsilon$ is small, we can use the series expanding $W_0$ (see Section III.17 in \cite{2016PrincetonCompanion}) to simplify the bound in Theorem \ref{thm:duality_small_KL_general}.
	
	Notice that when $\varepsilon$ is small, $-W_0(-e^{-(1+2\varepsilon)})$ is close to 1. According to the series expanding $W_0$ (see Section III.17 in \cite{2016PrincetonCompanion}), we have
	\begin{align}\label{equ:expand_W0}
		W_0(-e^{-(1+2\varepsilon)})=-1+2\sqrt{\varepsilon}-\dfrac{4}{3}\varepsilon+\dfrac{2}{9}\varepsilon^{1.5}+O(\varepsilon^{2}) 
	\end{align} 
	Now expand the $\log$ term around $-W_0(-e^{-(1+2\varepsilon)})=1$ using Taylor series for small $\varepsilon$.
	\begin{align}
		& \log (-W_0(-e^{-(1+2\varepsilon)}))\nonumber\\
		= & \log (1-W_0(-e^{-(1+2\varepsilon)})-1) \nonumber\\
		= & -W_0(-e^{-(1+2\varepsilon)})-1-\dfrac{1}{2}\left(-W_0(-e^{-(1+2\varepsilon)})-1\right)^2 \nonumber\\
		&+ \dfrac{1}{3}\left(-W_0(-e^{-(1+2\varepsilon)})-1\right)^3 + O\left((-W_0(-e^{-(1+2\varepsilon)})-1)^4\right)  \nonumber\\
		= & -2\sqrt{\varepsilon}+\dfrac{4}{3}\varepsilon-\dfrac{2}{9}\varepsilon^{1.5}+O(\varepsilon^{2})\nonumber\\
		&-\dfrac{1}{2}\left(-2\sqrt{\varepsilon}+\dfrac{4}{3}\varepsilon-\dfrac{2}{9}\varepsilon^{1.5}+O(\varepsilon^{2})\right)^2 \nonumber \\
		& + \dfrac{1}{3} \left(-2\sqrt{\varepsilon}+\dfrac{4}{3}\varepsilon-\dfrac{2}{9}\varepsilon^{1.5}+O(\varepsilon^{2})\right)^3  + O(\varepsilon^{2})  \\
		=&  -2\sqrt{\varepsilon}-\dfrac{2}{3}\varepsilon-\dfrac{2}{9}\varepsilon^{1.5} + O(\varepsilon^2) \label{equ:symmetry_bound_simplify_logpart}
	\end{align} 
	
	Plugging Equation \eqref{equ:expand_W0} and \eqref{equ:symmetry_bound_simplify_logpart} into the bound in Theorem \ref{thm:duality_small_KL_general}, we can have 
	\begin{align}
		& KL(\mathcal{N}(\bm{\mu}_2,\bm{\Sigma}_2)||\mathcal{N}(\bm{\mu}_1,\bm{\Sigma}_1)) \nonumber\\ 
		\leq &\dfrac{1}{2}\left(
		\dfrac{1}{1-2\sqrt{\varepsilon}+\dfrac{4}{3}\varepsilon-\dfrac{2}{9}\varepsilon^{1.5}+O(\varepsilon^{2})} \right. \nonumber\\
		&\left. \vphantom{\dfrac{1}{1-2\sqrt{\varepsilon}+\dfrac{4}{3}\varepsilon-\dfrac{2}{9}\varepsilon^{1.5}+O(\varepsilon^{2})}} + \left(-2\sqrt{\varepsilon}-\dfrac{2}{3}\varepsilon-\dfrac{2}{9}\varepsilon^{1.5} + O(\varepsilon^2)\right) -1 \right) \nonumber\\
		= &  \dfrac{1}{2}\dfrac{2\varepsilon-\dfrac{4}{3}\varepsilon^{1.5}+O(\varepsilon^2)}{1-2\sqrt{\varepsilon}+\dfrac{4}{3}\varepsilon-\dfrac{2}{9}\varepsilon^{1.5}+O(\varepsilon^{2})}  \nonumber\\
		= & \varepsilon + \dfrac{2\varepsilon^{1.5}+O(\varepsilon^2)}{1-2\sqrt{\varepsilon}+\dfrac{4}{3}\varepsilon-\dfrac{2}{9}\varepsilon^{1.5}+O(\varepsilon^{2})}  \nonumber\\
		= & \varepsilon + 2\varepsilon^{1.5} + O(\varepsilon^2) \label{equ:simplify_symmetry_bould}
	\end{align}
	
	$\hfill\square$ 
\end{proof}

\section{The first proof of Theorem \ref{thm:duality_big_KL_general}}\label{app:big_KL_general_proof1}


Theorem \ref{thm:duality_big_KL_general} can be proved using the similar method as that of Theorem \ref{thm:duality_small_KL_general}, except that the proof uses $W_{-1}$.
We put the key steps of the proof of Theorem \ref{thm:duality_big_KL_general} in Lemma \ref{thm:lemma_inf_fx_inv_nd} and Lemma \ref{thm:duality_big_KL}.
\begin{lemma}\label{thm:lemma_inf_fx_inv_nd}
	Given $n$-ary function $\bar{\bm{f}}(\bm{x})=\bar{\bm{f}}(x_1,\dots,x_n)=\sum_{i=1}^{n}x_i-\log x_i	\ (x_i\in \mathbb{R}^{++})$, if
	$\bar{\bm{f}}(x_1,\dots,x_n)\geq n+M (M>0)$, then 
	\begin{align}\label{equ:n_ary_f_inverse_inf}			
		&\inf \bar{\bm{f}}\big(\dfrac{1}{x_1},\dots,\dfrac{1}{x_n}\big)\nonumber\\
		= &\dfrac{1}{-W_{-1}(-e^{-(1+M)})}-\log \dfrac{1}{-W_{-1}(-e^{-(1+M)})} +n-1\nonumber
	\end{align}
\end{lemma}

\begin{proof}
	
	The structure of proof of Lemma \ref{thm:lemma_inf_fx_inv_nd} is similar to that of Lemma \ref{thm:lemma_sup_fx_invs_nd}.
	The  constraint in the following optimization problem 
	\begin{align}
		\text{minimize}\ & \bar{\bm{f}}\big(\dfrac{1}{x_1},\dots,\dfrac{1}{x_n}\big)\\
		\text{\textit{s.t.}}\ &   \sum_{i=1}^{n}x_i-\log x_i\geq n+M
	\end{align}
	%
	%
	can be replaced by the following constraints
	\begin{gather}
		\bigwedge_{i=1}^{n}f(x_i)=x_i-\log x_i\geq 1+M_i \wedge
		\bigwedge_{i=1}^{n}M_i\geq 0 \wedge \sum_{i=1}^{n}M_i\geq M
	\end{gather}
	Given fixed $M_1,\dots,M_n$ such that $\bigwedge_{i=1}^{n}M_i\geq 0 \wedge \sum_{i=1}^{n}M_i\geq M$, we define 
	\begin{align}\label{equ:inf_all_d_factor_multi_inf}
		\bar{\bm{I}}(M_1,\dots,M_n)
		=&\inf\limits_{\bigwedge_{i=1}^{n}f(x_i)\geq 1+M_i} \bar{\bm{f}}\big(\dfrac{1}{x_1},\dots,\dfrac{1}{x_n}\big)\nonumber\\
		=&\sum_{i=1}^{n} \inf\limits_{f(x_i)\geq 1+M_i} f\big(\dfrac{1}{x_i}\big)
		= \sum_{i=1}^{n}I(M_i)
	\end{align}
	So we have 
	\begin{align}\label{equ:inf_factor}
		\inf \bar{\bm{f}}\big(\dfrac{1}{x_1},\dots,\dfrac{1}{x_n}\big) = \inf\limits_{\bigwedge_{i=1}^{n}M_i\geq 0 \atop \sum_{i=1}^{n}M_i\geq M} \bar{\bm{I}}(M_1,\dots,M_n)
	\end{align}	
	It is easy to know that $\bar{\bm{I}}(M_1,\dots,M_n)$ is continuous and strictly increasing with $M_1,\dots,M_n$. So the condition $\sum_{i=1}^{n}M_i\geq M$ in Equation \eqref{equ:inf_factor} can be changed to $\sum_{i=1}^{n}M_i=M$.
	
	The remaining proof consists of two steps.  We find $\bar{\bm{I}}(M_1,\dots,M_n)$ for fixed $M_1,\dots,M_n$ in the first step, and then find $\inf \bar{\bm{I}}(M_1,\dots,M_n)$ for any  $M_1,\dots,M_n$ satisfying $\bigwedge_{i=1}^{n}M_i\geq 0 \wedge \sum_{i=1}^{n}M_i= M$ in the second step.
	
	\textbf{Step 1}: 
	According to Lemma \ref{prp:inf_f_x_invers_1_d}, for fixed $M_i$, we get 
	\begin{align}
		I(M_i)=\inf\limits_{f(x)\geq1+M_i} f\big(\dfrac{1}{x}\big)=f\big(\dfrac{1}{w_2(M_i)}\big)
	\end{align}
	Combining Equation \eqref{equ:inf_all_d_factor_multi_inf}, we know 
	\begin{gather}\label{equ:def_infi_fixed_Mi}
		\bar{\bm{I}}(M_1,\dots,M_n) = \sum_{i=1}^{n}f\big(\dfrac{1}{w_2(M_i)}\big)
	\end{gather}
	
	\textbf{Step 2}: We define function 
	\begin{align}\label{equ:def_delta_M}
		\Delta(M)  = & f(w_2(M))-f\big(\dfrac{1}{w_2(M)}\big)\nonumber\\
		= & w_2(M)-\dfrac{1}{w_2(M)}-2\log w_2(M)
	\end{align}
	Similarly, we can prove $\Delta(tM)\leq t\Delta(M)\ (0\leq t<1)$
	by showing $\Delta(0)=0$ (apparently) and $\Delta(M)$ is strictly increasing and strictly  convex.
	Combining Lemma \ref{thm:deriv_w1_w2}, we get the derivative of $\Delta(M)$ as
	\begin{align}
		\dfrac{\dif \Delta(M)}{\dif M}
		= \left(1+\dfrac{1}{w_2(M)^2}-\dfrac{2}{w_2(M)}\right)\times \dfrac{\dif w_2(M)}{\dif M}
		= 1-\dfrac{1}{w_2(M)}
	\end{align}
	The second order derivative of $\Delta(M)$ is 
	\begin{align}
		\dfrac{\dif^2 \Delta(M)}{\dif M^2}
		=  \dfrac{1}{w_2(M)^2}\times \dfrac{w_2(M)}{w_2(M)-1}
		=  \dfrac{1}{w_2(M)(w_2(M)-1)}
	\end{align}
	Since $w_2(M)\in (1,+\infty)$ for $M>0$, so it is easy to know $\frac{\dif \Delta(M)}{\dif M}>0, \frac{\dif^2 \Delta(M)}{\dif M^2}>0$ for $M>0$.
	This implies that $\Delta(M)$ is strictly increasing and strictly  convex. 
	We can use the similar deduction as Lemma \ref{thm:lemma_sup_fx_invs_nd} to prove $\Delta(tM)\leq t\Delta(M)$.		Thus, we have 
	\begin{align}
		\small
		\bar{\bm{\Delta}}(M_1,\dots,M_n)=&\sum_{i=1}^{n}f(w_2(M_i))-f\big(\dfrac{1}{w_2(M_i)}\big)
		= \sum_{i=1}^{n}\Delta(M_i) \nonumber\\
		= &\sum_{i=1}^{n}\Delta(\dfrac{M_i}{M}M)
		\leq  \sum_{i=1}^{n}\dfrac{M_i}{M}\Delta(M)
		= \Delta(M) \label{equ:sup_Delta_M}
	\end{align}
	Inequality \eqref{equ:sup_Delta_M} is tight when there exists only one $M_j=M$ and all other $M_i=0$ for $i\neq j$.
	Therefore, from Inequality \eqref{equ:sup_Delta_M}, we can obtain
	\begin{DispWithArrows}
		&\bar{\bm{I}}(M_1,\dots,M_n)\Arrow{ \eqref{equ:def_infi_fixed_Mi}}	\nonumber\\
		= & \sum_{i=1}^{n}f\big(\dfrac{1}{w_2(M_i)}\big) \Arrow{\eqref{equ:sup_Delta_M}} \nonumber \\
		\geq &  \sum_{i=1}^{n}f(w_2(M_i))-\Delta(M) \Arrow{$\overset{\text{Lemma \ref{thm:solution_f_x}}}{\text{\eqref{equ:def_delta_M}}}$} \nonumber \\
		= & \sum_{i=1}^{n}(1+M_i) -\big(f(w_2(M))-f\big(\dfrac{1}{w_2(M)}\big)\big)  \nonumber \\
		= &  n+M-(1+M)+f\big(\dfrac{1}{w_2(M)}\big) \nonumber\\
		= & \dfrac{1}{-W_{-1}(-e^{-(1+M)})}-\log \dfrac{1}{-W_{-1}(-e^{-(1+M)})} \nonumber\\
		&+n-1
	\end{DispWithArrows}
	Finally, we can conclude that 
	\begin{align}
		&\inf \bar{\bm{f}}\big(\dfrac{1}{x_1},\dots,\dfrac{1}{x_n}\big)
		= \inf\limits_{\bigwedge_{i=1}^{n}M_i\geq 0 \atop \sum_{i=1}^{n}M_i= M} \bar{\bm{I}}(M_1,\dots,M_n)  \nonumber\\
		= &\dfrac{1}{-W_{-1}(-e^{-(1+M)})}-\log \dfrac{1}{-W_{-1}(-e^{-(1+M)})} +n-1	
	\end{align}
	Similarly, $\bar{\bm{f}}(1/x_1,\dots,1/x_n)$ reaches its infimum when there exists only one $j$ such that $x_j=-W_{-1}(-e^{-(1+M)})$ and $f(x_i)=1$ for $i\neq j$. 
	$\hfill\square$ 
\end{proof}	

The following Lemma \ref{thm:duality_big_KL}  gives the infimum of KL divergence when one Gaussian is standard. 
\begin{lemma}\label{thm:duality_big_KL}
	For any  $n$-dimensional Gaussian distribution $\mathcal{N}(\bm{\mu},\bm{\Sigma})$,
	\begin{enumerate}[label=(\alph*), ref=\ref{thm:duality_small_KL}\alph*]
		\setlength{\itemsep}{0pt}
		\setlength{\parskip}{0pt}
		
		\item
		\label{thm:forward_to_backward_KL_large_indendent_n}
		If $KL(\mathcal{N}(\bm{\mu},\bm{\Sigma})||\mathcal{N}(0,I))\geq  M (M>0)$, then
		\begin{equation}
			\begin{aligned}
				&KL(\mathcal{N}(0,I)||\mathcal{N}(\bm{\mu},\bm{\Sigma}))\nonumber\\
				\geq & \dfrac{1}{2}\left(\dfrac{1}{-W_{-1}(-e^{-(1+2M)})}-\log \dfrac{1}{-W_{-1}(-e^{-(1+2M)})} -1\right)
			\end{aligned}
		\end{equation}
		
		\item \label{thm:backward_to_forward_KL_large_indendent_n}
		If $KL(\mathcal{N}(0,I)||\mathcal{N}(\bm{\mu},\bm{\Sigma}))\geq  M$, then
		\begin{equation}
			\begin{aligned}
				&KL(\mathcal{N}(\bm{\mu},\bm{\Sigma})||\mathcal{N}(0,I))\nonumber\\
				\geq & \dfrac{1}{2}\left(\dfrac{1}{-W_{-1}(-e^{-(1+2M)})}-\log \dfrac{1}{-W_{-1}(-e^{-(1+2M)})} -1\right)
			\end{aligned}
		\end{equation}
		%
	\end{enumerate}
\end{lemma}

\begin{proof}
	(a) We first consider the case when $KL(\mathcal{N}(\bm{\mu},\bm{\Sigma})||\mathcal{N}(0,I))=M$. At the end of the proof, we deal with the case when $KL(\mathcal{N}(\bm{\mu},\bm{\Sigma})||\mathcal{N}(0,I))\geq  M$.
	
	The condition $-\log|\bm{\Sigma}|+\mathop{\mathrm{Tr}}(\bm{\Sigma})+\bm{\mu}^{\top}\bm{\mu}-n= 2M$
	is equal to 
	\begin{align}
		-\log|\bm{\Sigma}|+\mathop{\mathrm{Tr}}(\bm{\Sigma})=\sum_{i=1}^{n}\lambda_i-\log \lambda_i= n+M_1\label{equ:item_1_large}\\ 
		\bm{\mu}^{\top}\bm{\mu}= 2M-M_1\label{equ:item_2_large}
	\end{align}
	where $0\leq M_1\leq 2M$.
	
	Applying Lemma \ref{thm:lemma_inf_fx_inv_nd} on Equation \eqref{equ:item_1_large},
	we can get
	\begin{equation}
		\begin{aligned}\label{equ:item_1_revers_inf}
			&\log|\bm{\Sigma}|+\mathop{\mathrm{Tr}}(\bm{\Sigma}^{-1}) 
			= \sum_{i=1}^{n}\dfrac{1}{\lambda_i}-\log \dfrac{1}{\lambda_i} \nonumber\\
			\geq & \dfrac{1}{-W_{-1}(-e^{-(1+M_1)})}-\log \dfrac{1}{-W_{-1}(-e^{-(1+M_1)})} +n-1
		\end{aligned}
	\end{equation}
	Inequality \eqref{equ:item_1_revers_inf} is tight when all eigenvalues $\lambda_i$ of $\bm{\Sigma}$ are equal to 1 except for one $\lambda_j=-W_{-1}(-e^{-(1+M_1)})$.
	
	From Equation \eqref{equ:item_2_large}, we know $\bm{\mu}^{\top}\bm{\Sigma}^{-1}\bm{\mu}\geq \lambda_*'\bm{\mu}^{\top}\bm{\mu}= \lambda_*'(2M-M_1)$ where $\lambda'_*$ is the smallest eigenvalue of $\bm{\Sigma}^{-1}$.
	Here $\lambda^*=1/\lambda'_*$ is the largest eigenvalue of $\bm{\Sigma}$. 
	From Equation \eqref{equ:item_1_large}, Lemma \ref{equ:f_convex_minimum} and \ref{prp:sup_f_x_invers_1_d}, we know $\lambda^*\leq -W_{-1}(-e^{-(1+M_1)})$.	So we obtain
	\begin{align}\label{equ:item2_inf}
		\bm{\mu}^{\top}\bm{\Sigma}^{-1}\bm{\mu}\geq \dfrac{2M-M_1}{-W_{-1}(-e^{-(1+M_1)})}
	\end{align}
	
	Note that, inequalities \eqref{equ:item2_inf} and \eqref{equ:item_1_revers_inf} become tight simultanously when the same condition holds. Now combining Equation \eqref{equ:item_1_revers_inf} and \eqref{equ:item2_inf}, we obtain
	\begin{align} \label{equ:inf_on_hM}
		&\log|\bm{\Sigma}|+\mathop{\mathrm{Tr}}(\bm{\Sigma}^{-1})+\bm{\mu}^{\top}\bm{\Sigma}^{-1}\bm{\mu}-n\nonumber \\
		\geq &  \dfrac{1}{-W_{-1}(-e^{-(1+M_1)})}-\log \dfrac{1}{-W_{-1}(-e^{-(1+M_1)})}\nonumber\\
		&+  \dfrac{2M-M_1}{-W_{-1}(-e^{-(1+M_1)})} -1 \nonumber\\
		= & \dfrac{1+2M-M_1}{w_2(M_1)}-\log \dfrac{1}{w_2(M_1)}-1
		=  L(M_1)\ (0\leq M_1 \leq 2M) 
	\end{align}
	It is easy to know that $L'(M_1)=\frac{M_1-2M}{w_2(M_1)(w_2(M_1)-1)}$.
	Since $M_1\leq 2M$ and $w_2(M_1)>1$ for $M_1>0$, so $L'(M_1)<0\ (M_1>0)$.
	This indicates that $L(M_1)>L(2M)$ for $0<M_1<2M$. Thus, we can conclude
	\begin{equation} \label{equ:inf_bKL}
		\begin{aligned}
			& KL(\mathcal{N}(0,I)||\mathcal{N}(\bm{\mu},\bm{\Sigma}))\nonumber\\
			\geq &\dfrac{1}{2}L(2M)\nonumber\\
			= 	 & \dfrac{1}{2}\left(\dfrac{1}{-W_{-1}(-e^{-(1+2M)})}-\log \dfrac{1}{-W_{-1}(-e^{-(1+2M)})}
			-1\right) 
		\end{aligned}
	\end{equation}
	
	Inequality \eqref{equ:inf_bKL} is tight when there exist only one eigenvalue $\lambda_j$ of $\bm{\Sigma}$  equal to $-W_{-1}(-e^{-(1+2M)})$ and all other eigenvalues $\lambda_i (i\neq j)$ are equal to 1, and $\bm{\mu}=0$.
	
	Finally, we can consider the case when $KL(\mathcal{N}(\bm{\mu},\bm{\Sigma})||\mathcal{N}(0,I))\geq  M$.
	The bound in Equation \eqref{equ:inf_bKL} is strictly increasing with $M$.
	Therefore, the precondition $KL(\mathcal{N}(\bm{\mu},\bm{\Sigma})||\mathcal{N}(0,I))=  M$ can be changed to $KL(\mathcal{N}(\bm{\mu},\bm{\Sigma})||\mathcal{N}(0,I))\geq  M$.

	(b) The proof of Lemma \ref{thm:backward_to_forward_KL_large_indendent_n} is the similar to that of Lemma \ref{thm:forward_to_backward_KL_large_indendent_n}. We list it here for clarity.
	
	The condition $KL(\mathcal{N}(0,I)||\mathcal{N}(\bm{\mu},\bm{\Sigma}))= M$ is equal to 
	\begin{align}
		\log|\bm{\Sigma}|+\mathop{\mathrm{Tr}}(\bm{\Sigma}^{-1})=n+M_1\label{equ:bKL_item1_large}\\
		\bm{\mu}^{\top}\bm{\Sigma}^{-1}\bm{\mu}=2M-M_1\label{equ:bKL_item2_large}
	\end{align}
	where $0\leq M_1\leq 2M$.
	Applying Lemma \ref{thm:lemma_inf_fx_inv_nd} on Equation \eqref{equ:bKL_item1_large}, we can obtain
	\begin{equation}
		\begin{aligned}\label{equ:bKL_item_1_revers_inf}
			&-\log|\bm{\Sigma}|+\mathop{\mathrm{Tr}}(\bm{\Sigma})\nonumber\\
			\geq & \dfrac{1}{-W_{-1}(-e^{-(1+M_1)})}-\log \dfrac{1}{-W_{-1}(-e^{-(1+M_1)})}
			+n-1
		\end{aligned}
	\end{equation}
	
	From Equation \eqref{equ:bKL_item1_large} and Lemma \ref{equ:f_convex_minimum} and \ref{prp:sup_f_x_invers_1_d},  we have $\lambda'\leq -W_{-1}(-e^{-(1+M_1)})$ where $\lambda'$ is the eigenvalue of $\bm{\Sigma}^{-1}$.
	Now let $\lambda'^*$ be the largest eigenvalues of $\bm{\Sigma}^{-1}$. It is easy to know
	\begin{align}\label{equ:bKL_large_item2_inf}
		&\lambda'^*\bm{\mu}^{\top}\bm{\mu}\geq \bm{\mu}^{\top}\bm{\Sigma}^{-1}\bm{\mu}=2M-M_1\nonumber\\
		\Longrightarrow & \bm{\mu}^{\top}\bm{\mu} \geq \dfrac{2M-M_1}{-W_{-1}(-e^{-(1+M_1)})}
	\end{align}
	Inequalities \eqref{equ:bKL_item_1_revers_inf} and \eqref{equ:bKL_large_item2_inf} are tight simutanously when there exist only one eigenvalue $\lambda'_j=-W_{-1}(-e^{-(1+M_1)})$ and all other eigenvalues are equal to 1, and $|\bm{\mu}|=0$. Therefore, combining Equation \eqref{equ:bKL_item_1_revers_inf} and  \eqref{equ:bKL_large_item2_inf}, we obtain
	\begin{align}
		&-\log|\bm{\Sigma}|+\mathop{\mathrm{Tr}}(\bm{\Sigma})+\bm{\mu}^{\top}\bm{\mu}-n\nonumber \\
		\geq &\dfrac{1}{-W_{-1}(-e^{-(1+M_1)})}-\log \dfrac{1}{-W_{-1}(-e^{-(1+M_1)})} \nonumber\\
		&+ \dfrac{2M-M_1}{-W_{-1}(-e^{-(1+M_1)})}-1 
	\end{align}
	Finally, using the similar analysis as Equation \eqref{equ:inf_bKL}, we can conclude that 
	\begin{equation} \label{equ:inf_fKL}
		\begin{aligned}
			&KL(\mathcal{N}(\bm{\mu},\bm{\Sigma})||\mathcal{N}(0,I))\nonumber\\
			\geq & \dfrac{1}{2}\left(\dfrac{1}{-W_{-1}(-e^{-(1+2M)})}-\log \dfrac{1}{-W_{-1}(-e^{-(1+2M)})} \right.
			\left. \vphantom{\dfrac{1}{-W_{-1}(-e^{-(1+2M)})}}-1\right) 
		\end{aligned}
	\end{equation}
	Similarly, the precondition $KL(\mathcal{N}(0,I)||\mathcal{N}(\bm{\mu},\bm{\Sigma}))= M$ can be changed to $KL(\mathcal{N}(0,I)||\mathcal{N}(\bm{\mu},\bm{\Sigma}))\geq  M$ because the bound in Equation \eqref{equ:inf_fKL} is strictly increasing with $M$.
	
	$\hfill\square$ 
\end{proof}

\textbf{Notes}.  
It needs strict conditions to reach the infimum in Lemma \ref{thm:duality_big_KL}. 

Now we can also obtain Theorem \ref{thm:duality_big_KL_general} on two general Gaussians. We can use linear transformation on Gaussians and apply Lemma \ref{thm:duality_big_KL} on them as what we do in the main proof of Theorem \ref{thm:duality_small_KL_general}. The key steps have been proven in Lemma \ref{thm:lemma_inf_fx_inv_nd} and \ref{thm:duality_big_KL}. More details are ommited.

\section{Proof of Lemma \ref{thm:w2_farther_from_1_than_w1}}\label{sec:app_proof_thm:w2_farther_from_1_than_w1}
\begin{proof}
	
	With the helper functions $f_l, f_r$ (Equation \eqref{equ:flfr}), we define function 
	\begin{align}
		\Delta_w(\varepsilon)=g_r(\varepsilon)-g_l(\varepsilon)=(w_2(\varepsilon)-1)-(1-w_1(\varepsilon))
	\end{align}
	It is straightforward to know 
	\begin{align}\label{equ:delta_w_e_0}
		\Delta_w(0)=w_2(0)-1-(1-w_1(0))=0
	\end{align}
	In the following, we prove $\Delta_w{'}(\varepsilon)>0$ for $\varepsilon>0$.
	Plugging Equation \eqref{equ:deriv_flr}, we have 
	\begin{align}
		g_r'(\varepsilon)=&f_r^{-1}{'}(\varepsilon) = \dfrac{1}{f_r'(f_r^{-1}(\varepsilon))}
		=\dfrac{1}{f_r'(w_2(\varepsilon)-1)} \nonumber\\
		= & \dfrac{1}{1-\dfrac{1}{w_2(\varepsilon)}}
		=  \dfrac{1}{f'(w_2(\varepsilon))}\\
		g_l'(\varepsilon)=&f_l^{-1}{'}(\varepsilon)= \dfrac{1}{f_l'(f_l^{-1}(\varepsilon))}
		=   \dfrac{1}{f_l'(1-w_1(\varepsilon))} \nonumber\\
		=&  \dfrac{1}{\dfrac{1}{w_1(\varepsilon)}-1}
		=  \dfrac{1}{-f'(w_1(\varepsilon))}
	\end{align}
	According to Lemma \ref{thm:lemma_1_d_fx} and Lemma \ref{thm:inequality_solution_fx}, $f(x)$ is strictly  decreasing in $(0,1)$ and $f(w_1(\varepsilon))>f(\frac{1}{w_2(\varepsilon)})$. So we can know $w_1(\varepsilon)<\frac{1}{w_2(\varepsilon)}$. Since $f(x)$ is convex and $f'(x)<0$ in $(0,1)$, we can know 
	$f'(w_1(\varepsilon))<f'(\frac{1}{w_2(\varepsilon)})$. 
	Now combining Lemma \ref{thm:deriv_w2_inv_w2_compare}, we can obtain  $0<f'(w_2(\varepsilon))\leq -f'(\frac{1}{w_2(\varepsilon)})< -f'(w_1(\varepsilon))\ (\varepsilon>0)$. 
	This leads to
	\begin{align}
		g_r'(\varepsilon) =\dfrac{1}{f'(w_2(\varepsilon))}>\dfrac{1}{-f'(w_1(\varepsilon))}=g_l'(\varepsilon)
	\end{align}
	for $\varepsilon>0$, which means $\Delta_w{'}(\varepsilon)=g_r'(\varepsilon)-g_l'(\varepsilon)>0\ (\varepsilon>0)$.
	Now combining Equation \eqref{equ:delta_w_e_0}, we can conclude 
	\begin{align}
		\Delta_w(\varepsilon)=g_r(\varepsilon)-g_l(\varepsilon)=(w_2(\varepsilon)-1)-(1-w_1(\varepsilon))\geq 0 
	\end{align}
	
	$\hfill\square$ 
\end{proof}

\section{Proof of Lemma \ref{thm:sup_fxy}}\label{sec:app_proof_thm:sup_fxy}
\begin{proof}
	
	From Lemma \ref{prp:sup_f_x_invers_1_d}, we know $w_1(\varepsilon_x)\leq x\leq w_2(\varepsilon_x)$ and  $w_1(\varepsilon_y)\leq y\leq w_2(\varepsilon_y)$. So we have  $w_1(\varepsilon_x)w_1(\varepsilon_x)\leq xy\leq w_2(\varepsilon_x)w_2(\varepsilon_y)$.
	According to Lemma \ref{equ:f_convex_minimum}, it suffices to show  $f(w_1(\varepsilon_x)w_1(\varepsilon_y))\leq f(w_2(\varepsilon_x)w_2(\varepsilon_y))$. By the definition of $f(x)$, we have
	\begin{align} 
		&f(w_2(\varepsilon_x)w_2(\varepsilon_y))-f(w_1(\varepsilon_x)w_1(\varepsilon_y)) \nonumber\\
		= & w_2(\varepsilon_x)w_2(\varepsilon_y)-\log (w_2(\varepsilon_x)w_2(\varepsilon_y))\nonumber\\
		&-(w_1(\varepsilon_x)w_1(\varepsilon_y)-\log (w_1(\varepsilon_x)w_1(\varepsilon_y)))  \nonumber\\
		= & w_2(\varepsilon_x)w_2(\varepsilon_y) -\log w_2(\varepsilon_x)-\log w_2(\varepsilon_y) \nonumber\\
		& - (w_1(\varepsilon_x)w_1(\varepsilon_y)- \log w_1(\varepsilon_x)-\log w_1(\varepsilon_y)) \nonumber\\
		= & w_2(\varepsilon_x)w_2(\varepsilon_y) -w_2(\varepsilon_x)+w_2(\varepsilon_x)-\log w_2(\varepsilon_x)\nonumber\\
		&-w_2(\varepsilon_y)+w_2(\varepsilon_y)-\log w_2(\varepsilon_y) \nonumber\\
		& - (w_1(\varepsilon_x)w_1(\varepsilon_y)-w_1(\varepsilon_x)+w_1(\varepsilon_x)- \log w_1(\varepsilon_x)\nonumber\\
		&-w_1(\varepsilon_y)+w_1(\varepsilon_y)-\log w_1(\varepsilon_y))\nonumber \\
		= & w_2(\varepsilon_x)w_2(\varepsilon_y)-w_2(\varepsilon_x)+\varepsilon_x-w_2(\varepsilon_y)+\varepsilon_y \nonumber\\ 
		& - (w_1(\varepsilon_x)w_1(\varepsilon_y)-w_1(\varepsilon_x)+\varepsilon_x-w_1(\varepsilon_y)+\varepsilon_y) \nonumber\\
		=& w_2(\varepsilon_x)w_2(\varepsilon_y)-w_2(\varepsilon_x)-w_2(\varepsilon_y)\nonumber\\
		&- (w_1(\varepsilon_x)w_1(\varepsilon_y)-w_1(\varepsilon_x)-w_1(\varepsilon_y)) \nonumber\\
		=& w_2(\varepsilon_x)w_2(\varepsilon_y)-w_2(\varepsilon_x)-w_2(\varepsilon_y)+1\nonumber\\
		&- (w_1(\varepsilon_x)w_1(\varepsilon_y)-w_1(\varepsilon_x)-w_1(\varepsilon_y)+1) \nonumber\\
		=& (w_2(\varepsilon_x)-1)(w_2(\varepsilon_y)-1)-(w_1(\varepsilon_x)-1)(w_1(\varepsilon_y)-1) \label{equ:fw2exw2ey_fw1exw1ey}
	\end{align}
	From Lemma \ref{thm:w2_farther_from_1_than_w1}, it is easy to know $w_2(\varepsilon_x)-1 \geq 1- w_1(\varepsilon_x)$ and $w_2(\varepsilon_y)-1 \geq 1- w_1(\varepsilon_y)$. Thus we can conclude 
	\begin{align}
		f(w_2(\varepsilon_x)w_2(\varepsilon_y))-f(w_1(\varepsilon_x)w_1(\varepsilon_y))\geq 0
	\end{align}
	
	$\hfill\square$ 
\end{proof}

\section{Proof of Theorem \ref{thm:triangle_n1_n2_n3}}\label{sec:appen_proof_thm:triangle_n1_n2_n3}
\begin{proof}
	
	For $X_2\sim \mathcal{N}(\bm{\mu}_2,\bm{\Sigma}_2)$, there exists an invertible matrix $B_2$ such that $X_2'=B_2^{-1}(X_2-\bm{\mu}_2)\sim \mathcal{N}(0,I)$ \cite{PRML}. Here $B_2=P_2D_2^{1/2}$, $P_2$ is an orthogonal matrix whose columns are the eigenvectors of $\bm{\Sigma}_2$, $D_2=diag(\lambda_{2,1},\dots,\lambda_{2,n})$ whose diagonal elements are the corresponding eigenvalues. 
	We define the following two invertible linear transformations $T$, $T^{-1}$ on random vectors.
	\begin{gather}
		X'=T(X)=B_2^{-1}(X-\bm{\mu}_2),\ 
		X=T^{-1}(X')=B_2 X'+\bm{\mu}_2
	\end{gather}
	Applying transformation $T$ on $X_1,X_2,X_3$, we can get three Gaussians.
	\begin{align}
		X_1'=T(X_1)\sim \mathcal{N}(\bm{\mu}'_1,\bm{\Sigma}'_1)\nonumber\\
		X_2'=T(X_2)\sim \mathcal{N}(0,I)\nonumber\\
		X_3'=T(X_3)\sim \mathcal{N}(\bm{\mu}'_3,\bm{\Sigma}'_3)\nonumber
	\end{align}
	According to Proposition \ref{thm:preserve_KL},  $T$ and $T^{-1}$ preserve KL divergence. Thus, we have
	\begin{align}
		&KL(\mathcal{N}(\bm{\mu}_1,\bm{\Sigma}_1)||\mathcal{N}(\bm{\mu}_2,\bm{\Sigma}_2))=KL(\mathcal{N}(\bm{\mu}'_1,\bm{\Sigma}'_1)||\mathcal{N}(0,I))\label{equ:kl_n1_n2_equals_n1'_n2'}\\
		&KL(\mathcal{N}(\bm{\mu}_2,\bm{\Sigma}_2)||\mathcal{N}(\bm{\mu}_3,\bm{\Sigma}_3))=KL(\mathcal{N}(0,I)||\mathcal{N}(\bm{\mu}'_3,\bm{\Sigma}'_3))\label{equ:kl_n2_n3_equals_n2'_n3'}\\
		&KL(\mathcal{N}(\bm{\mu}_1,\bm{\Sigma}_1)||\mathcal{N}(\bm{\mu}_3,\bm{\Sigma}_3))=KL(\mathcal{N}(\bm{\mu}'_1,\bm{\Sigma}'_1)||\mathcal{N}(\bm{\mu}'_3,\bm{\Sigma}'_3))\label{equ:kl_n1_n3_equals_n1'_n3'}
	\end{align}
	Combining the preconditions and Equations \eqref{equ:kl_n1_n2_equals_n1'_n2'}, \eqref{equ:kl_n2_n3_equals_n2'_n3'}, we can know 
	\begin{align}
		KL(\mathcal{N}(\bm{\mu}'_1,\bm{\Sigma}'_1)||\mathcal{N}(0,I))\leq \varepsilon_1,\ 
		KL(\mathcal{N}(0,I)||\mathcal{N}(\bm{\mu}'_2,\bm{\Sigma}'_2))\leq \varepsilon_2
	\end{align}
	Now we can apply Lemma \ref{thm:triangle_n1_n2_standard} on $\mathcal{N}(\bm{\mu}'_1,\bm{\Sigma}'_1)$, $\mathcal{N}(0,I)$ and $\mathcal{N}(\bm{\mu}'_3,\bm{\Sigma}'_3))$ and get the bound of $KL(\mathcal{N}(\bm{\mu}'_1,\bm{\Sigma}'_1)||\mathcal{N}(\bm{\mu}'_3,\bm{\Sigma}'_3))$. Finally, combining Equation \eqref{equ:kl_n1_n3_equals_n1'_n3'}, we can prove Theorem \ref{thm:triangle_n1_n2_n3}.
	
	$\hfill\square$ 
\end{proof}

\section{Proof of Theorem \ref{thm:triangle_n1_n2_n3_simplify}}\label{sec:proof_thm:triangle_n1_n2_n3_simplify}
\begin{proof}
	
	Suppose that $\varepsilon_1,\varepsilon_2$ are sufficiently small. According to the series expanding $W_0$ and $W_1$ (Section III.17 in \cite{2016PrincetonCompanion}), we have
	\begin{align}
		W_0(-e^{-(1+2\varepsilon)})=&-1+2\sqrt{\varepsilon}+O(\varepsilon)\\ W_{-1}(-e^{-(1+2\varepsilon)})=&-1-2\sqrt{\varepsilon}+O(\varepsilon)
	\end{align}
	So we can obtain 
	\begin{align}\label{equ:simplify_part1}
		& W_{-1}(-e^{-(1+2\varepsilon_1)})W_{-1}(-e^{-(1+2\varepsilon_2)})+ W_{-1}(-e^{-(1+2\varepsilon_1)})\nonumber\\
		&+ W_{-1}(-e^{-(1+2\varepsilon_2)})+1 \nonumber\\
		= & (W_{-1}(-e^{-(1+2\varepsilon_1)})+1)(W_{-1}(-e^{-(1+2\varepsilon_2)})+1) \nonumber\\
		= & (2\sqrt{\varepsilon_1}+O(\varepsilon_1))(2\sqrt{\varepsilon_2}+O(\varepsilon_2))\nonumber\\
		= & 4\sqrt{\varepsilon_1\varepsilon_2}+o(\varepsilon_1)+o(\varepsilon_2)
	\end{align}
	and
	
	\begin{align}\label{equ:simplify_part2}
		& -W_{-1}(-e^{-(1+2\varepsilon_2)})\left(\sqrt{2\varepsilon_1}+\sqrt{\dfrac{2\varepsilon_2}{-W_{0}(-e^{-(1+2\varepsilon_2)})}}\right)^2 \nonumber \\
		= &   (1+2\sqrt{\varepsilon_2}+O(\varepsilon_2))\left(\sqrt{2\varepsilon_1} + \sqrt{\dfrac{2\varepsilon_2}{1-2\sqrt{\varepsilon_2}+O(\varepsilon_2)}}   \right)^2 \nonumber \\
		\leq & (1+2\sqrt{\varepsilon_2}+O(\varepsilon_2))\left( 4\varepsilon_1 +  \dfrac{4\varepsilon_2}{1-2\sqrt{\varepsilon_2}+O(\varepsilon_2)} \right) \nonumber \\
		= & 4\varepsilon_1 + o(\varepsilon_1) + o(\varepsilon_2) + \dfrac{4\varepsilon_2 (1+2\sqrt{\varepsilon_2}+O(\varepsilon_2))}{1-2\sqrt{\varepsilon_2}+O(\varepsilon_2)} \nonumber \\
		= & 4\varepsilon_1 + o(\varepsilon_1) + o(\varepsilon_2)  +  4\varepsilon_2 +  \dfrac{4\varepsilon_2(4\sqrt{\varepsilon_2}+O(\varepsilon_2))}{1-2\sqrt{\varepsilon_2}+O(\varepsilon_2)} \nonumber \\
		= &  4\varepsilon_1 +  4\varepsilon_2 + o(\varepsilon_1) + o(\varepsilon_2)  + O(\varepsilon_2^{1.5})
	\end{align}	
	Using Equations \eqref{equ:simplify_part1} and \eqref{equ:simplify_part2}, we can rewrite the bound in Theorem \ref{thm:triangle_n1_n2_n3} as 
	\begin{align}
		&KL((\mathcal{N}(\bm{\mu}_1,\bm{\Sigma}_1)||\bm{\Sigma}(\bm{\mu}_3,\bm{\Sigma}_3))\nonumber \\
		<&3\varepsilon_1+3\varepsilon_2+2\sqrt{\varepsilon_1\varepsilon_2}+o(\varepsilon_1)+o(\varepsilon_2)
	\end{align}
	
	$\hfill\square$ 
\end{proof}


%

%
%
%
%

%

\ifCLASSOPTIONcaptionsoff
  \newpage
\fi

\end{document}